\newif\ifconfver
\confverfalse     
\confvertrue        

\newif\ifplainver  
\plainverfalse

\ifplainver
    \confverfalse   
\fi

\ifconfver
     \documentclass[10pt,twocolumn,twoside]{IEEEtran}
\else
    \ifplainver
        \documentclass[11pt]{article}
        \usepackage{fullpage}
    \else
        \documentclass[11pt,draftcls,onecolumn]{IEEEtran}
    \fi
\fi
\usepackage{calc,amsfonts,amssymb,amsmath,bm,url,color,graphicx,cite,shortcuts_OPT,enumitem,tcolorbox,booktabs}
\usepackage{psfrag,float,hyperref,bbm}
\hypersetup{
  colorlinks   = true, 
  urlcolor     = blue, 
  linkcolor    = blue, 
  citecolor    = blue,   
  hypertexnames = false
}
\usepackage{algorithm}
\usepackage{algorithmic}
\usepackage{amsthm}
\newcommand{\update}[1]{#1}

\newtheorem{Lemma}{Lemma}
\newtheorem{Remark}{Remark}
\newtheorem{Prop}{Proposition}
\newtheorem{Theorem}{Theorem}
\newtheorem{Def}{Definition}
\newtheorem{Corollary}{Corollary}

\newtheorem*{Theorem*}{Theorem}

\newtheorem{Exa}{Example}
\newtheorem{assumption}{H\!\!}

\usepackage{pgfplots}
\usetikzlibrary{arrows,shapes,calc,tikzmark,backgrounds,matrix,decorations.markings}
\usepgfplotslibrary{groupplots}
\usepgfplotslibrary{fillbetween}

\pgfplotsset{compat=1.3}
\usepackage{graphicx} 
\usepackage{diagbox}  
\usepackage{makecell}
 \usepackage{multirow}
\definecolor{lavander}{cmyk}{0,0.48,0,0}
\definecolor{violet}{cmyk}{0.79,0.88,0,0}
\definecolor{burntorange}{cmyk}{0,0.52,1,0}

\definecolor{asuorange}{rgb}{1,0.699,0.0625}
\definecolor{asured}{rgb}{0.598,0,0.199}
\definecolor{asuborder}{rgb}{0.953,0.484,0}
\definecolor{asugrey}{rgb}{0.309,0.332,0.340}
\definecolor{asublue}{rgb}{0,0.555,0.836}
\definecolor{asugold}{rgb}{1,0.777,0.008}

    \makeatletter
    \def\multilimits@{\bgroup
  \Let@
  \restore@math@cr
  \default@tag
 \baselineskip\fontdimen10 \scriptfont\tw@
 \advance\baselineskip\fontdimen12 \scriptfont\tw@
 \lineskip\thr@@\fontdimen8 \scriptfont\thr@@
 \lineskiplimit\lineskip
 \vbox\bgroup\ialign\bgroup\hfil$\m@th\scriptstyle{##}$\hfil\crcr}
    \def\Sb{_\multilimits@}
    \def\endSb{\crcr\egroup\egroup\egroup}
\makeatother

\newcommand*{\grayout}[1]{}
 
\allowdisplaybreaks

\makeatletter
\DeclareRobustCommand*\cal{\@fontswitch\relax\mathcal}
\makeatother
\title{Learning Graph Topology with Functional Priors via Bilevel Optimization}
\author{Chenyue Zhang, Shangyuan Liu, {Hoi-To Wai}, Anthony Man-Cho So\thanks{The authors are with the Department of SEEM, The Chinese University of Hong Kong, Shatin, Hong Kong SAR of China. Emails: \texttt{\{cyuezhang,shangyuanliu\}@link.cuhk.edu.hk}, \texttt{\{htwai,manchoso\}@se.cuhk.edu.hk}. }}
\date{}

\begin{document}
\maketitle

\begin{abstract}
Learning graph topology of complex networks is challenging due to limited data availability and imprecise data models. Different from prior works that focus on structural priors with explicit control on macroscopic properties such as sparsity, this paper proposes a novel \emph{functional prior} approach for graph topology learning. We postulate that complex networks are inherently optimized to perform a certain task (e.g., social networks specialize at optimizing a welfare function, biological networks are resilient towards node/edge deletion), which can be incorporated as a regularizer to assist in graph learning. Mathematically, we formulate a bilevel optimization problem where the lower-level problem solves the associated task on a candidate graph topology and the upper-level problem trades off between data fitting and task performance.
We design a two-timescale gradient descent (TTGD) algorithm and show that under verifiable conditions, it finds a stationary point to the bilevel graph learning problem with a sublinear convergence rate.
We provide theoretical insights on the graph topology learned from the functional priors and 
\update{show that the resulting regularizers subsume a broad class of graph filter regularizers, including polynomial graph regularizers as special cases.
}
We show via extensive experiments on synthetic and real datasets that the proposed formulation gives rise to reliable estimates of graph topology, even with insufficient data.
\end{abstract}




\section{Introduction}
Graph-based structures are widely used across data science to capture relationships among features and labels. Learning with graph structures plays a key role in many real-world applications such as life science, graph neural networks, and recommendation systems. However, natural graph topologies are not always available.
Consequently, the inference of unknown graph topology from nodal observations has become an important focus across fields of machine learning, signal processing, sociology, and biology \cite{newman2018networks}.

Such a problem, also known as graph topology learning, has been a longstanding challenge in data science and signal processing \cite{mateos2019connecting, dong2019learning, friedman2008sparse}. Earlier works have focused on identifying graph signal models, i.e., \emph{data generation models}, that relate the nodal observations to the graph topology and developing algorithms to efficiently infer the graph topology. For example, the works \cite{dong2016learning, kalofolias2016learn} proposed to learn the graph topology from smooth signals, the works \cite{friedman2008sparse, loh2012structure,rue2005gaussian, meinshausen2006high} utilized a Gaussian Markov Random Field (GMRF) to learn a sparse graph topology, the works \cite{shen2017kernel,wai2019joint} developed physics-inspired models for graph learning, and the work \cite{segarra2017network} utilized models that are inspired by graph signal processing \cite{sandryhaila2013discrete,shuman2013emerging}. 
As the networks of interest become more complex, recent developments in graph topology learning have switched to the rising issues of limited data availability and imprecise data models. On one hand, the number of nodal observations is often smaller than the total number of nodes. This type of data scarcity can arise from the high measurement costs \cite{chen2025data}, privacy concerns \cite{li2023private}, or difficulty in collecting complete observation data in large, complex systems \cite{kossinets2006effects}. As a result, the graph topology learning problem becomes vulnerable to random noise and often struggles to produce accurate inference \cite{yang2020network}. On the other hand, the interplay between nodal observations and underlying graph topology is often complicated, which cannot be captured by the simplified data generation models deployed in existing works.  

As a remedy, a common practice is to incorporate prior information in the learning process. Existing literature focuses on applying \emph{structural priors} developed from explicit features of real network structures. They are usually driven by empirical observations on common graph topologies. For this class of priors, a canonical conjecture is edge sparsity. In fact, a common observation across real networks is that \update{the graph topology admits a small number of edges and thus a sparse graph representation} \cite{jackson2008social, smith2011network}. Based on this observation, a series of works proposed to incorporate sparsity regularizers in different graph learning paradigms and achieved empirical successes \cite{friedman2008sparse, dong2016learning, kalofolias2016learn}. 
Recently, structural priors focusing on the clustering properties of graphs have been studied. Examples include $k$-component graphs \cite{chung1997spectral}, bipartite graphs \cite{zha2001bipartite}, and connected graphs \cite{lovasz2019graphs}. As demonstrated in \cite{egilmez2017graph, kumar2020unified, nie2016constrained}, carefully crafted regularizers based on the graph spectra can induce desired clustering characteristics in the learned graph. 
Note that, however, these priors are imposed directly on the graph topology itself in a macroscopic manner and do not account for the more intrinsic relationships between nodes. 

\begin{figure}[t]
    \centering
\includegraphics[width=0.96\linewidth]{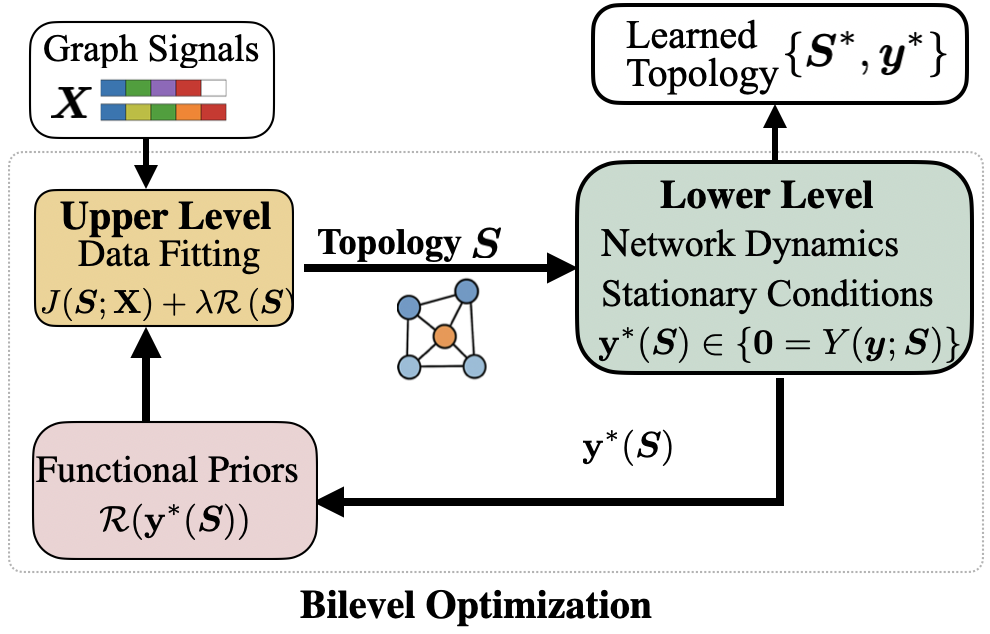}\vspace{-.2cm}
    \caption{Overview of the bilevel optimization framework for graph topology learning with functional priors. The upper-level problem minimizes a data fitting term $J(\GSO ; \mathbf{X})$ augmented with a task-driven regularizer $\mathcal{R}\left(y^*(\GSO)\right)$, where the lower-level problem computes the stationary state $y^*(\GSO)$ of a network dynamics model.}\vspace{-.4cm}
    \label{fig:flowchart}
\end{figure}

This paper initiates the study of a \emph{functional prior} approach to graph topology learning. We aim to rigorously capture the complex dynamics and relationships between nodes in the graph topology learning process. The key idea is a \emph{task-induced formation model} that treats the graph topology as a latent variable optimized to perform a specific task. Here, the task can be complex and nonlinear such as maximizing social welfare in social networks or enhancing resilience in biological networks. 
We notice that recent works in network science have demonstrated that network topology may play a crucial role in shaping the performance of the networked systems. 
For social networks, prior works have shown the importance of network topology in network game-induced tasks such as optimal targeting \cite{demange2017optimal} and optimal pricing \cite{Candogan}, suggesting that certain network structures may inherently favor specific applications.
For biological networks, the work \cite{gao2016universal} studied the relationship between network topology and resilience to network failure. Particularly, it is found that natural networks such as ecological and gene regulatory networks tend to have a heterogeneous {degree distribution}, which provides evidence for a resilient network topology.
In these examples, the graph topology of natural networked systems can be viewed as the outcome of a `self-optimization' process, which captures the complexities of underlying tasks they are solving.  

These findings have motivated the current paper to look beyond the direct {\it structural priors} approach for graph topology learning, and to develop a novel class of {\it functional priors} for efficient graph topology learning. 
In a nutshell, the functional prior framework encapsulates the purported complex dynamics of the network into a merit function that regularizes the graph learning objective, \update{leading to a \emph{biased} graph topology estimator}. 
We develop the latter merit function as a shaping function on the equilibria or stationary states of a suitable network dynamics model, which are then treated as intermediate variables that implicitly depend on the latent graph topology.
As we will demonstrate later, the advantage of doing so is that the merit function can be defined to impose a wide class of desired properties on the interpretable equilibria or stationary states, e.g., maximizing social welfare in social networks, or maximizing resilience to network failure in gene regulatory networks.
Mathematically speaking, the functional prior framework gives rise to a novel bilevel optimization problem for graph topology learning. Here, the lower-level problem pertains to the network dynamics model and the upper-level problem trades off between data fitting and desired task performance, as shown in Fig.~\ref{fig:flowchart}.

Computationally, recent literature on bilevel optimization has made significant progress in developing efficient algorithms for tackling such problems. For example, the work \cite{hong2023two} proposed a two-timescale stochastic approximation (TTSA) framework that utilizes a pair of step sizes for tackling the coupled KKT system, the work \cite{shen2023penalty} proposed a penalty method, and the work \cite{shaban2019truncated} proposed a truncated gradient unrolling based method.
Among others, we develop a two-timescale gradient descent (TTGD) algorithm inspired by TTSA \cite{hong2023two} to tackle the bilevel optimization problem.
The contributions of this paper are as follows:
\begin{itemize}[leftmargin=*]
\item We propose and formulate a framework for graph topology learning with functional priors (GLFP). The key feature of this framework is to use a task-inspired formation model that treats the unknown graph topology as the result of a self-optimization process and design the functional prior regularizers. The graph learning problem is then formulated as a bilevel optimization problem, where the upper-level problem is a regularized graph learning objective and the lower-level problem captures the result of the self-optimization process.

\item As applications of the GLFP framework, we study two specific cases: (i) a social welfare maximization prior for learning social network topology, and (ii) a resilience maximization prior for learning gene regulatory network topology. In both cases, we provide empirical evidence of the functional priors in real networks and design the induced graph topology learning problems with different data observation models that specify the GLFP framework.

\item To understand the optimal solutions that GLFP leads to, we illustrate how they can be approximated by ``high-order" polynomial regularizers. By analyzing the optimality conditions of the approximations, we provide theoretical insights into the graph topology learned from GLFP. Focusing on the two special cases mentioned above, we show that the drawn insights coincide with the topology features observed in real networks. Inspired by the approximations, we further suggest that our proposed functional prior regularizers may provide a more general form than the classical structural priors.


\item Finally, we design an efficient TTGD algorithm to tackle the GLFP problem. The TTGD algorithm is inspired by the TTSA framework \cite{hong2023two} and extends TTSA to \update{tackle variational inequalities in the lower-level problems}. To measure the convergence of TTGD, we adopt the squared norm of certain residual map as the stationarity measure. This {metric} is designed for constrained smooth nonconvex problems \cite{zhou2017unified}. Under some canonical assumptions, we show that TTGD finds a stationary point to the GLFP problem with a sublinear rate.
\end{itemize}
We remark that the idea of optimizing the graph topology for a task-specific application has been explored in several prior studies. For example, \update{the work \cite{kuhne2025optimizing} modified social network topology for polarization reduction, the work \cite{xu2024performance} modified network topology to enhance performance in action-coordination, and the work \cite{hu2022learning} focused on learning network structures tailored for node classification}. Their primary goal is to improve downstream performance in control or graph machine learning applications. In this paper, we focus on improving the accuracy of learning real networks.
Compared to our conference version \cite{zhang2025network}, this paper includes a general GLFP framework, strengthened theoretical results, and extended numerical experiments.

\vspace{.2cm}
\noindent
{\bf Notation.}
The notation we use in this paper is standard. We use $[m]$ to denote the set $\{1,2,\ldots,m\}$ for any positive integer $m$. Let the Euclidean space of all real vectors/matrices be equipped with the inner product $\langle \bm{X}, \bm{Y}\rangle := \tr(\bm{X}^{\top} \bm{Y})$ for any matrices $\bm{X}, \bm{Y}$ and denote the induced Frobenius norm by $\|\cdot\|_F$ (or $\Vert\cdot\Vert_2$ when the argument is a vector). For a vector $\bm{x}\in\RR^m$, we use $x_i$ to denote its $i$-th element and $\bm{x}_{-i}\in\RR^{m-1}$ to denote the vector without $x_i$. We use ${\sf Diag}(\bm{x})$ to denote the diagonal matrix whose diagonal elements are given by $\bm{x}$. For any matrix $\bm{X}\in \mathbb{R}^{m\times n}$, let $\|\bm{X}\|_2$ be the operator norm,  $\|\bm{X}\|_1:= \sum_{i,j}|X_{ij}|$, $\|\bm{X}\|_{\infty}:= \max_i\sum_{j}|X_{ij}|$, and $\diag(\bm{X})$ be the vector formed by the diagonal entries of a square matrix $\bm{X}$. 
Given a point $\bm{w}$ and a closed convex set $\mathcal{C}$, we use ${\sf Proj}_{\mathcal{C}}(\bm{w}) :=  \argmin_{\bm{v}\in\mathcal{C}}\|\bm{v}-\bm{w}\|_2$ to denote the projection of $\bm{w}$ onto $\mathcal{C}$. We use $\bm{1}$ (resp.\ $\bm{0}$) to denote an all-one vector (resp. all-zero vector) whose dimension will be clear from the context.  For a set $\mathcal{S}$, we let $\iota_\mathcal{S}(\cdot)
$ be its indicator function, i.e., $\iota_\mathcal{S}(x) = 0$ if $x \in {\cal S}$,  $\iota_\mathcal{S}(x) = \infty$ if $x \notin {\cal S}$.

\section{Preliminaries}\label{sec:prelim}
This section reviews the basic concepts and models for graph topology learning.
To fix notation, we first consider a networked system characterized by a directed graph ${\cal G} = (V,E)$ with node set $V = [N]$ and edge set $E \subseteq V \times V$, such that the ordered tuple $(i,j) \in E$ indicates an edge from $i$ to $j$. To describe ${\cal G}$, we denote its \emph{graph shifting operator} (GSO) by $\GSO \in \RR^{N \times N}$, such that $S_{ij} = 0$ if $(j,i) \notin E$. For instance, $\GSO$ can be the (weighted) adjacency matrix or the (weighted) Laplacian matrix, and it is not necessarily symmetric. Notice that in graph signal processing, a related concept is that of \emph{linear graph filter} $H(\cdot):\RR^{N\times N}\rightarrow\RR^{N\times N}$ with order $p$ \cite{sandryhaila2013discrete}, which is defined as a $p$-th order matrix polynomial of a GSO $\GSO$. Mathematically, it takes the form $H(\GSO) = \sum_{i= 0}^ph_i\GSO^i$. Here, we denote $\GSO^0 = \bm{I}$ by convention and $\{h_i\}_{i=0}^p$ are the graph filter coefficients.

\vspace{.1cm}
\noindent {\bf Graph Topology Learning.} The problem pertains to inferring the unknown graph ${\cal G}$, or equivalently, its GSO matrix $\GSO$, from a set of \emph{graph signals} defined on the node set $V$. The graph signals are denoted by ${\bm x}^{(1)}, \ldots, {\bm x}^{(M)}$ such that for each $m \in [M]$, the graph signal ${\bm x}^{(m)} \in \RR^N$ describes the nodes' states. Throughout, we set ${\bm X} = ( {\bm x}^{(1)}, \ldots, {\bm x}^{(M)} )$ to simplify notation. 
As in statistical learning, the general formulation for graph topology learning can be expressed as
\begin{equation} \label{eq:graph_learn_prob}
    \min_{ \GSO \in {\cal S} }~J( \GSO; {\bm X} ) + \lambda {\cal R}( \GSO ),
\end{equation}
where $J( \GSO; {\bm X} )$ is the \emph{data fitting} term, ${\cal S} \subseteq \RR^{ N \times N }$ refers to the feasible set of GSOs, ${\cal R}( \GSO )$ describes the regularizer pertaining to the prior knowledge on the unknown graph, and $\lambda > 0$ is a regularization parameter.

\vspace{.1cm}
\noindent {\bf Generation Model.}
{The data fitting term $J(\GSO; {\bm X})$ in \eqref{eq:graph_learn_prob} pertains to the \emph{generation model} for graph signal observations. In scenarios where ${\bm X}$ are low pass or smooth graph signals \cite{ramakrishna2020user}, a popular strategy is to fit an adjacency matrix $\GSO$ such that ${\bm X}$ is `smooth' with respect to (w.r.t.)~the latter, i.e.,
\begin{equation} \label{eq:smooth-GL}
J_{\sf smo}( \GSO; {\bm X} ) = {\rm Tr}( \GSO^\top {\bm D} ) ,~D_{ij} = {\textstyle \frac{1}{2M}} \| {\bm x}_i^{\rm row} - {\bm x}_j^{\rm row} \|^2_2 
\end{equation}
for all $i,j \in [N]$. Observe that $J_{\sf smo}( \GSO; {\bm X} )$ can be expressed as $\frac{1}{2M}\sum_{m=1}^M \sum_{i,j} S_{ij} | {x}_i^{(m)} - x_j^{(m)} |^2$ and measures the average Dirichlet energy of the graph signals w.r.t.~a proposed ${\GSO}$.\footnote{\update{The formulation~\eqref{eq:smooth-GL} is commonly used to infer a symmetric GSO matrix $\GSO$. We do not impose such a constraint in our discussions to avoid restricting the design of our functional prior; see Section \ref{subsec:ng}.}} The set of feasible adjacency matrices can be chosen as
\begin{equation} \label{eq:Sng}
{\cal S}_{\sf ng} = \{ \GSO \in \RR^{N \times N} : \GSO \geq {\bm 0},~{\rm diag}( \GSO ) = {\bm 0}, ~\GSO {\bf 1} = c {\bf 1} \},
\end{equation}
where $c>0$ is enforced to avoid the trivial solution.
Along the same line, alternative models include factor analysis \cite{dong2016learning}, stationary graph signals \cite{segarra2017network}, etc. They involve similar forms for the data fitting term $J( \GSO; {\bm X} )$ and the constraint set ${\cal S}$.}

In scenarios pertaining to network dynamics systems such as gene regulatory networks, the input ${\bm P} \in \mathbb{R}^{N \times M}$\update{, whose $m$-th column records the exogenous intervention or external input applied in experiment $m$,} is available alongside graph signals ${\bm X}$ of the steady-state responses. The following data-fitting loss is proposed in \cite{tjarnberg2013optimal, hillerton2022fast}:
\begin{equation} \label{eq:data_ridge}
J_{\sf per} (\GSO; \boldsymbol{X} ) = \|\GSO \boldsymbol{X} + \boldsymbol{P}\|_F^2.
\end{equation}
\update{This loss encourages each observed steady-state response to satisfy the linearized equilibrium relation $\GSO{\bm x}^{(m)} \approx -{\bm p}^{(m)}$.}
We adopt the set of feasible adjacency matrices
\begin{equation}\label{eq:Snd}
    \mathcal{S}_{\sf nd}  = \{ \GSO \in \RR^{N \times N} : \GSO \geq {\bm 0}, {\rm diag}( \GSO ) = {\bm 0},
     \bm{1}^\top\GSO\bm{1} = a
    \},
\end{equation}
where $a>0$ is enforced to avoid the trivial solution.

We remark that there are alternative models to \eqref{eq:smooth-GL} and \eqref{eq:data_ridge}, where they depend on the types of available data; see \cite{mateos2019connecting, dong2019learning}.


\vspace{.1cm}
\noindent {\bf Formation Model.}
The regularizer ${\cal R}( \GSO )$  in \eqref{eq:graph_learn_prob} pertains to the \emph{prior} information on the graph topology. It is especially important when the number of graph signal observations is insufficient, i.e., $N \gg M$. From a high level, the prior information is related to a \emph{formation model} which characterizes how the graph ${\cal G}$ is generated. In the existing literature, common choices entail \emph{explicit} and \emph{macroscopic} control over the graph topology. For example, the sparsity prior is relevant to the multi-variate exponential distribution. It is shown in \cite{egilmez2017graph} that this prior distribution can induce the regularizer 
\begin{equation} \label{eq:RS_l1}
    {\cal R}_{\ell_1}(\GSO) = \| \GSO \|_1 
\end{equation}
in the maximum a posteriori (MAP) estimate.\footnote{It is shown in \cite{kalofolias2016learn} that one can equivalently use the squared-Frobenius regularizer ${\cal R}(\GSO)=\|\GSO\|_F^2$.} Another example is the prior that the graph is formed with a modular structure which has $K$ densely connected components \cite{nie2016constrained, kumar2020unified}. This formation model is related to the regularizer
\begin{equation} \label{eq:RS_spect}
\textstyle {\cal R}_{\sf mod}( \GSO ) = \sum_{i=K+1}^N \sigma_i( \GSO ),
\end{equation}
where $\sigma_i(\cdot)$ represents the $i$-th largest singular value of $\GSO$.

\section{Problem Formulation}\label{sec:prob}
This paper proposes a graph learning framework that departs from the {structure-inspired} designs for the regularizer ${\cal R}( \GSO )$, e.g., \eqref{eq:RS_l1}, \eqref{eq:RS_spect}. Instead, we develop \emph{task-inspired} regularizers that \emph{implicitly} regulate the graph topology through the latter's influence on tasks performed by the network. 
At a high level, our framework is developed through a new \emph{formation model} of the graph topology. Inspired by the hypothesis that real-world networks are formed through a ``natural selection'' process, we are motivated by the overarching conjecture:
\begin{center}
\emph{The graph topologies of real-world networks are self-optimized according to a task-specific merit function.}
\end{center}
For example, a possible formation model for socio-economic networks is that the graph topology leads to the highest payoff at the Nash Equilibrium (NE) in a network game \cite{chen2022impact}; for gene regulatory networks, the graph topology leads to resilient states in the stationary states of a bio-physical system \cite{li2004yeast,ciliberti2007robustness}. 

Overall, consider a general merit function design that is defined via the stationarity condition for network states
\[
{\bm y}^\star( \GSO ) \in {\cal Y}^\star( \GSO ) := \left\{ {\bm y} \in {\cal Y} :  {\sf Y} ( {\bm y} ; \GSO ) = {\bm 0}\right\} ,
\]
where ${\cal Y} \subseteq \RR^N$ is the set of feasible states such that ${\sf Y}( \cdot; \GSO ) \in \RR^N$ is a stationary condition map. For example, in network games, it is related to the fixed point of best response dynamics. Subsequently, the merit function is given by the implicit function 
\beq \label{eq:merit}
\textstyle {\cal R}( \GSO ) := \min_{ {\bm y} \in {\cal Y}^\star(\GSO) } \hat{\cal R} ( {\bm y} ) + ( {\tilde{\beta}} / {2} ) \| \GSO \|_F^2,
\eeq 
such that $\hat{\cal R}: {\cal Y} \to \RR$ quantifies the merit of the network stationary states and the Frobenius norm term is added to stabilize the optimization solution with $\tilde{\beta}>0$. 
For example, $\hat{\cal R}( {\bm y} )$ is the negated sum of payoffs in a network game. Together, this leads to the \emph{bilevel program} 
\begin{tcolorbox}[boxsep=1pt,left=1pt,right=1pt,top=-4pt,bottom=1pt]
\begin{align}
\displaystyle \min_{ \GSO\in \mathcal{S}, {\bm y} \in {\cal Y} } & ~~\Phi( \GSO; {\bm y} ) := J( \GSO; {\bm X} ) + \lambda \hat{{\cal R}}( \bm{y} ) + ( {\lambda \tilde{\beta}} / {2} ) \| \GSO \|_F^2 \notag \\
\text{s.t.} & ~~\textstyle \bm{y} \in {\cal Y}^\star( \GSO ), \tag{GLFP} \label{eq:GL_bilevel:UNI}
\end{align} 
\end{tcolorbox}
\noindent which combines the data fitting term from the \emph{generation model} and the regularizer term from the \emph{formation model}. 

We shall describe regularizers given in the form of \eqref{eq:merit} as \emph{functional regularizers},  which depart from the structural regularizers considered in the prior works. 
We next discuss two special cases of the formation model in \eqref{eq:GL_bilevel:UNI}.

\subsection{Case Study 1: Network Games} \label{subsec:ng}
We consider a formation model driven by benchmarking the total welfare resulting from the Nash Equilibrium (NE) of a network game \cite{Candogan}. 
The game setting is related to learning socio-economic networks, where agents are assumed to be rational with actions shaped by a certain utility function.
{The lower-level state ${\bm y}$ collects the agents' equilibrium actions induced by a candidate graph topology, so the functional prior favors graphs with high aggregate welfare at equilibrium.}

In this setting, each node represents an agent and the edge weights represent the strength of trust between pairs of neighboring agents. Each agent is endowed with a payoff function $U_i(\cdot)$ that depends on the actions of his/her neighbors and himself/herself. For each agent $i$, s/he finds an action $y_i^\star$ 
\begin{align} \label{eq:game_payoff}
    y_i^\star & = \argmax_{y_i \in \mathcal{Y}_i} ~U_i( y_i , {\bm y}_{-i} ; \GSO ),
\end{align}
where ${\bm y}_{-i}$ denotes the vector ${\bm y} = ( y_1, \ldots, y_N )$ with the $i$-th agent's action removed and ${\cal Y}_i \subseteq \RR$ is the set of admissible actions. The agents are non-cooperative and aim to maximize their own payoffs \eqref{eq:game_payoff}. To this end, the NE describes the set of actions where no agent shall change his/her action. Mathematically, assuming that the NE exists and is unique, it is defined through the best response map
\begin{equation*}
\update{{\sf T}_i( {\bm y}; \GSO ) :=  \argmax_{ y_i \in \mathcal{Y}_i} U_i( y_i , {\bm y}_{-i} ; \GSO ),~\forall~i \in V.}
\end{equation*}
{Collecting the component maps gives the vector operator}
\begin{equation*}
\update{{\sf T}( {\bm y}; \GSO ) := \big( {\sf T}_1( {\bm y}; \GSO ), \ldots, {\sf T}_N( {\bm y}; \GSO ) \big)^\top.}
\end{equation*}
In particular, ${\bm y}^{\sf NE} (\GSO) = (y_1^{\sf NE} (\GSO), \ldots, y_N^{\sf NE} (\GSO) )$ is said to be an NE if   
\[
y_i^{\sf NE} (\GSO) = {\sf T}_i( {\bm y}^{\sf NE}(\GSO) ; \GSO ),~\forall~i \in V.
\]
To simplify notation, we denote ${\bm y}^{\sf NE} (\GSO)$ as the fixed point to the equation ${\bm y} = {\sf T}( {\bm y} ; \GSO )$. This yields a special case of \eqref{eq:GL_bilevel:UNI} with
\begin{equation*}
\update{{\sf Y}( {\bm y} ; \GSO ) := {\bm y} - {\sf T}( {\bm y} ; \GSO ).}
\end{equation*}

We consider two models for the payoff function. The first model is a \emph{generalized linear-quadratic} (LQ) game \cite{cai2024optimal, ballester2006s,Candogan,bramoulle2014strategic} whose payoff is given by 
\beq \textstyle \label{eq:lq_game}
U_i^{lq}( y_i , {\bm y}_{-i} ; \GSO ) = - \frac{y_i^2}{2} + y_i \left( \sum_{j=1}^N S_{ij} f(y_j) + b_i \right)
\eeq 
over $\mathcal{Y}_i = [0,\infty)$, where \(b_i\) is the marginal benefit and \(f(\cdot)\) is an interaction function that modulates the effects of neighbors' action on agent $i$. 
\update{The generalized LQ game is a canonical model for strategic interaction on a network \cite{jackson2015games}. In this model, each agent $i$ chooses a nonnegative action level $y_i$, which can be interpreted as effort or participation intensity. The first term $-y_i^2/2$ represents a private cost of exerting action, while the term $\sum_{j=1}^{N} S_{ij} f(y_j)$ captures how agent $i$'s incentive is shaped by the actions of its neighbors on the graph. The interaction function $f(\cdot)$ acts on individual neighbors' actions before aggregation, so the generalized LQ game captures individually mediated spillovers.}
See \cite{cai2024optimal} for conditions on the existence and uniqueness of NE for the above model.

The second model is a \emph{race and tournament} (RT) game \cite{allouch2015private, parise2019variational, belhaj2014competing} whose payoff is 
\beq \textstyle \label{eq:rt_game}
U_i^{rt}( y_i , {\bm y}_{-i} ; \GSO ) = - \frac{y_i^2}{2} + y_i \left( g\left(\sum_{j=1}^N S_{ij} y_j\right) + b_i \right)
\eeq 
over ${\cal Y}_i = [\underline{b}_i, a_i]$, where $0 < \underline{b}_i < a_i$ and \(g(\cdot)\) is a nonlinear function that captures the interactions between agent $i$ and his/her neighbors. 
\update{The RT game models competitive interactions with bounded actions $y_i \in [\underline{b}_i, a_i]$. Note that the nonlinear function $g(\cdot)$ acts on the aggregated neighbor activity $\sum_{j=1}^{N} S_{ij} y_j$, so the return to effort depends nonlinearly on the overall local competitive environment. Thus, the RT game has a different interaction structure from the generalized LQ game. The latter captures individually mediated spillovers, whereas the former is suited for settings where incentives depend on the aggregate level of peer activity.}
See \cite{parise2019variational} for conditions on the existence and uniqueness of NE for the above model.

In general, the graph topology $\GSO$ directly affects the NE ${\bm y}^{\sf NE}(\GSO)$ \cite{cai2024optimal,Candogan}. Herein, it is possible to use the \emph{total welfare} \cite{demange2017optimal} to measure the performance of a socio-economic system, defined as 
\[
{\sf Wel}( \GSO ) := {\bf 1}^\top {\bm y}^{\sf NE}( \GSO ) .
\]
It can be viewed as the overall economic gain. Notice that ${\sf Wel}( \GSO )$ depends on $\GSO$ via the NE of the network game.
We formally state the formation-model conjecture --- 
\begin{center}
    \emph{the graph topology of socio-economic and human-made systems shall maximize the total welfare ${\sf Wel}( \GSO )$}
\end{center} 
and thus the \emph{merit function} with ${\cal R}(\GSO) = - {\sf Wel}(\GSO) + \frac{\tilde{\beta}}{2}\|\bm{S}\|_{\sf F}^2$. 

\begin{table}[t]
\centering
\includegraphics[width=0.47\textwidth]{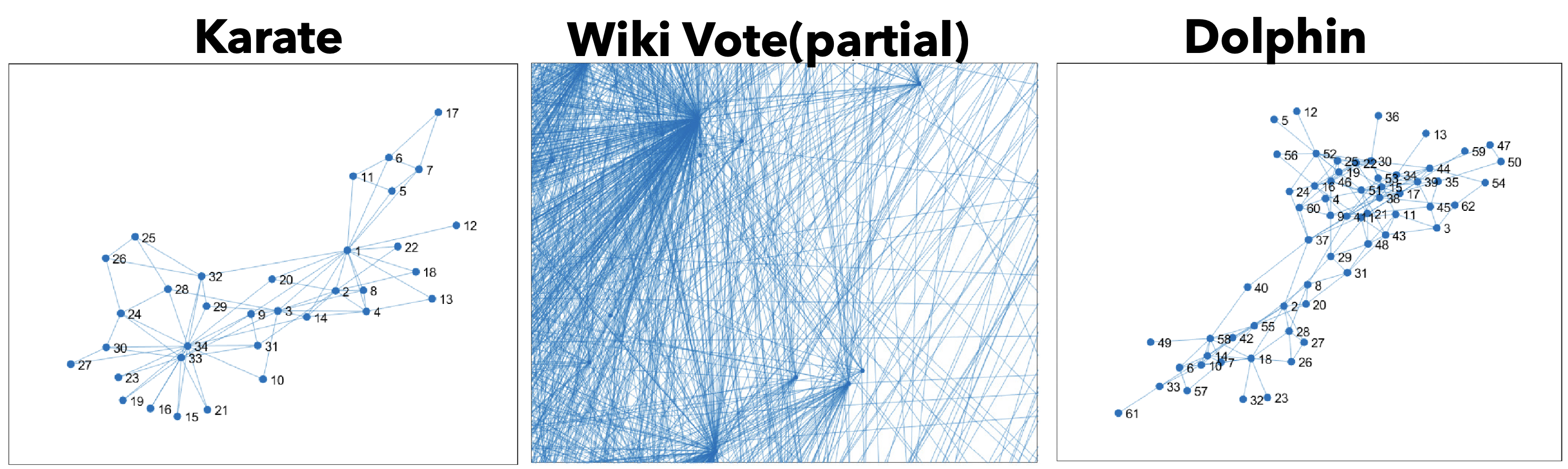}\vspace{.2cm}
\begin{tabular}{l|l|l|l} 
\toprule 
Rewiring & $10\%$ & $30\%$ & $50\%$ \\
\midrule
{\tt Karate} & $\mathbf{94.06\%}$ & $\mathbf{84.72\%}$ & $\mathbf{78.30\%}$ \\
{\tt WikiVote} & $\mathbf{96.28\%}$ & $\mathbf{90.32\%}$ & $\mathbf{86.11\%}$ \\
{\tt LesMiserables} & $\mathbf{93.79\%}$ & $\mathbf{84.43\%}$ & $\mathbf{78.56\%}$ \\
{\tt JazzMusician} & $\mathbf{97.43\%}$  & $\mathbf{93.47\%}$ & $\mathbf{90.77\%}$ \\
{\tt SchoolNet} & $\mathbf{97.28\%}$ & $\mathbf{93.47\%}$ & $\mathbf{91.20\%}$ \\
{\tt MalawiVillage} & $\mathbf{96.72\%}$ & $\mathbf{91.82\%}$ & $\mathbf{88.61\%}$ \\
\midrule
{\tt Dolphins} & $98.15\%$ & $95.13\%$ & $93.08\%$ \\
{\tt Ant}&$98.23\%$ & $95.90\%$ & $94.61\%$ \\
{\tt Weaver} & $99.43\%$ & $98.49\%$ & $97.80\%$ \\
\bottomrule
\end{tabular}
\vspace{.2cm}
\caption{Top: visualizing the topology of several tested networks. Bottom: impact of random rewiring on the welfare ratio after rewiring; cf.~\eqref{eq:welfare_ratio}. 
}\label{table:ReaNetWkRewir}
\end{table}

Directly verifying the above conjecture is impossible as the ground truth formation model is unknown. We consider an empirical approach through testing on real-world networks. It relies on the insight that if the conjecture holds, then the corresponding $\GSO$ would be a \emph{local maximum} to ${\sf Wel}(\GSO)$ such that any perturbation to $\GSO$ could lead to a significant drop in ${\sf Wel}(\GSO)$. We verify the above phenomenon by applying random rewiring to the graph topology of real-world networks and comparing the reduction in ${\sf Wel}(\GSO)$. We fix {${\bm b} = {\bm 1}$} and consider the payoff function in \eqref{eq:lq_game} with linear interaction function $f(x) = x$. We compute the perturbed welfare ratio
\beq \label{eq:welfare_ratio}
\mathbb{E} \left[ { ( {\sf Wel}( \GSO_{\sf pt} ) - {\bf 1}^\top {\bm b} ) } / { ( {\sf Wel}( \GSO_{\sf og} ) - {\bf 1}^\top {\bm b} ) } \right],
\eeq 
where $\GSO_{\sf og}$ and $\GSO_{\sf pt}$ are respectively the unweighted adjacency matrices of the original and randomly rewired graphs.
Table \ref{table:ReaNetWkRewir} shows the welfare ratio on the real-world networks. Observe that human-made networks (e.g., {\tt Karate}) suffer a significant drop in the welfare ratio after rewiring, while other networks (e.g., {\tt Dolphins}) are less sensitive to rewiring. This observation supports our conjecture.
We remark that related observations can also be found in \cite{demange2017optimal, sanhedrai2022reviving, meena2023emergent} on how real networks show traits of self-optimization.

Finally, we consider the smoothness based generation model \eqref{eq:smooth-GL} with the above formation model to yield
\begin{tcolorbox}[boxsep=1pt,left=1pt,right=1pt,top=-4pt,bottom=1pt]
\begin{align}
\displaystyle \min_{ \GSO, {\bm y}\in{\cal Y} } & ~\Phi( \GSO; {\bm y} ) := {\rm Tr}( \GSO^\top {\bm D} ) + \beta \| \GSO \|_{\rm F}^2 - \lambda {\bf 1}^\top {\bm y} \notag \\
    \text{s.t.} & ~{\bm y} - {\sf T}( {\bm y}; \GSO ) = {\bm 0},~ \GSO \in {\cal S}_{\sf ng}, \label{eq:glfp_ng} \tag{GL-NG}
\end{align} 
\end{tcolorbox}
\noindent where ${\cal Y} = \{\bm{y}\in \RR^N : y_i\in{\cal Y}_i, ~\forall~ i\in[N]\}$, ${\cal S}_{\sf ng}$ was defined in \eqref{eq:Sng}, and 
$\beta = \frac{\lambda\tilde{\beta}}{2} > 0, \lambda > 0$ are regularization parameters.\vspace{-.2cm}

\subsection{Case Study 2: Network Dynamics} \label{subsec:nd}
We consider a formation model driven by the resilience property of gene regulatory network (GRN) dynamics. The GRN dynamics model the time varying states, quantified by their expression levels, of genes in a cell. 
{This case study targets GRNs, where the candidate graph topology determines how genes influence one another through a nonlinear stationary dynamics model. The lower-level variable therefore represents a steady-state gene-expression profile, and the functional prior rewards graph topologies that sustain large nonzero stationary responses under perturbations.}

In a GRN, each node represents a gene and the edge weights represent the influence strengths between genes (possibly directed). The state of each gene is represented by its expression level, which can be affected by the expression levels of its neighboring genes.
To capture the changes in the gene's expression level, a widely used model is the Michaelis-Menten (MM) system of differential equations \cite{alon2019introduction,karlebach2008modelling}

\beq  \label{eq:gene_LL} \textstyle
\frac{ {\rm d} }{ {\rm d} \, t } {y}_i(t) = - y_i(t) + \sum_{j=1}^N S_{ij} \frac{y_j(t)^2}{y_j(t)^2 + 1},~\forall~i \in V,
\eeq
where $y_i(t)$ denotes the $i$-th gene's expression level at time $t$. 
\update{The linear decay term $-y_i(t)$ models the natural degradation of gene products, while the nonlinear regulatory input from neighboring genes takes the form of a saturating response function. In particular, the term $y_j(t)^2/(y_j(t)^2+1)$ possesses a Hill-function-type nonlinearity: When the expression level of gene $j$ is low (resp., high), its regulatory effect is weak (resp., saturates). Such saturation is a standard feature in biochemical regulation models and is one reason why MM systems are widely used in systems biology \cite{alon2019introduction,karlebach2008modelling}.}
Under \eqref{eq:gene_LL}, a stationary state of the cell $\bm{y}^\star$ defines a collection of gene states that do not change, i.e.,
\beq \label{eq:gene_LL_stat}
\textstyle {\cal Y}^\star( \GSO ) := \left\{ {\bm y} \in \RR^N : y_i = \sum_{j=1}^N S_{ij} \frac{y_j^2}{ y_j^2 + 1},~\forall~i \in V \right\}.
\eeq

Note that in general $|{\cal Y}^\star( \GSO )| > 1$. 
For example, it must hold that ${\bm 0} \in {\cal Y}^\star( \GSO )$, while there may be other non-zero solutions satisfying \eqref{eq:gene_LL_stat}. 


In the context of GRN dynamics, the resilience property refers to the ability of a GRN to maintain non-zero stationary states after the system is perturbed due to, e.g., gene knockout or suppression of regulatory interactions. We are interested in the notion of \emph{universal resilience} that is governed by the GRN dynamics equations and graph topology, independent of the specific types of perturbations. To this end, the work \cite{gao2016universal} proposed a metric based on the heterogeneity of graph topology and showed that the latter is correlated with the resilience of GRN against edge deletion perturbations.

Departing from \cite{gao2016universal}, we propose an alternative resilience metric by directly measuring the number of non-zero genes in the stationary state. In particular, we consider 
\begin{align}
    \Res(\bm{S}) & \textstyle = \max_{\bm{y} \in {\cal Y}^\star(\bm{S})}\sum_{i=1}^N \frac{ 1 - e^{-  \sigma y_i} }{{ 1 + e^{- \sigma y_i} }},
    \label{eq:resilience-smooth}
\end{align}
which approximates the number of non-zero genes in the stationary states resulting from $\bm{S}$. Here, $\sigma > 0$ is to control the approximation error and a larger $\sigma$ leads to a closer approximation. On one hand, we have $\Res(\GSO) = 0$ if and only if the GRN only has null stationary states. On the other hand, a GRN with higher $\Res(\GSO)$ has more genes that are active, which intuitively implies a more resilient GRN.
Eq.~\eqref{eq:resilience-smooth} enables us to formally state our conjecture on the formation model for GRN graph topology $\GSO$ --- 
\begin{center}
\emph{cells in nature shall maximize the resilience property approximated by $\Res(\GSO)$},
\end{center}
thus the \emph{merit function} with ${\cal R}(\bm{S}) = -\Res(\bm{S}) + ({\tilde{\beta}} / {2}) \|\bm{S}\|_{\sf F}^2$.
\update{Note that
\[
\frac{1-e^{-\sigma y_i}}{1+e^{-\sigma y_i}}
=
2\frac{1}{1+e^{-\sigma y_i}}-1.
\]
Hence, the sigmoid-based objective used below is equivalent to \eqref{eq:resilience-smooth} up to an additive constant and a rescaling of the regularization parameter.}

\begin{figure}[t]
\begin{center}
\resizebox{0.99\linewidth}{!}{\definecolor{mycolor1}{rgb}{0.00000,0.44700,0.74100}%
\definecolor{mycolor2}{rgb}{0.85000,0.32500,0.09800}%
\definecolor{mycolor3}{rgb}{0.92900,0.69400,0.12500}%
\definecolor{mycolor4}{rgb}{0.49400,0.18400,0.55600}%
\definecolor{mycolor5}{rgb}{0.46600,0.67400,0.18800}%
\definecolor{mycolor6}{rgb}{0.30100,0.74500,0.93300}%
\definecolor{mycolor7}{rgb}{0.63500,0.07800,0.18400}%

\begin{tikzpicture}

\begin{groupplot}[group style={group name=myplot,group size=2 by 1}]
\nextgroupplot
[%
width=6cm,
height=5cm,
xmin=0,
xmax=0.7,
xlabel style={font=\color{white!15!black}},
xlabel={\large $p_{\sf perturb}$},
ymin=0,
ymax=5,
ylabel style={font=\color{white!15!black}},
ylabel={\large $\Res(\hat{\GSO})/2$},
axis background/.style={fill=white},
axis x line*=bottom,
axis y line*=left,
xmajorgrids,
ymajorgrids,
legend style={legend cell align=left, align=left, draw=white!15!black}
]

\addplot[area legend, draw=none, fill=mycolor1, fill opacity=0.1]
table[row sep=crcr] {%
x	y\\
0	4.332683904796\\
0.05	4.593751224995\\
0.1	4.23510992579\\
0.2	3.91105516206\\
0.3	2.711495842812\\
0.4	2.164415189907\\
0.5	0.047815928805\\
0.6	0.081353446091\\
0.7	0.015052182657\\
0.7	0.011685963822\\
0.6	0\\
0.5	0\\
0.4	0\\
0.3	0\\
0.2	1.273247048708\\
0.1	2.291642254516\\
0.05	2.512636787167\\
0	4.305785533701\\
}--cycle;

\addplot [color=mycolor1, line width=3pt, mark=o, mark options={solid, mycolor1}]
  table[row sep=crcr]{%
0	4.319234719248\\
0.05	3.553194006081\\
0.1	3.263376090153\\
0.2	2.592151105384\\
0.3	1.3844412006\\
0.4	0.995871936555\\
0.5	0.030383833949\\
0.6	0.038179554416\\
0.7	0.01336907324\\
};

\addplot[area legend, draw=none, fill=mycolor2, fill opacity=0.1]
table[row sep=crcr] {%
x	y\\
0	4.941187733477\\
0.05	0.941401590235\\
0.1	0.268128477017\\
0.2	0.085301870772\\
0.3	0.043444241047\\
0.4	0.028522880404\\
0.5	0.020562120791\\
0.6	0.015794768884\\
0.7	0.012842780963\\
0.7	0.012182457618\\
0.6	0.01469790629\\
0.5	0.018889917586\\
0.4	0.026356440715\\
0.3	0.036100683002\\
0.2	0.061209515877\\
0.1	0.158063365159\\
0.05	0.297794623923\\
0	0.882333536669\\
}--cycle;

\addplot [color=mycolor2, line width=3pt, mark=o, mark options={solid, mycolor2}]
  table[row sep=crcr]{%
0	2.911760635073\\
0.05	0.619598107079\\
0.1	0.213095921088\\
0.2	0.073255693325\\
0.3	0.039772462025\\
0.4	0.02743966056\\
0.5	0.019726019188\\
0.6	0.015246337587\\
0.7	0.012512619291\\
};

\addplot[area legend, draw=none, fill=mycolor3, fill opacity=0.1]
table[row sep=crcr] {%
x	y\\
0	1.685248076605\\
0.05	0.62881953758\\
0.1	0.304321468655\\
0.2	0.157378370692\\
0.3	0.053862913906\\
0.4	0.036599681327\\
0.5	0.027740659558\\
0.6	0.01832850761\\
0.7	0.014384292117\\
0.7	0.012593754769\\
0.6	0.01459581395\\
0.5	0.018776195647\\
0.4	0.023554004518\\
0.3	0.033886909499\\
0.2	0.042886511289\\
0.1	0.083768515913\\
0.05	0.063536632636\\
0	0\\
}--cycle;

\addplot [color=mycolor3, line width=3pt, mark=o, mark options={solid, mycolor3}]
  table[row sep=crcr]{%
0	0.772121748324\\
0.05	0.346178085108\\
0.1	0.194044992284\\
0.2	0.10013244099\\
0.3	0.043874911702\\
0.4	0.030076842923\\
0.5	0.023258427602\\
0.6	0.01646216078\\
0.7	0.013489023443\\
};

\nextgroupplot[%
width=6cm,
height=5cm,
xmin=0,
xmax=0.7,
xlabel style={font=\color{white!15!black}},
xlabel={\large $p_{\sf perturb}$},
ymin=0,
ymax=1,
ylabel style={font=\color{white!15!black}},
ylabel={Survival rate},
axis background/.style={fill=white},
axis x line*=bottom,
axis y line*=left,
xmajorgrids,
ymajorgrids,
legend style={legend cell align=left, align=left, draw=white!15!black}
]
\addplot [color=mycolor1, line width=3pt, mark=diamond, mark options={solid, mycolor1}]
  table[row sep=crcr]{%
0	1\\
0.05	1\\
0.1	1\\
0.2	1\\
0.3	0.98\\
0.4	0.866\\
0.5	0.52\\
0.6	0.182\\
0.7	0.014\\
};
\addlegendentry{E. coli}

\addplot [color=mycolor2, line width=3pt, mark=diamond, mark options={solid, mycolor2}]
  table[row sep=crcr]{%
0	1\\
0.05	1\\
0.1	1\\
0.2	1\\
0.3	0.834\\
0.4	0\\
0.5	0\\
0.6	0\\
0.7	0\\
};
\addlegendentry{ER}

\addplot [color=mycolor3, line width=3pt, mark=diamond, mark options={solid, mycolor3}]
  table[row sep=crcr]{%
0	1\\
0.05	1\\
0.1	1\\
0.2	1\\
0.3	0.876\\
0.4	0.188\\
0.5	0.002\\
0.6	0\\
0.7	0\\
};
\addlegendentry{PA}

\end{groupplot}
\end{tikzpicture}}\vspace{-.2cm}
\end{center}
  \caption{Resilience in gene regulatory networks in E.~coli \cite{gao2016universal}. The x-axis $p_{\sf perturb}$ is the proportion of link loss. \update{Left: $\Res(\hat{\GSO})/2$ averaged over perturbed graphs. 
Right: The survival rate, i.e., the fraction of perturbed graphs that retain non-null stationary states.}
  }\vspace{-.3cm}
    \label{fig:gene_perturb}
\end{figure}
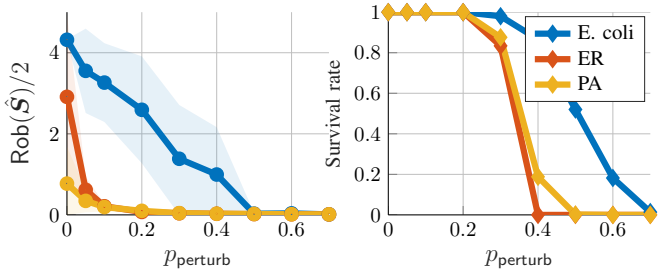

It is difficult, if not impossible, to directly verify the above conjecture.
We adopt an empirical approach by examining $\Res(\GSO)$ when $\GSO$ is taken as the \emph{E. coli} GRN \cite{gao2016universal}, which has $N=1454$ genes and $|E| = 3170$ edges,\footnote{We took the largest connected component in the \emph{E. coli} graph.} and comparing its corresponding values to control graphs generated as Erd\H{o}s-R\'{e}nyi (ER) or Preferential Attachment (PA) graphs with the same number of genes and edges.
The left panel of Fig.~\ref{fig:gene_perturb} compares average values of $\Res(\hat{\GSO})$ against $p_{\sf perturb}$, where $\hat{\GSO}$ is obtained by randomly deleting a fraction of $p_{\sf perturb}$ edges from $\GSO$.
Observe that \emph{E. coli} has a consistently higher $\Res(\hat{\GSO})$ than the control graphs with similar edge densities. This supports our conjecture.
In the right panel of Fig.~\ref{fig:gene_perturb}, we compare the survival rates of the GRNs (i.e., the probability of attaining non-null stationary states) under edge deletion perturbations. Similar observation as $\Res(\hat{\GSO})$ can be drawn, showing that \emph{E. coli} admits a more resilient topology than random graphs. 

Finally, we combine the formation model with a generation model inspired by the linear dynamic approximation in GRN inference \cite{gardner2003inferring,yuan2011robust}. Utilizing observations of the GRN stationary states after knockout experiments as ${\bm X} \approx - \GSO^{-1} {\bm P}$, we consider the following graph topology learning problem: 
\begin{tcolorbox}[boxsep=1pt,left=1pt,right=1pt,top=-4pt,bottom=1pt]
\begin{align}
\min_{\GSO, {\bm y}} &~ \Phi( \GSO; {\bm y} ):= \|\GSO {\bm X} + {\bm P}\|_F^2 + \beta \|\GSO\|_F^2 - \lambda \sum_{i=1}^N \frac{ 1 }{ 1+e^{-\sigma y_i}} \notag \\
\text{s.t.} &~ y_i = \sum_{j=1}^N \frac{ S_{ij} y_j^2}{y_j^2 + 1} ,~\forall~ i \in V,~\GSO \in  \mathcal{S}_{\sf nd}.\label{eq:glfp_gene} \tag{GL-GENE}
\end{align}
\end{tcolorbox}
\noindent Here, $\beta = \frac{\lambda\tilde{\beta}}{2} > 0, \lambda > 0$ are regularization parameters and $\mathcal{S}_{\sf nd}$ was defined in \eqref{eq:Snd} with the normalization constant $a>0$.
\update{We note that, unlike the network-game case, the lower-level stationary system in \eqref{eq:glfp_gene} may admit multiple equilibria. This distinction will be important when we discuss the convergence of our proposed algorithm in Section~\ref{sec:alg}.}


\section{TTGD Algorithm}\label{sec:alg}
This section develops a two-timescale gradient descent (TTGD) algorithm to tackle the graph learning problem with functional priors. Recall that \eqref{eq:GL_bilevel:UNI} can be written as
\begin{equation} 
    \label{eq:GL_bilevel:UNI2}
\displaystyle \min_{ \GSO\in \mathcal{S}, {\bm y}\in {\cal Y}} \Phi( \GSO; {\bm y} )
~~\text{s.t.}~\textstyle {\sf Y}(\bm{y};\GSO) = \bm{0}.
\end{equation}
Note that ${\sf Y}(\bm{y};\GSO) = {\bm 0}$ consists of $N$ nonlinear equality constraints. For the purpose of illustration, we assume that (i) \emph{the set of constraints admits a unique solution $\bm{y}^\star(\GSO)$ for any $\GSO \in \mathcal{S}$, i.e., ${\cal Y}^\star(\GSO) = \{ \bm{y}^\star(\GSO) \}$}, and (ii) \emph{$\bm{y}^\star(\GSO)$ is Lipschitz continuous w.r.t. $\GSO$}.
These are common assumptions in the bilevel optimization literature \cite{liu2022inducing, hong2023two}.
Under these assumptions, \eqref{eq:GL_bilevel:UNI2} can be reformulated as
\[
\min_{ \GSO \in {\cal S} }~ \ell( \GSO ) := \Phi( \GSO; {\bm y}^\star ( \GSO ) ).
\]

The above problem can be handled via the standard projected gradient descent (PGD) algorithm: At iteration $k \geq 0$,
\begin{equation} \label{eq:proj_grad}
\GSO^{k+1} = {\sf Proj}_{ \cal S } ( \GSO^k - \gamma \grd \ell( \GSO^k ) ),~\forall~k \geq 0,
\end{equation} 
where $\gamma>0$ is the step size, ${\sf Proj}_{{\cal S}}(\cdot)$ is the Euclidean projection onto ${\cal S}$. The challenge of \eqref{eq:proj_grad}, however, lies in the gradient computation $\grd \ell( \GSO^k )$ since $\ell( \GSO )$ has an implicit dependence on $\GSO$ via ${\bm y}^*(\cdot)$. To see this, we let $\bar{\bm y}^k := {\bm y}^\star( \GSO^k )$ and note that when ${\sf Y}(\cdot)$ is smooth, it is shown that \cite{liu2022inducing}
\begin{align}
& \grd \ell( \GSO^k ) = \widehat{\grd} \Phi( \GSO^k ; \bar{\bm y}^k ) := \grd_{\GSO} \Phi( \GSO^k; \bar{\bm y}^k ) \label{eq:grad_ideal} \\
& - ( {\rm J}_{\GSO} {\sf Y}( \bar{\bm y}^k; \GSO^k) )^\top ( {\rm J}_{\bm y} {\sf Y}( \bar{\bm y}^k; \GSO^k))^{-\top} \grd_{\bm y} \Phi( \GSO^k; \bar{\bm y}^k ), \notag
\end{align} 
where ${\rm J}_{\GSO} {\sf Y}(\cdot), {\rm J}_{\bm y} {\sf Y}(\cdot)$ denote the Jacobian of the operator ${\sf Y}$ w.r.t.~$\GSO, {\bm y}$, respectively; and $\grd_{\GSO} \Phi(\cdot), \grd_{\bm y} \Phi(\cdot)$ denote the partial gradient taken w.r.t.~$\GSO, {\bm y}$, respectively. Hence, evaluating $\grd \ell( \GSO^k )$ requires the unique solution to ${\sf Y}(\bm{y};\GSO^k ) = \bm{0}$ in $\bm{y}$.
\update{We use the symbol $\widehat{\grd}\Phi(\cdot)$ to emphasize that an approximate implicit gradient is used. It coincides with the exact gradient $\grd \ell(\GSO^k)$ when evaluated at the exact lower-level solution $\bar{\bm y}^k$, and it serves as an approximation when ${\bm y}^k$ only tracks $\bm y^\star(\GSO^k)$.}

To tackle the issue, we observe that
\[
{\sf Y}(\bm{y};\GSO) = \bm{0} \Leftrightarrow {\sf F}(\bm{y};\GSO) = \bm{y}, 
\quad {\sf F}( \bm{y} ;\GSO) := \bm{y} - {\sf Y}( \bm{y} ; \GSO )
.\] 
If ${\sf F}(\cdot; \GSO): {\cal Y} \to {\cal Y}$ is a contractive self-map for any $\GSO$, then the fixed point iteration ${\bm y}^+ \leftarrow {\sf F}( {\bm y}; \GSO^k )$ finds ${\bm y}^\star( \GSO^k )$ at a linear rate \cite{parise2019variational}. 
The above observations suggest a two-timescale algorithm mimicking \eqref{eq:proj_grad} that simultaneously updates ${\bm y}$ and $\GSO$ at different paces --- a larger step size for ${\bm y}$ and a smaller step size for $\GSO$. In this way, $\GSO^k$ will appear to be 'static' w.r.t.~${\bm y^k}$. Hence, the fast convergence of the fixed point iteration of ${\bm y^k}$ allows it to track the true solution $\bm{y}^\star(\GSO^k)$. Subsequently, we have $\widehat{\nabla}\Phi(\GSO^k; {\bm y}^k) \approx \nabla\ell(\GSO^k)$.

Inspired by \cite{hong2023two} and the above idea, we formalize our algorithm as follows: Let $\gamma >0$, $\alpha \in (0,1]$ be the step sizes and consider
\begin{tcolorbox}[boxsep=2pt,left=4pt,right=4pt,top=3pt,bottom=3pt, 
]
\underline{\emph{Two Timescale Gradient Descent (TTGD)}}: For $k \geq 0$,
\vspace{-.0cm}\begin{subequations} \label{eq:ttgd}
\begin{align}
{\bm y}^{k+1} & = {\bm y}^k + \alpha ( {\sf F}( {\bm y}^k ; \GSO^k ) - {\bm y}^k ), \label{eq:ttgd_y} \\
\GSO^{k+1} & = {\sf Proj}_{{ \cal S }} ( \GSO^k - \gamma \widehat{\grd} \Phi( \GSO^k ; {\bm y}^{k+1} ) ). \vspace{-.1cm} \label{eq:ttgd_W}
\end{align}
\end{subequations} 
\end{tcolorbox}
In the update of ${\bm y}^k$ in \eqref{eq:ttgd_y}, we have resorted to a relaxed version of the fixed point iteration with $\alpha \in (0,1]$. Our analysis will show that such a relaxation is necessary to guarantee convergence in several cases. Since ${\bm y}^{k+1}$ is formed by a convex combination between ${\bm y}^k$ and ${\sf F}( {\bm y}^k ; \GSO^k )$, by induction ${\bm y}^{k+1} \in {\cal Y}$ for any $k$ as long as ${\bm y}^0\in{\cal Y}$. 



Although TTGD is developed under the assumption that lower-level equations have a unique solution, \update{the TTGD iterates remain well defined and can still be computed} when the uniqueness assumption is violated. In the latter case, $\bm{y}^\star(\GSO)$ may not be a singleton and $\nabla \ell(\cdot)$ is not well defined. \update{Accordingly, $\widehat{\nabla}\Phi(\GSO^k;\bar{\bm{y}}^k)$ should then be interpreted only as a computable search direction akin to a `subdifferential' of $\ell(\cdot)$.} Nevertheless, \update{the updates themselves remain well-defined} as long as ${\sf Y}(\cdot;\GSO)$ is differentiable and ${\sf J}_{\bm{y}}{\sf Y}$ is invertible.


\begin{Remark}\label{rem:ttgd_complexity}
\update{The key computational bottleneck of TTGD lies in \eqref{eq:grad_ideal}, which requires the inverse $({\rm J}_{\bm y} {\sf Y}( {\bm y}^k; \GSO^k))^{-1}$ at a cost of ${\cal O}( N^3 )$ FLOPs per iteration. In large-scale settings, one may approximate ${\rm J}_{\bm y} {\sf Y}$ by its diagonal part $D$ before taking the inverse \cite{zhang2022advancing}. Writing ${\rm J}_{\bm y} {\sf Y} = D - E$ with $E$ collecting the off-diagonal entries, if $\| D^{-1} E \| < 1$, then $({\rm J}_{\bm y} {\sf Y})^{-1} = D^{-1}(I - D^{-1}E)^{-1}$ admits a Neumann expansion, and the diagonal approximation retains the leading-order term $D^{-1}$.
Note that in a network, the off-diagonal entries of ${\rm J}_{\bm y} {\sf Y}$ encode inter-node interactions in the equilibrium system. Discarding them may lead to inaccurate gradient estimates when network coupling effects are strong. 
}\vspace{-.2cm}
\end{Remark}

\subsection{Convergence Analysis for TTGD}

This subsection studies the convergence of TTGD. Throughout, we characterize the convergence behavior via the stationary measure for optimizing a smooth non-convex function $\ell$ over a closed convex set ${\cal S}$. For any $\gamma > 0$, we define
\begin{align*}
{\sf G}_{\gamma} (\GSO) := \| \gamma^{-1} ( \GSO - {\sf Proj}_{{\cal S}} ( \GSO - \gamma \grd \ell( \GSO ) ) ) \|_F^2.
\end{align*}
A point $\GSO$ is called a stationary point of \eqref{eq:GL_bilevel:UNI2} if it satisfies $0 \in \grd\ell(\GSO) + \partial \iota_{{\cal S}}(\GSO)$ \cite{davis2019stochastic}. Here, $\partial \iota_{{\cal S}}(\GSO)$ is the subdifferential of the indicator function $\iota_{{\cal S}}$ at $\GSO$ \cite{rockafellar2009variational}. From the definition of projection operators, it is obvious that if $\GSO = {\sf Proj}_{{\cal S}} ( \GSO - \gamma \grd \ell( \GSO ) )$ for some $\gamma>0$, then $\GSO$ is a stationary point. When ${\sf G}_{\gamma} (\GSO)$ is small, we expect that $\GSO$ is close to a stationary point.

We specify a few general conditions that are sufficient for the convergence of TTGD:
\begin{assumption}\label{assu:cal_S}
    The set $\mathcal{S}$ is convex and compact.
\end{assumption}
\begin{assumption}\label{assu:self-F}
    For any $\GSO\in{\cal S}$ and $\bm{y}\in{\cal Y}$, ${\sf F}(\bm{y};\GSO) \in {\cal Y}$. Here,
    \\${\sf F}(\bm{y};\GSO):=\bm{y} - {\sf Y}(\bm{y};\GSO)$.
\end{assumption}
\begin{assumption}\label{assu:Y}
    There exists a $\mody \in (0,1)$ such that for any $\GSO\in \mathcal{S}$ and $\bm{y}_1, \bm{y}_2\in{\cal Y}$, the following chain holds:
    \begin{align*}
        &(1-\mody)\|\bm{y}_1 - \bm{y}_2\|_2^2\leq \langle {\sf Y}(\bm{y}_1;\GSO)-{\sf Y}(\bm{y}_2;\GSO),\bm{y}_1-\bm{y}_2\rangle \\
        & \leq \|{\sf Y}(\bm{y}_1;\GSO)-{\sf Y}(\bm{y}_2;\GSO)\|_2\|\bm{y}_1-\bm{y}_2\|_2 \leq (1+\mody)\|\bm{y}_1-\bm{y}_2\|_2^2.
    \end{align*}
\end{assumption}
\begin{assumption}\label{assu:phi_grd}
    The objective function $(\bm{S},\bm{y}) \mapsto \Phi(\bm{S};\bm{y})$ is smooth. In particular, there exists an $L_{\Phi} > 0$ such that for all $\GSO_1, \GSO_2\in{\cal S}$ and $\bm{y}_1, \bm{y}_2\in{\cal Y}$,
        \begin{align*}
            \|\grd \Phi(\GSO_1;\bm{y}_1) - \grd \Phi(\GSO_2;\bm{y}_2)\|_F 
            \leq L_{\Phi}\|(\GSO_1;{\bm y}_1) - (\GSO_2;{\bm y}_2)\|_F.
        \end{align*}
        Here, $\grd \Phi = \left(\grd_{\GSO} \Phi; \grd_{\bm y} \Phi\right)$ is the gradient of $\Phi(\GSO;{\bm y})$ taken w.r.t. $(\GSO;\bm{y})$.
\end{assumption}
\begin{assumption}\label{assu:Y_grd}
    The mapping $(\bm{y},\bm{S}) \mapsto {\sf Y}(\bm{y};\bm{S})$ is smooth. In particular, for all $B > 0$, there exists an $L_{\sf Y} > 0$ such that for all $\GSO_1,\GSO_2\in {\cal S}$ and $\bm{y}_1, \bm{y}_2\in {\cal Y}$ satisfying $\|{\bm y}_1\|_2, \|\bm{y}_2\|_2\leq B$, we have
        \begin{align*}
            \|{\sf J} {\sf Y}(\bm{y}_1;\GSO_1) - {\sf J} {\sf Y}(\bm{y}_2;\GSO_2)\|_F 
            \leq L_{\sf Y}\|(\GSO_1;{\bm y}_1) - (\GSO_2;{\bm y}_2)\|_F,
        \end{align*}
        where ${\sf J} {\sf Y} = \left({\sf J}_{\bm y} {\sf Y}; {\sf J}_{\GSO} {\sf Y}\right)$ is the Jacobian of ${\sf Y}({\bm y};\GSO)$ w.r.t. $(\bm{y};\GSO)$.
\end{assumption}

Here, a mapping is deemed smooth if it has continuous partial derivatives of all orders. Among the above assumptions, H\ref{assu:self-F} and H\ref{assu:Y} imply that ${\sf F}(\cdot;\GSO)$ is a contractive self-map for any $\GSO\in {\cal S}$. Hence, there exists a unique solution $\bm{y}^\star$ to ${\sf Y}(\bm{y};\GSO) = \bm{0}$ for any $\GSO\in{\cal S}$. Together with H\ref{assu:cal_S}, it can be shown that $\{\bm{y}^k\}$ from updates \eqref{eq:ttgd_y} with properly chosen $\alpha$ is bounded. Furthermore, H\ref{assu:phi_grd} and H\ref{assu:Y_grd} imply that the gradient map $\widehat{\grd}\Phi(\cdot)$ is locally Lipschitz continuous w.r.t. $\GSO$ and $\bm{y}$, respectively. Together, $\widehat{\grd}\Phi(\cdot)$ is Lipschitz continuous over the sequence $\{(\GSO^k, \bm{y}^k)\}$ generated by TTGD updates \eqref{eq:ttgd}. Based on the above observations, we adapt the proof in \cite{hong2023two} and obtain the following convergence result. \update{The full proof is deferred to Appendix~\ref{app:prof_conv_thm}. We give the main argument below.}

\begin{Theorem} \label{thm:ttgd}
\update{Suppose that H\ref{assu:cal_S} to H\ref{assu:Y_grd} hold. Let $L_\ell$, $L_{\widehat\Phi}$, and $L_y$
be the constants in Proposition~\ref{prop:Lipschitz-ness_all_functions} in 
 Appendix~\ref{app:prof_conv_thm}. Set
\[
\alpha = \frac{1-\mody}{(1+\mody)^2},
\qquad
\gamma \leq \min\!\left\{ \frac{3}{4L_\ell},\ \frac{\alpha(1-\mody)}{4L_{\widehat\Phi}L_y} \right\}.
\]}
Then, for any $K \geq 1$, we have
\[
\min_{k=1,\ldots,K} {\sf G}_\gamma(\GSO^k) = {\cal O}(K^{-1}).
\]
\end{Theorem}


\update{\emph{Proof sketch.} The proof has three parts. We first show that the lower-level iteration \eqref{eq:ttgd_y} is contractive, so ${\bm y}^k$ tracks ${\bm y}^\star(\GSO^k)$. We then use this to control the gap between the computable implicit gradient $\widehat{\grd}\Phi(\GSO^k;{\bm y}^{k+1})$ and the exact gradient $\grd \ell(\GSO^k)$. Finally, the Jacobian regularity assumptions ensure that the implicit gradient is well defined along the iterates. The result then follows from a projected descent argument. \hfill $\square$}

\vspace{.1cm}
\noindent \textbf{Network Games.}
We next provide sufficient conditions for satisfying H\ref{assu:cal_S} to H\ref{assu:Y_grd} by specializing \eqref{eq:GL_bilevel:UNI2} to the network games settings in Sec.~\ref{subsec:ng}.
Throughout our discussion, we take ${\cal S} = {\cal S}_{\sf ng}$ as specified in \eqref{eq:Sng} with the parameter $c>0$.

We first consider the assumptions about the functional regularizer ${\cal R}(\GSO)$ induced by LQ games. 
\begin{assumption}\label{assu:lipsf_lq}
Consider the payoff function $U_i^{lq}(\cdot)$ in  \eqref{eq:lq_game}  with ${\cal Y} = [0,\infty)^N$.
\begin{itemize}[leftmargin=*]
    \item For all $i \in [N]$, the marginal benefit satisfies $b_i \geq b^\star> 0$. 
    \item The interaction function satisfies $f(0)=0$, $f(y) > 0$ for any $y>0$. 
    \item There exist $L_{f,1}, L_{f,2}$ such that the interaction function satisfies $|f'(y)| \leq L_{f,1}$, $|f''(y)| \leq L_{f,2}$ for any $y \in (0,\infty)$. Moreover, we have $L_{f,1} < \frac{1}{c}$.
\end{itemize}
\end{assumption}
\noindent 
Under the above assumption, the work \cite{cai2024optimal} showed that an NE exists and is unique for any $\GSO \in {\cal S}$. Furthermore, we observe that
\begin{Prop}\label{prop:conv_gllq}
    Under H\ref{assu:lipsf_lq}, there exists a smooth map ${\sf \tilde{Y}}(\cdot)$ such that ${\sf \tilde{Y}}(\bm{y};\GSO) = {\sf Y}(\bm{y};\GSO)$ for any $\bm{y}\in{\cal Y}$, $\GSO\in{\cal S}_{\sf ng}$. 
    Furthermore, H\ref{assu:cal_S} to H\ref{assu:Y_grd} hold for the special case of \eqref{eq:glfp_ng} with ${\sf \tilde{Y}}(\cdot)$ and $\Phi(\cdot)$.
\end{Prop}
\noindent \update{The full proof is deferred to Appendix~\ref{app:proof_conv_prop}.}
We next consider RT games to instantiate the functional regularizer design:
\begin{assumption}\label{assu:lipsg_rt}
Consider the payoff function $U_i^{rt}(\cdot)$ in  \eqref{eq:rt_game} with ${\cal Y} = [\underline b_i, a_i]^N := \{\bm{y}\in\RR^N: y_i\in [ \underline b_i, a_i ], ~\forall~i\in[N]\}$.
\begin{itemize}[leftmargin=*]
    \item For all $i \in [N]$, the marginal benefit satisfies $b_i \geq b^\star > 0$. 
    \item The interaction function satisfies $g(0)=0$, $g(y) > 0$ for any $y>0$. 
    \item There exist $L_{g,1}, L_{g,2}$ such that the interaction function satisfies $|g'(x)| \leq L_{g,1}$, $|g''(x)| \leq L_{g,2}$ for any $x \in [0, c \, \max_i a_i ]$. Moreover, we have $L_{g,1} < c^{-1} \left(1- \max_i \frac{\underline b_i}{a_i}\right)$.
\end{itemize}
\end{assumption}
\noindent Under the above assumption, the work \cite{parise2019variational} showed that an NE exists and is unique for any $\GSO \in {\cal S}$. Similarly, we observe that 
\begin{Prop}\label{prop:conv_glrt}
    Under H\ref{assu:lipsg_rt}, there exists a smooth map ${\sf \tilde{Y}}(\cdot)$ such that ${\sf \tilde{Y}}(\bm{y};\GSO) = {\sf Y}(\bm{y};\GSO)$ for any $\bm{y}\in{\cal Y}$, $\GSO\in{\cal S}_{\sf ng}$. 
    Furthermore, H\ref{assu:cal_S} to H\ref{assu:Y_grd} hold for the special case of \eqref{eq:glfp_ng} with ${\sf \tilde{Y}}(\cdot)$ and $\Phi(\cdot)$.
\end{Prop}

The sufficient condition for convergence of TTGD is satisfied for both cases. Subsequently, the ${\cal O}(1/{K})$ convergence rate in Theorem \ref{thm:ttgd} holds for the TTGD applied to \eqref{eq:glfp_ng}.


\begin{Remark}\label{rem:gene_limitation}
\update{For the graph learning formulation \eqref{eq:glfp_gene} considered in Sec.~\ref{subsec:nd}, H\ref{assu:self-F} and H\ref{assu:Y} are violated whenever the system ${\sf Y}({\bm y};\GSO)={\bm 0}$ admits a nontrivial solution ${\bm y}\neq {\bm 0}$ alongside the persisting trivial solution ${\bm y}={\bm 0}$. Since the regime with nontrivial stationary states is precisely the case of interest in GL-GENE, Theorem~\ref{thm:ttgd} does not provide a convergence guarantee for TTGD when it is applied to \eqref{eq:glfp_gene}.
Nevertheless, the numerical results in Sec.~\ref{sec:num-gene} suggest that TTGD can still produce meaningful graph estimates in this more challenging regime, which is beyond the scope of our current analysis. This suggests that TTGD still provides an approximate descent direction when tackling \eqref{eq:glfp_gene}, even when the lower-level dynamics are nonlinear and admit multiple equilibrium branches.
}
\end{Remark}

\section{Structural Interpretation for \texorpdfstring{${\cal R}(S)$}{R(S)}} \label{sec:structural_interpretation}
This section investigates the relationship between functional priors and structural priors on graph learning. 
\update{We note that as \eqref{eq:GL_bilevel:UNI} gives a \emph{biased} graph estimator with $\lambda > 0$, even when $M \gg 1$, an optimal solution to \eqref{eq:GL_bilevel:UNI} may not coincide with the true underlying graph. Thus, it may be futile to discuss exact recoverability with \eqref{eq:GL_bilevel:UNI}. In this section,
we focus on analyzing the optimal solutions to graph learning problems that are regularized by functional priors and suggest the types of graph structure they induce.}
As bilevel optimization problems are often difficult to analyze, we concentrate on the two case studies in Sec.~\ref{sec:prob} and consider the approximations to the latter.

\vspace{.1cm}
\noindent \textbf{Network Games.} Our first step is to construct an approximation of \eqref{eq:glfp_ng} that is amenable to analysis. Consider the following properties, which are implied by H\ref{assu:lipsf_lq}, H\ref{assu:lipsg_rt} in the previous section:
\begin{enumerate}
    \item The functions $f(\cdot),g(\cdot)$ are $\ell_1$-Lipschitz continuous with $\ell_1<\frac{1}{c}$.
    \item We have $f(y), g(y) > 0$ for any $y > 0$.
    \item For any $\GSO\in{\cal S}_{\sf ng}$, $\bm{y}^{\sf NE}(\GSO)$ is in the interior of ${\cal Y}$.
\end{enumerate}
Observe the following proposition:
\begin{Prop} \label{prop:approx}
Consider the bilevel problem \eqref{eq:glfp_ng}. Under the above conditions, we have
\begin{equation*}
    \|\bm{y}^{\sf NE}(\GSO)\|_2 \leq B^\star=\frac{\sqrt{N}\,\|\bm{b}\|_\infty}{1-c\ell_1}, ~~\forall~\GSO\in{\cal S}_{\sf ng}.
\end{equation*}
Consequently, for any $\GSO\in{\cal S}_{\sf ng}$, we have
\begin{equation}\label{eq:bounds_glng}
	\begin{aligned}
		&{\rm Tr}( \GSO^\top {\bm D} ) + \beta \| \GSO \|_F^2 - \lambda \ell_1 \bm{1}^\top\GSO\bm{b} - \lambda \bm{1}^\top\bar{\bm{c}}\\ 
        & \leq {\rm Tr}( \GSO^\top {\bm D} ) + \beta \| \GSO \|_F^2 - \lambda\bm{1}^\top\bm{y}^{\sf NE}(\GSO)
        \\
		& \leq
		{\rm Tr}( \GSO^\top {\bm D} ) + \beta \| \GSO \|_F^2 - \lambda \mu_1\bm{1}^\top\GSO\bm{b} - \lambda \bm{1}^\top\hat{\bm{c}},
	\end{aligned}	
\end{equation}
where $\bar{\bm c} = \frac{(c \ell_1)^2}{1-c \ell_1}{\| {\bm b} \|_\infty} {\bf 1} + {\bm b}$, $\hat{\bm c} = \frac{(c\mu_1)^2}{1-c\mu_1}{ b^\star} {\bf 1} + {\bm b}$ with $b^\star = \min_{i\in[N]}b_i$. For LQ games, $\mu_1 = \frac{\min_{y\in[b^\star, B^\star]}f(y)}{B^\star}$, while for RT games, $\mu_1 = \frac{\min_{y\in[b^\star, B^\star]}g(cy)}{cB^\star}$. 
\end{Prop}
\noindent The proof is relegated to Appendix~\ref{pf:approx_glng}. 
Eq.~\eqref{eq:bounds_glng} implies that \eqref{eq:glfp_ng} can be approximated by 
\begin{equation} \label{eq:glfp_ng-app} \textstyle
	\min_{ \GSO \in {\cal S}_{\sf ng} } {\rm Tr}( \GSO^\top {\bm D} ) + \beta \| \GSO \|_F^2 - \tilde{\lambda} {\bf 1}^\top \GSO {\bm b} \tag{GL-NG-App}
\end{equation}
for some $\tilde{\lambda} > 0$ proportional to $\lambda$. Importantly, \eqref{eq:glfp_ng-app} is a convex program with affine constraints. Analyzing the KKT conditions of \eqref{eq:glfp_ng-app} leads to
\begin{Prop} \label{prop:kkt}
There exists an $\bm{\eta} \in \RR^N$ such that any optimal solution to Problem \eqref{eq:glfp_ng-app} is given by 
\[
    S_{ij}^\star = {\textstyle \frac{1}{2 \beta}} \max \Big\{ 0, \tilde{\lambda} b_j + \eta_i -D_{ij} \Big\}.
\]
for any $i \neq j$ and $S_{ii}^\star = 0$.
It also holds that $\GSO^\star {\bf 1} = c {\bf 1}$.
\end{Prop}
\noindent The proof is in Appendix~\ref{pf:kkt}.

Together, Propositions \ref{prop:approx} and \ref{prop:kkt} show that when $\tilde{\lambda} \gg 1$, an optimal solution to \eqref{eq:glfp_ng} can be approximated as $S_{ij}^\star \approx \frac{\tilde{\lambda}}{2\beta} b_j$ for any $i, j \in V$. As the marginal benefits $\bm{b}$ may vary between agents in practice, any optimal solution to \eqref{eq:glfp_ng-app} (and thus \eqref{eq:glfp_ng}) gives a graph topology that exhibits a \emph{`multiple hub' structure} where the majority of edges will be emanating from nodes with large $b_j$. This observation coincides with the human-made network examples studied in Table~\ref{table:ReaNetWkRewir}, where we observe a number of `hub' nodes.

\update{
Proposition~\ref{prop:kkt} provides a structural characterization of the biased estimator induced by the functional prior, rather than a guarantee of exact topology recovery. 
When $\lambda > 0$, the functional prior introduces a task-aligned bias, so the learned graph should be understood as balancing data fidelity and the prescribed functional objective. In particular, Proposition~\ref{prop:kkt} explains why hub-promoting structures emerge in the large-$\lambda$ regime,
which is consistent with the network-formation viewpoint underlying the functional prior. 
}

\vspace{.1cm}
\noindent \textbf{Network Dynamics.} Similar to the previous case, we proceed to derive an approximation of \eqref{eq:glfp_gene}. 
\begin{Prop} \label{prop:approx_GLgene}
Consider the problem \eqref{eq:glfp_gene} and denote its optimal value by $\upsilon_{\sf GENE}^\star$. We have \vspace{-.1cm}
\begin{equation}\label{eq:glfp_gene_approx}
\upsilon_{\sf GENE}^\star \geq 
\min_{\GSO \in {\cal S}_{\sf nd}} ~ \|\GSO \boldsymbol{X} + \boldsymbol{P}\|_F^2 + \beta \|\GSO\|_F^2 - \hat{\lambda} {\bm 1}^\top \GSO^2 {\bm 1} - \frac{N\lambda}{2}, \tag{GL-GENE-App} \vspace{-.1cm}
\end{equation}
where $\hat{\lambda} = \frac{\sigma\lambda}{8}$.
\end{Prop}
\noindent The proof is relegated to Appendix~\ref{pf:approx_glgene}. 
Proposition \ref{prop:approx_GLgene} gives a one-sided approximation of \eqref{eq:glfp_gene}. This approximation is less tight compared to that of Proposition \ref{prop:approx}. Still, it provides useful insights into the optimal solution to \eqref{eq:glfp_gene}.
\update{Additionally, we remark that the regularization term in \eqref{eq:glfp_gene_approx} is related to the resilience quantity proposed in \cite{gao2016universal}, which is defined as
\[
\beta_{\mathrm{eff}}(\bm S) = \frac{{\bm 1}^\top \bm S d(\bm S)}{{\bm 1}^\top \bm S {\bm 1}}, 
\qquad d(\bm S):=\bm S{\bm 1}.
\]
This proxy reflects a degree-heterogeneity effect at the topology level.
Under ${\bm 1}^\top \bm S {\bm 1} = a$,  $\beta_{\mathrm{eff}}(\bm S)$ is proportional to the regularization term ${\bm 1}^\top \bm S^2 {\bm 1}$ in \eqref{eq:glfp_gene_approx}.
%
}

The single-level problem \eqref{eq:glfp_gene_approx} remains nonconvex in general due to the term $- {\bf 1}^\top \GSO^2 {\bf 1}$. Therefore, we concentrate on the solutions that directly maximize ${\bf 1}^\top \GSO^2 {\bf 1}$, which correspond to the regime with $\lambda \gg 1$.
This observation motivates the structural characterization:
\begin{Prop} \label{prop:gene-interpret}
For any $\GSO^\star \in \argmax_{\GSO \in {\cal S}_{\sf nd}} {\bm 1}^\top \GSO^2 {\bm 1}$, 
there exists a pair $(i^\star,j^\star)$ such that $i^\star \neq j^\star$ and 
\[
S_{ij}^\star = \begin{cases}
    \frac{a}{2}, & \textit{if}~ (i,j) = (i^\star, j^\star)~\text{or}~(j^\star, i^\star), \\
    0, & \textit{otherwise}.
\end{cases} 
\]
\end{Prop}
\noindent The proof can be found in Appendix~\ref{pf:1s21}. In other words, when $\lambda \gg 1$, every optimal solution to \eqref{eq:glfp_gene_approx} corresponds to a graph with only one bidirected edge.

As with Proposition~\ref{prop:kkt}, this result should be interpreted as a structural characterization of the prior-dominated estimator in the large-$\lambda$ regime, rather than as a claim that the underlying GRN topology is recovered by such a one-edge structure.
In practice, the regularization parameter $\lambda$ is moderate, so the actual estimator balances data fidelity against the resilience-promoting bias rather than collapsing to the limiting one-edge solution.


\subsection{Functional Priors as Generalized Graph Filter Priors}
Our last endeavor is to interpret the functional priors via the conduit of generalized graph filter priors. In particular, we show that any polynomial regularizer of the GSO $\GSO$ can be written as a functional prior regularizer. 

To facilitate our discussion, we first introduce the notion of graph filter (GF) regularizers for use in \eqref{eq:graph_learn_prob} to induce structural properties in the learned graph. Formally
\begin{Def}
    A graph filter (GF) regularizer ${\cal R}_{\sf GF}:\RR^{N\times N}\rightarrow\RR$ is \vspace{-.1cm}
    \begin{equation*}
        {\cal R}_{\sf GF}(\GSO) = f\circ H(\GSO).
    \end{equation*}
    Here, $f:\RR^{N\times N}\rightarrow\RR$ is a linear operator and $H:\RR^{N\times N}\rightarrow\RR^{N\times N}$ is a linear graph filter defined as $H(\GSO) = h_0\bm{I} + \sum_{i = 1}^p h_i\GSO^i$, where $p$ is the filter order and $\{h_i\}_{i = 0}^p$ are the filter's coefficients  \cite{ramakrishna2020user}.
\end{Def}
The class of GF regularizers is general enough to cover most of the structural regularizers. We consider two examples. 
\begin{Exa}
    The structural regularizer ${\cal R}(\GSO) = \| \GSO \|_F^2$ is commonly used in graph learning formulations \cite{kalofolias2016learn}. In this case, $f(\cdot) = \tr(\cdot)$ and $H(\GSO) =  \GSO^\top\GSO$.
\end{Exa}
\begin{Exa}
    The spectral regularizer ${\cal R}(\GSO)$ is designed to regulate the clustering properties of a graph and can be expressed as a function  of singular values of $\GSO$ (e.g. \eqref{eq:RS_spect}) \cite{kumar2020unified}. Note that for a symmetric GSO, $\GSO$ admits an eigen-decomposition $\GSO = \bm{U}\bm{\Lambda} \bm{U}^\top$ and $H(\GSO) = \bm{U}(\sum_{i = 1}^Nh_i\bm{\Lambda}^i)\bm{U}^\top$. Hence, a spectral regularizer can be expressed as a GF regularizer with $f(\cdot) = \tr(\cdot)$ and a carefully designed low-pass graph filter $H$ that attenuates the high-pass part of the singular values of $\GSO$.
\end{Exa}
In fact, the approximate forms of the bilevel problems described above, i.e., \eqref{eq:glfp_ng-app} and \eqref{eq:glfp_gene_approx}, are also special cases of GF regularizers. For the term $-\bm{1}^\top\GSO\bm{b}$ in Proposition \ref{prop:approx}, we have $f(\cdot) = -\bm{1}^\top(\cdot)\bm{b}$ and $H(\GSO) = \GSO$. For the term $-\bm{1}^\top\GSO^2\bm{1}$ in Proposition \ref{prop:approx_GLgene}, we have $f(\cdot) = -\bm{1}^\top(\cdot)\bm{1}$ and $H(\GSO) = \GSO^2$.

Finally, we conclude with the following lemma, which shows that any GF regularizer can be expressed as a functional prior regularizer, and the corresponding graph learning problem can be written as a special case of \eqref{eq:GL_bilevel:UNI2}.
\begin{Lemma}\label{lem:gf_fp}
    For any GF regularizer ${\cal R}_{\sf GF} (\GSO) = f\circ H(\GSO)$, there exists a ${\sf Y}:\RR^N\times\RR^{N\times N} \rightarrow \RR^N$ such that for all $\GSO \in \RR^{N\times N}$,
    \begin{align}
    {\cal R}_{\sf GF} (\GSO) =  \bm{1}^\top\bm{y}(\GSO) \label{eq:riesz}
    \end{align}
and $\bm{y}(\GSO)$ is the unique solution to ${\sf Y}(\bm{y};\GSO) = \bm{0}$.
\end{Lemma}
\begin{proof}
    As $f:\RR^{N\times N}\rightarrow\RR$ is a linear operator, the Riesz representation theorem \cite{yosida2012functional} shows that there exists a constant matrix $\bm{Z}\in\RR^{N\times N}$ such that $f( H(\GSO) ) = \tr(\bm{Z}^\top H(\GSO) ) = \bm{1}^\top\bm{y}$, where $\bm{y} = {\sf diag}(\bm{Z}^\top H(\GSO) )$. Subsequently, we may choose ${\sf Y}(\cdot;\GSO)$ as the mapping ${\sf Y}(\bm{y};\GSO) = \bm{y} - {\sf diag}(\bm{Z}^\top H(\GSO))$. Observe that the desired condition \eqref{eq:riesz} holds.
\end{proof}
\vspace{-.5cm}

\section{Numerical Experiments}
This section presents numerical experiments to validate the effectiveness of the proposed graph learning method using functional priors induced by network games in Sec.~\ref{subsec:ng} and gene network dynamics in Sec.~\ref{subsec:nd}. \update{Throughout, we use the area under the ROC curve (AUC) as the primary metric for evaluating topology recovery. When relevant, we report the welfare and resilience values to illustrate how the functional priors influence the learned graph structures. Unless otherwise stated, TTGD is terminated when
\[
\frac{\|\GSO^{k+1}-\GSO^k\|_F}{\max\{1,\|\GSO^k\|_F\}} \le 10^{-6}.
\]
The parameter $\beta$ is selected from the grid $\{1,5,10,50,100,200\}$ based on the validation AUC, separately for each experiment. The step sizes $(\alpha,\gamma)$ are tuned for stable TTGD updates, while $c$ is fixed to satisfy the feasibility and contraction requirement in the analysis.}

\vspace{-.2cm}

\subsection{Network Games}\vspace{-.1cm}
In this subsection, we focus on the case with priors induced by network games in \eqref{eq:glfp_ng}. 
 
\vspace{.1cm}
\noindent $\bullet$ {\bf Synthetic Data.} 
Our first experiment evaluates the graph learning performance using synthetic random graphs and synthetic graph signals. We generate ${\cal G}$ as
a preferential attachment (PA) graph with $N = 50$ nodes and one edge to attach for every new node. 
The probability of attaching to an existing node is proportional to its degree, normalized over the total degree, i.e., $ P(i) = {d_i}/{\sum_j d_j}$. Additionally, we generate ${\cal G}$ as an Erdős–Rényi (ER) graph with $N=50$ nodes and connection probability of $0.1$.
We concentrate on a scenario with limited data acquired, where only $M = 10 \ll N$ smooth graph signals are observed. Each graph signal is generated via a low pass graph filter as \( \bm{x}_m = \exp({ \GSO /2} ) \bm{u}_m+ {\bm w}_m \), where \( {\bm u}_m \sim \mathcal{N}(\mathbf{0},\I) \) is an i.i.d.~white noise excitation and \update{${\bm w}_m \sim \mathcal{N}(\mathbf{0},\omega^2\I)$} is an i.i.d.~additive noise with $\omega=0.2$. We also fix \( (\beta,c) = (200,0.95) \) for \eqref{eq:glfp_ng}.

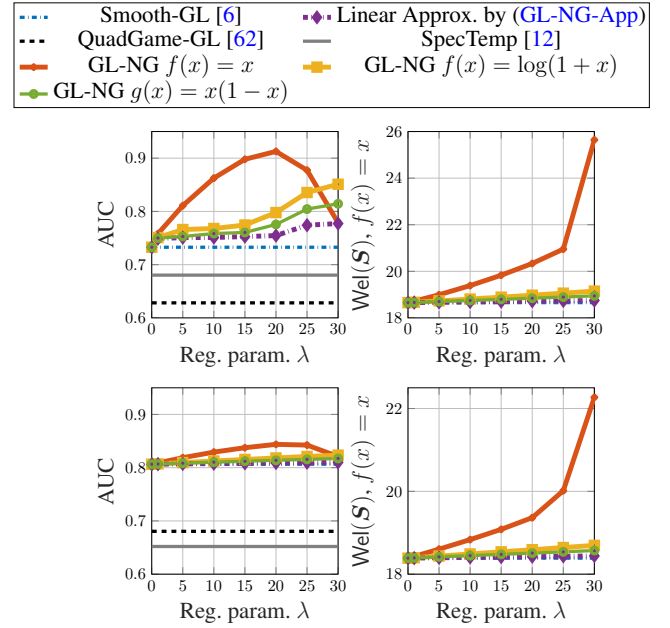
\begin{figure}[t]
\begin{center}
\resizebox{0.95\linewidth}{!}{\definecolor{mycolor1}{rgb}{0.00000,0.44700,0.74100}%
\definecolor{mycolor2}{rgb}{0.85000,0.32500,0.09800}%
\definecolor{mycolor3}{rgb}{0.92900,0.69400,0.12500}%
\definecolor{mycolor4}{rgb}{0.49400,0.18400,0.55600}%
\definecolor{mycolor5}{rgb}{0.46600,0.67400,0.18800}%
\definecolor{mycolor6}{rgb}{0.30100,0.74500,0.93300}%
\definecolor{mycolor7}{rgb}{0.63500,0.07800,0.18400}%
\begin{tikzpicture}
\begin{groupplot}[group style={group name=myplot,group size=2 by 2,horizontal sep=1.5cm, vertical sep=1.5cm}]
\nextgroupplot[%
width=4cm,
height=4cm,
scale only axis,
xmin=0,
xmax=30,
xlabel style={font=\color{white!15!black}},
xlabel={\Large Reg.~param.~$\lambda$},
ymin=0.6,
ymax=.95,
ylabel style={font=\color{white!15!black}},
ylabel={\Large AUC},
xtick={0,5,10,15,20,25,30}, 
axis background/.style={fill=white},
xmajorgrids,
ymajorgrids,
legend pos = north west,
legend style={
legend cell align=left, align=left, 
  draw opacity=0.8,
  text opacity=1,draw=white!15!black,font=\large},
]

\addplot [color=mycolor1, dashdotted, line width=2.0pt]
  table[row sep=crcr]{%
0.01	0.732710368910263\\
30	0.732710368910263\\
};\label{plots:Smooth-GL1}

\addplot [color=black, dashed, line width=2.0pt]
  table[row sep=crcr]{%
0.01	0.628264711379973\\
30	0.628264711379973\\
};\label{plots:xiaowen1}

\addplot [color=gray, line width=2.0pt]
  table[row sep=crcr]{%
0.01	0.680347499532703\\
30	0.680347499532703\\
};\label{plots:spec1}

\addplot [color=mycolor2, line width=3.0pt, mark=x, mark options={solid, mycolor2}]
  table[row sep=crcr]{%
0.01	0.732797541164676\\
1	0.758156871824501\\
5	0.810848931162807\\
10	0.862482709986578\\
15	0.897866531291953\\
20	0.91267861391018\\
25	0.877404458869313\\
30	0.778602440143419\\
};\label{plots:bi-x1}

\addplot [color=mycolor3, line width=3.0pt, mark=square, mark options={solid, mycolor3}]
  table[row sep=crcr]{%
0.01	0.732730037893593\\
1	0.750821020748017\\
5	0.765920555149621\\
10	0.767921396285409\\
15	0.77441230947\\
20	0.797793866505804\\
25	0.835454085880815\\
30	0.85103060374858\\
};\label{plots:bi-log1}

\addplot [color=mycolor4, dashdotted, line width=3.0pt, mark=diamond, mark options={solid, mycolor4}]
  table[row sep=crcr]{%
0.01	0.732714829478836\\
1	0.74904422759945\\
5	0.750766368162586\\
10	0.751004825910381\\
15	0.752758058760557\\
20	0.754593387313294\\
25	0.774134437288655\\
30	0.777176332647964\\
};\label{plots:bi-app1}

\addplot [color=mycolor5, line width=2.0pt, mark=o, mark options={solid, mycolor5}]
  table[row sep=crcr]{%
0.01	0.7327215415725\\
1	0.749949404407892\\
5	0.752878477119408\\
10	0.758428584173054\\
15	0.760464600078167\\
20	0.775513772536493\\
25	0.804405661098746\\
30	0.814308866760694\\
};\label{plots:bi-race1}

\nextgroupplot[%
width=4cm,
height=4cm,
scale only axis,
xmin=0,
xmax=30,
xlabel style={font=\color{white!15!black}},
xlabel={\Large Reg.~param.~$\lambda$},
ymin=18,
ymax=26,
ylabel style={font=\color{white!15!black}},
ylabel={\Large ${\sf Wel}(\GSO)$, $f(x)=x$},
axis background/.style={fill=white},
xmajorgrids,
ymajorgrids,
xtick={0,5,10,15,20,25,30}, 
legend style={legend cell align=left, align=left, draw=white!15!black},
]
\addplot [color=mycolor1, dashdotted,   line width=2.0pt]
  table[row sep=crcr]{%
0.01	18.6559922114244\\
30	18.6559922114244\\
};

\addplot [color=mycolor2, line width=3.0pt, mark=x, mark options={solid, mycolor2}]
  table[row sep=crcr]{%
0.01	18.6559922114244\\
1	18.7187011571116\\
5	18.9949173686688\\
10	19.3869067793036\\
15	19.8326656995996\\
20	20.3362145110314\\
25	20.9527448938633\\
30	25.6417357551255\\
};

\addplot [color=mycolor3, line width=3.0pt, mark=square, mark options={solid, mycolor3}]
  table[row sep=crcr]{%
0.01	18.6555211470622\\
1	18.6701685223811\\
5	18.7309408234496\\
10	18.8094896448317\\
15	18.8905893209423\\
20	18.9749778753209\\
25	19.0617033890861\\
30	19.1504497333683\\
};

\addplot [color=mycolor4, dashdotted, line width=3.0pt, mark=diamond, mark options={solid, mycolor4}]
  table[row sep=crcr]{%
0.01	18.6554069980318\\
1	18.6583487517021\\
5	18.6703295305705\\
10	18.6855369436846\\
15	18.7009121628393\\
20	18.7164857536113\\
25	18.732250098059\\
30	18.7482488725199\\
};

\addplot [color=mycolor5, line width=2.0pt, mark=o, mark options={solid, mycolor5}]
  table[row sep=crcr]{%
0.01	18.6554638482277\\
1	18.6643872653966\\
5	18.7011124991949\\
10	18.7482319433893\\
15	18.7962813319375\\
20	18.8452475479665\\
25	18.8952518208233\\
30	18.9466310830431\\
};

\nextgroupplot[%
width=4cm,
height=4cm,
scale only axis,
xmin=0,
xmax=30,
xlabel style={font=\color{white!15!black}},
xlabel={\Large Reg.~param.~$\lambda$},
ymin=0.6,
ymax=.95,
ylabel style={font=\color{white!15!black}},
ylabel={\Large AUC},
xtick={0,5,10,15,20,25,30}, 
axis background/.style={fill=white},
xmajorgrids,
ymajorgrids,
legend pos = north west,
legend style={
legend cell align=left, align=left, 
  draw opacity=0.8,
  text opacity=1,draw=white!15!black,font=\large},
]
\addplot [color=mycolor1, dashdotted,   line width=2.0pt]
  table[row sep=crcr]{%
0.01	0.806481143505729\\
30	0.806481143505729\\
};

\addplot [color=black, dashed, line width=2.0pt]
  table[row sep=crcr]{%
0.01	 0.6805\\
30	 0.6805\\
};

\addplot [color=gray, line width=2.0pt]
  table[row sep=crcr]{%
0.01	0.6520\\
30	0.6520\\
};

\addplot [color=mycolor2, line width=3.0pt, mark=x, mark options={solid, mycolor2}]
  table[row sep=crcr]{%
0.01	0.806503771909897\\
1	0.809006433233156\\
5	0.818724474138318\\
10	0.829258092629333\\
15	0.837483204200314\\
20	0.84378911293524\\
25	0.842457839231684\\
30	0.820891302100142\\
};

\addplot [color=mycolor3, line width=3.0pt, mark=square, mark options={solid, mycolor3}]
  table[row sep=crcr]{%
0.01	0.806485001866213\\
1	0.807089120435149\\
5	0.809541536871005\\
10	0.812562137573773\\
15	0.815431647508984\\
20	0.818259453462448\\
25	0.820978985032714\\
30	0.823559327043102\\
};

\addplot [color=mycolor4, dashdotted, line width=3.0pt, mark=diamond, mark options={solid, mycolor4}]
  table[row sep=crcr]{%
0.01	0.806481934827433\\
1	0.806598762825022\\
5	0.807083100511462\\
10	0.80766523664601\\
15	0.808280871314789\\
20	0.808904528091311\\
25	0.809506716135221\\
30	0.810061888072538\\
};

\addplot [color=mycolor5, line width=2.0pt, mark=o, mark options={solid, mycolor5}]
  table[row sep=crcr]{%
0.01	0.806482172386475\\
1	0.806848881158415\\
5	0.808354793842855\\
10	0.810269698385675\\
15	0.812079512510925\\
20	0.813837081271249\\
25	0.815575351292969\\
30	0.817292583232222\\
};

\nextgroupplot[%
width=4cm,
height=4cm,
scale only axis,
xmin=0,
xmax=30,
xlabel style={font=\color{white!15!black}},
xlabel={\Large Reg.~param.~$\lambda$},
ymin=18,
ymax=22.5,
ylabel style={font=\color{white!15!black}},
ylabel={\Large ${\sf Wel}(\GSO)$, $f(x)=x$},
axis background/.style={fill=white},
xmajorgrids,
ymajorgrids,
xtick={0,5,10,15,20,25,30}, 
legend style={legend cell align=left, align=left, draw=white!15!black}
]

\addplot [color=mycolor1, dashdotted,   line width=2.0pt]
  table[row sep=crcr]{%
0.01	18.3865334057906\\
30	18.3865334057906\\
};

\addplot [color=mycolor2, line width=3.0pt, mark=x, mark options={solid, mycolor2}]
  table[row sep=crcr]{%
0.01	18.3869617143937\\
1	18.429413087119\\
5	18.6038325309536\\
10	18.8341092677032\\
15	19.0820605066082\\
20	19.3640835789467\\
25	20.0149692155797\\
30	22.2683939160429\\
};

\addplot [color=mycolor3, line width=3.0pt, mark=square, mark options={solid, mycolor3}]
  table[row sep=crcr]{%
0.01	18.3866341581018\\
1	18.3966405593397\\
5	18.4370868418648\\
10	18.4878997691688\\
15	18.5389454058512\\
20	18.5904366543407\\
25	18.6425068331498\\
30	18.6952728593818\\
};

\addplot [color=mycolor4, dashdotted, line width=3.0pt, mark=diamond, mark options={solid, mycolor4}]
  table[row sep=crcr]{%
0.01	18.3865529837061\\
1	18.3883514594875\\
5	18.3956262903679\\
10	18.4047361542458\\
15	18.4138721366439\\
20	18.4230391999775\\
25	18.4322403180333\\
30	18.4414700531902\\
};

\addplot [color=mycolor5, line width=2.0pt, mark=o, mark options={solid, mycolor5}]
  table[row sep=crcr]{%
0.01	18.3865940122696\\
1	18.3926301960531\\
5	18.4170120348085\\
10	18.4475759842478\\
15	18.4782366834609\\
20	18.5089898279431\\
25	18.5398207175854\\
30	18.5708119118544\\
};

\coordinate (top) at (rel axis cs:1,0);
\coordinate (bot) at (rel axis cs:0,0);

\end{groupplot}

\path (top|-current bounding box.north)--
      coordinate(legendpos)
      (bot|-current bounding box.north);
\matrix[
    matrix of nodes,
    anchor=south,
    draw,
    inner sep=0.1em,
    draw,
    font = \small,
  ]at([yshift=1.75ex,xshift=-23ex]legendpos)
  {\ref{plots:Smooth-GL1}& {\Large Smooth-GL \cite{kalofolias2016learn}}&[5pt]
  \ref{plots:bi-app1}& {\Large Linear Approx. by \eqref{eq:glfp_ng-app}} \\
  \ref{plots:xiaowen1}& {\Large QuadGame-GL \cite{leng2020learning}} &[5pt] 
  \ref{plots:spec1}& {\Large SpecTemp \cite{segarra2017network}} \\
  \ref{plots:bi-x1}& {\Large GL-NG $f(x) = x$}&[5pt]
\ref{plots:bi-log1}& {\Large GL-NG $f(x) = \log(1+x)$}\\
\ref{plots:bi-race1}& {\Large GL-NG $g(x) = x(1-x)$}\\
};

\end{tikzpicture}%
}\vspace{-.2cm}
\end{center}
\caption{Performance of \eqref{eq:glfp_ng} using TTGD for learning from PA graphs (top) and ER graphs (bottom). Left: AUC scores. Right: Social welfare values.}\vspace{-.3cm}\label{fig:cr_game}
\end{figure}

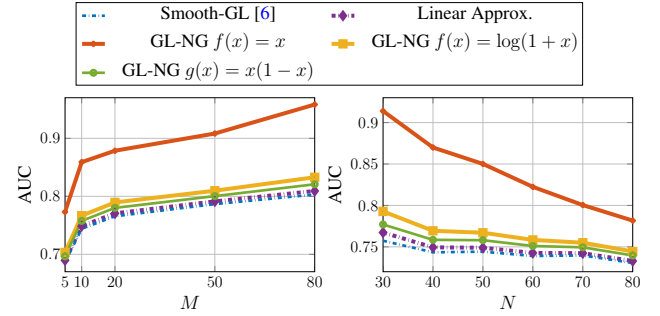
\begin{figure}[t]
\begin{center}
\resizebox{0.95\linewidth}{!}{\definecolor{mycolor1}{rgb}{0.00000,0.44700,0.74100}
\definecolor{mycolor2}{rgb}{0.85000,0.32500,0.09800}
\definecolor{mycolor3}{rgb}{0.92900,0.69400,0.12500}
\definecolor{mycolor4}{rgb}{0.49400,0.18400,0.55600}
\definecolor{mycolor5}{rgb}{0.46600,0.67400,0.18800}

\begin{tikzpicture}
\begin{groupplot}[
group style={group name=myplot, group size=2 by 1, horizontal sep=1.8cm},
axis background/.style={fill=white},
xmajorgrids,
ymajorgrids,
tick label style={font=\large},
label style={font=\large},
width=0.45\textwidth,
height=0.34\textwidth,
]

\nextgroupplot[
xlabel={\Large $M$},
ylabel={\Large AUC},
xmin=5,
xmax=80,
xtick={5,10,20,50,80},
ymin=0.67,
ymax=0.97,
]

\addplot [color=mycolor2, line width=3.0pt, mark=x, mark options={solid, mycolor2}]
table[row sep=crcr]{
5   0.773194850000000\\
10  0.858927190000000\\
20  0.878604350000000\\
50  0.908268090000000\\
80  0.958253430000000\\
};\label{plots:m-li}

\addplot [color=mycolor4, dashdotted, line width=3.0pt, mark=diamond, mark options={solid, mycolor4}]
table[row sep=crcr]{
5   0.689689930000000\\
10  0.748726830000000\\
20  0.770638310000000\\
50  0.791291270000000\\
80  0.809225940000000\\
};\label{plots:m-ab}

\addplot [color=mycolor1, dashdotted, line width=2.0pt]
table[row sep=crcr]{
5   0.686193800000000\\
10  0.744114390000000\\
20  0.765605940000000\\
50  0.786562110000000\\
80  0.803116130000000\\
};\label{plots:m-sig}

\addplot [color=mycolor3, line width=3.0pt, mark=square, mark options={solid, mycolor3}]
table[row sep=crcr]{
5   0.703111140000000\\
10  0.766622630000000\\
20  0.789473800000000\\
50  0.809652990000000\\
80  0.832880550000000\\
};\label{plots:m-log}

\addplot [color=mycolor5, line width=2.0pt, mark=o, mark options={solid, mycolor5}]
table[row sep=crcr]{
5   0.696414980000000\\
10  0.757610470000000\\
20  0.780069120000000\\
50  0.800372460000000\\
80  0.820758210000000\\
};\label{plots:m-race}

\nextgroupplot[
xlabel={\Large $N$},
ylabel={\Large AUC},
xmin=30,
xmax=80,
xtick={30,40,50,60,70,80},
ymin=0.72,
ymax=0.93,
]

\addplot [color=mycolor2, line width=3.0pt, mark=x, mark options={solid, mycolor2}]
table{
30 0.91394729
40 0.86991720
50 0.85010383
60 0.82242706
70 0.80043714
80 0.78170213
};\label{plots:n-li}

\addplot [color=mycolor4, dashdotted, line width=3.0pt, mark=diamond, mark options={solid, mycolor4}]
table{
30 0.76734847
40 0.74944262
50 0.74892139
60 0.74268061
70 0.74235594
80 0.73291841
};\label{plots:n-ab}

\addplot [color=mycolor1, dashdotted, line width=2.0pt]
table{
30 0.75740151
40 0.74343938
50 0.74419039
60 0.73898819
70 0.73957210
80 0.73052965
};\label{plots:n-sig}

\addplot [color=mycolor3, line width=3.0pt, mark=square, mark options={solid, mycolor3}]
table{
30 0.79282865
40 0.76938467
50 0.76705046
60 0.75837890
70 0.75498853
80 0.74447552
};\label{plots:n-log}

\addplot [color=mycolor5, line width=2.0pt, mark=o, mark options={solid, mycolor5}]
table{
30 0.77704235
40 0.75847973
50 0.75803174
60 0.75107315
70 0.74943187
80 0.73963026
};\label{plots:n-race}

\end{groupplot}

\matrix[
    matrix of nodes,
    anchor=south,
    draw,
    inner sep=0.15em,
    column sep=0.3cm,
    row sep=0.15cm,
    font=\small,
] at ([yshift=2.0ex]current bounding box.north)
{
\ref{plots:m-sig} & \Large Smooth-GL \cite{kalofolias2016learn} & \ref{plots:m-ab} &\Large  Linear Approx. \\
\ref{plots:m-li} & \Large GL-NG $f(x)=x$ & \ref{plots:m-log} & \Large GL-NG $f(x)=\log(1+x)$ \\
\ref{plots:m-race} &\Large  GL-NG $g(x)=x(1-x)$ & {} & {} \\
};

\end{tikzpicture}}\vspace{-.2cm}
\end{center}
\caption{\update{Performance of TTGD applied to \eqref{eq:glfp_ng} on PA graphs. Left: AUC versus $M$ with $N=50$. Right: AUC versus $N$ with $M=10$. In both panels, $\lambda=10$.}}\vspace{-.3cm}\label{fig:MNfigu}
\end{figure}

We benchmark \eqref{eq:glfp_ng} against four methods: (i) Smooth-GL method  \cite{kalofolias2016learn}, (ii) approximate bilevel problem in \eqref{eq:glfp_ng-app} (referred as {\sf linear approx.}), (iii) QuadGame-GL method \cite{leng2020learning}, which learns network structure from network games model; (iv) SpecTemp method \cite{segarra2017network}, which uses the spectral signature of the covariance matrix.
For \eqref{eq:glfp_ng}, the network games \eqref{eq:lq_game} and \eqref{eq:rt_game} are specified with the normalized marginal benefit vector ${\bm b} = \max({\bm v}_1, 0) / \| \max \{ {\bm v}_1, 0 \} \|_1$, where ${\bm v}_1$ is the top eigenvector of the Euclidean distance matrix ${\bm D}$. The matrix ${\bm D}$ is given by $D_{ij} = \| {\bm x}_i^{\rm row} - {\bm x}_j^{\rm row} \|^2_2$. Additionally, we set ${\bm a} = {\bm 1}$ in \eqref{eq:rt_game}. 

For the linear-quadratic (LQ) game \eqref{eq:lq_game}, we consider two different interaction functions, namely $f(x) = x$ and $f(x) = \log(1 + x)$; for the race \& tournament (RT) game \eqref{eq:rt_game}, we consider the interaction function $g(x)=x(1-x)$. In our experiments, the bilevel optimization problem is tackled using TTGD with the following step sizes. For LQ games with \( f(x) = x \), \( f(x) = \log(1 + x) \), the step sizes are \( (\alpha, \gamma) = (0.5, 0.003) \), while for RT games with \( g(x) = x(1 - x) \), the step sizes are \( (\alpha, \gamma) = (0.1, 0.0005) \). The algorithm terminates after 700 iterations for \( f(x) = x \) and 195 iterations for \( f(x) = \log(1 + x) \). For \( g(x) = x(1 - x)  \), the algorithm terminates after 250 iterations.

Fig.~\ref{fig:cr_game} shows the performance of graph learning algorithms against the regularization parameter $\lambda$, averaged over 20 Monte-Carlo trials. Note that as $\lambda$ increases, the regularized graph learning objective will become more dependent on the network games prior and, according to Proposition~\ref{prop:kkt}, \eqref{eq:glfp_ng} tends to learn a graph topology with few hub nodes. 

We concentrate on the performance of \eqref{eq:glfp_ng}. For PA graphs with hub structures, we observe substantial improvement in AUC as $\lambda$ increases. As a control, for ER graphs that lack obvious hub structures, the improvement in AUC is only modest. The social welfare consistently improves for both graph types as $\lambda$ grows.
These results corroborate that of Proposition~\ref{prop:kkt}, which shows that with $M \ll N$, the functional regularizer in \eqref{eq:glfp_ng} promotes a hub structure and improves the graph learning performance. 
\update{We next examine how the performance varies with the sample size and the network size in Figure~\ref{fig:MNfigu}. The results show a consistent advantage of the proposed GL-NG model with $f(x) = x$: Its AUC improves steadily as $M$ increases, while it remains clearly superior to the competing methods as $N$ increases.}

Our second experiment considers synthetic graph signal observations while focusing on ${\cal G}$ given by the \texttt{Karate Club} graph, which consists of \( N = 34 \) nodes. We have \( M = 50 \) samples of smooth graph signals generated from the Gaussian Markov Random Field (GMRF) model with the precision matrix given by the graph Laplacian \cite{dong2016learning}. Our aim is to showcase the necessity of tackling the bilevel problem \eqref{eq:glfp_ng} using TTGD instead of the single-level optimization approximation \eqref{eq:glfp_ng-app}. Fig.~\ref{fig:nume-pareto} shows the \emph{Pareto fronts} of the bilevel solution and the approximate solution, computed by varying the regularization parameter \( \lambda \) that trades off between the smooth-GL objective $J(\GSO; \bm{X})$ and ${\sf Wel}( \GSO )$. As expected, we observe that the TTGD algorithm achieves a better Pareto front than the approximate solution in all cases.

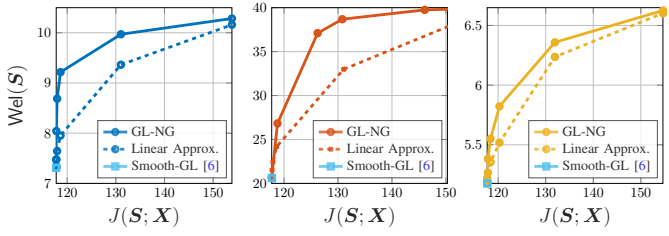
\begin{figure}[t]
\centering
\resizebox{0.99\linewidth}{!}{\definecolor{mycolor1}{rgb}{0.00000,0.44700,0.74100}%
\definecolor{mycolor2}{rgb}{0.85000,0.32500,0.09800}%
\definecolor{mycolor3}{rgb}{0.92900,0.69400,0.12500}%

\definecolor{mycolor6}{rgb}{0.30100,0.74500,0.93300}%

\begin{tikzpicture}

\begin{groupplot}[group style={group name=myplot,group size=3 by 1}]
\nextgroupplot[%
width=6cm,
height=6cm,
xmin=117.749304433409,
xmax=153.856055193291,
xlabel style={font=\color{white!15!black}},
xlabel={\Large $J( \GSO; {\bm X} )$},
ymin=7,
ymax=10.5,
ylabel style={font=\color{white!15!black}},
ylabel={\Large ${\sf Wel}(\GSO)$},
axis background/.style={fill=white},
xmajorgrids,
ymajorgrids,
legend pos = south east,
legend style={  
fill opacity=0.8,
draw opacity=0.8,
text opacity=0.8,
legend cell align=left, 
align=left, 
draw=white!15!black}]
\addplot [color=mycolor1, line width=2.0pt, mark=o, mark options={solid, mycolor1}]
  table[row sep=crcr]{%
117.755985796927	7.31405422792078\\
117.749304433409	7.47184527708778\\
117.782104858476	8.04305794038479\\
117.940064174086	8.68789807765267\\
118.636172592594	9.21653338208747\\
131.052579911335	9.97051561007478\\
153.856055193291	10.284132671904\\
};
\addlegendentry{GL-NG}

\addplot [color=mycolor1, dashed, line width=2.0pt, mark=o, mark options={solid, mycolor1}]
  table[row sep=crcr]{%
117.755985796927	7.31230134894373\\
117.749304433409	7.34558697423854\\
117.782104858476	7.47979454160738\\
117.940064174086	7.64365280256893\\
118.636172592594	7.95344765258109\\
131.052579911335	9.36391026366383\\
153.856055193291	10.1600951608408\\
};
\addlegendentry{Linear Approx.}

\addplot[
    color=mycolor6,
    line width=2.0pt,
    mark=square,
    mark options={solid, mycolor6}
    ]
    coordinates {
    (117.755985796927,7.31230134894373)
    };
    \addlegendentry{Smooth-GL \cite{kalofolias2016learn}}

\nextgroupplot[%
width=6cm,
height=6cm,
xmin=117.710637041695,
xmax=150.279764294101,
xlabel style={font=\color{white!15!black}},
xlabel={\Large $J( \GSO; {\bm X} )$},
ymin=20,
ymax=40,
ylabel style={font=\color{white!15!black}},
axis background/.style={fill=white},
xmajorgrids,
ymajorgrids,
legend pos = south east,
legend style={  
fill opacity=0.8,
draw opacity=0.8,
text opacity=0.8,
legend cell align=left, 
align=left, 
draw=white!15!black}
]
\addplot [color=mycolor2, line width=2.0pt, mark=o, mark options={solid, mycolor2}]
  table[row sep=crcr]{%
117.710637041695	20.6718312723392\\
118.788928778172	26.8324380236931\\
126.222203303195	37.0975231907485\\
130.790878925626	38.6854207821264\\
146.073861920048	39.751973328749\\
161.279764294101	39.9989700303178\\
};
\addlegendentry{GL-NG}

\addplot [color=mycolor2, dashed, line width=2.0pt, mark=x, mark options={solid, mycolor2}]
  table[row sep=crcr]{%
117.755985796927	20.6200351322813\\
117.749304433409	20.7980319136741\\
117.782104858476	21.5245812471704\\
117.940064174086	22.4298416785578\\
118.636172592594	24.1909329777626\\
131.052579911335	32.9902296181939\\
153.856055193291	38.7289883933107\\
};
\addlegendentry{Linear Approx.}

\addplot[
    color=mycolor6,
    line width=2.0pt,
    mark=square,
    mark options={solid, mycolor6}
    ]
    coordinates {
    (117.755985796927,20.6200351322813)
    };
    \addlegendentry{Smooth-GL \cite{kalofolias2016learn}}

\nextgroupplot[%
xmin=117.699225517511,
xmax=154.645302456487,
width=6cm,
height=6cm,
xlabel style={font=\color{white!15!black}},
xlabel={\Large $J( \GSO; {\bm X} )$},
ymin=5.17636351159447,
ymax=6.64912585436806,
ylabel style={font=\color{white!15!black}},
axis background/.style={fill=white},
xmajorgrids,
ymajorgrids,
legend pos = south east,
legend style={  
fill opacity=0.8,
draw opacity=0.8,
text opacity=0.8,
legend cell align=left, 
align=left, 
draw=white!15!black}
]

\addplot [color=mycolor3, line width=2.0pt, mark=o, mark options={solid, mycolor3}]
  table[row sep=crcr]{%
117.713314851486	5.17664023766534\\
117.699225517511	5.22094516880248\\
117.846480453781	5.38203878775381\\
118.398774222486	5.55012413640242\\
120.280830079352	5.82251725482614\\
131.925769108054	6.35863409077094\\
154.645302456487	6.62554223187532\\
};
\addlegendentry{GL-NG}

\addplot [color=mycolor3, dashed, line width=2.0pt, mark=o, mark options={solid, mycolor3}]
  table[row sep=crcr]{%
117.713314851486	5.17636351159447\\
117.699225517511	5.19437028190476\\
117.846480453781	5.26655287817175\\
118.398774222486	5.35405208825764\\
120.280830079352	5.51792715860902\\
131.925769108054	6.23638959039434\\
154.645302456487	6.60241004323305\\
};
\addlegendentry{Linear Approx.}

\addplot[
    color=mycolor6,
    line width=2.0pt,
    mark=square,
    mark options={solid, mycolor6}
    ]
    coordinates {
    (117.713314851486,5.17636351159447)
    };
    \addlegendentry{Smooth-GL \cite{kalofolias2016learn}}
\end{groupplot}
\end{tikzpicture}%
}\vspace{-.1cm}
 \caption{Comparing the social welfare ${\sf Wel}(\GSO) = {\bf 1}^\top {\bm y}^{\sf NE}(\GSO)$ against the data fidelity term $J( \GSO; {\bm X})$ for \eqref{eq:glfp_ng}. Left: Prior with $f(x) = \log(1+x)$. Middle: Prior with $f(x) = x$. Right: Prior with $g(x) = x(1-x)$. } \vspace{-.3cm} \label{fig:nume-pareto}
\end{figure}

\vspace{.1cm}
\noindent $\bullet$ {\bf Real Data: Case Study I.}
We evaluate our proposed TTGD algorithm on two real-world social network datasets in the limited data setting, characterized by $M\ll N$.

The 
\href{https://www.stats.ox.ac.uk/~snijders/siena/vDuijn_data.htm}{\tt Groningen} dataset contains $M = 13$ signals taken from a friendship network of students in University of Groningen with $N = 38$ nodes in 1996. The signals capture each student's level of interest in social activities (e.g., attending concerts and movies). We set $\beta = 10$ and define $\mathbf{b}$ to encode the study program, where $\text{Program} = 1$ corresponds to the 4-year track and $\text{Program} = 2$ to the 2-year track. Additionally, we normalize $\mathbf{b}$ as ${\mathbf{b} = \mathbf{b} / \mathbf{1}^\top \mathbf{b}}$.
The \href{https://www.stats.ox.ac.uk/~snijders/siena/tutorial2010_data.htm}{\tt Dutch} dataset comprises $N = 26$ nodes and $m = 7$ signals, representing the friendship network of teenagers in a Dutch school in 2003--2004. The signals measure negative behavior, such as the frequency of stealing or alcohol consumption. Again, we set $\beta = 10$ and let $\mathbf{b}$ encode ethnicity, where $\text{Ethnicity}=1$ is Dutch and $\text{Ethnicity}=2$ otherwise. We then normalize $\mathbf{b}$ accordingly.

Table~\ref{tab:perf-gron_dutch} presents the highest AUC values by tuning $\lambda$ for our \eqref{eq:glfp_ng}-based methods, alongside the corresponding social welfare values. We observe that in both datasets, \eqref{eq:glfp_ng}-based methods consistently outperform other benchmark methods in terms of AUC and welfare. Specifically, for {\tt Groningen}, the LQ game induced prior with the interaction $f(x)=x$ achieves the best AUC performance. 
For {\tt Dutch}, the RT game induced prior with the interaction $g(x) = x(1-x)$ gives the best AUC performance. 
Note that we anticipate different real-world network data may inherently have different interaction dynamics. Nevertheless, our \eqref{eq:GL_bilevel:UNI} framework improves performance by incorporating functional priors. 


\begin{table}[t]
\centering
\resizebox{.999\linewidth}{!}{
\begin{tabular}{lcccc}
\hline
\multicolumn{1}{c}{} & \multicolumn{1}{|c|}{\multirow{2}{*}{{\sf AUC}}} & \multicolumn{3}{c}{${\sf Wel}(\GSO)$ by}                  \\ 

\multicolumn{1}{c}{} & \multicolumn{1}{|c|}{} & \multicolumn{1}{c}{$ x$} & \multicolumn{1}{c}{$\log(1+x)$} & \multicolumn{1}{c}{$ x(1-x)$} \\ \hline
   \multicolumn{5}{c}{{\tt Groningen} dataset}                                                                                                                     \\ \hline
 \multicolumn{1}{l|}{ GL-NG ($f(x)= x$)}  &   \multicolumn{1}{c|}{{\bf 0.6357}}                     &{\bf 27.7721}     &{\bf 8.7035} &5.9668      \\ 
 \multicolumn{1}{l|}{ GL-NG ($ f(x)=\log(1+x)$)}  &    \multicolumn{1}{c|}{0.6207}                  &26.5906     &8.6225 &5.9557       \\ 
 \multicolumn{1}{l|}{ GL-NG ($g(x)=x(1-x)$)}  &   \multicolumn{1}{c|}{0.6201}                      &26.5722     &8.6432 &{\bf 5.9739}  \\ 
\hline 
 \multicolumn{1}{l|}{Linear Approx} &\multicolumn{1}{c|}{0.5354}  &20.0577&7.5192 &5.4005 \\ 
 \multicolumn{1}{l|}{Smooth-GL \cite{kalofolias2016learn}} &     \multicolumn{1}{c|}{0.5354}          &20.0566     &7.5190 &5.4003    \\ 
 \multicolumn{1}{l|}{QuadGame-GL \cite{leng2020learning}} &    \multicolumn{1}{c|}{0.5490 }           &14.5312     &5.7325 &4.2091   \\ 
 \multicolumn{1}{l|}{SpecTemp \cite{segarra2017network}} & \multicolumn{1}{c|}{0.4890}                  &19.7591     &7.4576 &5.3658  \\
 \multicolumn{1}{l|}{Glasso \cite{friedman2008sparse}} & \multicolumn{1}{c|}{0.5617}                  &16.6454     &6.4513 &4.6939  \\
  \hline\hline
   \multicolumn{5}{c}{{\tt Dutch} dataset}  
                                                                  \\ \hline
    \multicolumn{1}{l|}{ GL-NG ($f(x)=x$)}  &   \multicolumn{1}{c|}{0.6127}                     &21.3681     &6.7664     &4.6940      \\ 
 \multicolumn{1}{l|}{ GL-NG ($f(x)=\log(1+x)$)}  &    \multicolumn{1}{c|}{0.6159}                  &22.5223     &6.9465  &4.7810        \\ 
 \multicolumn{1}{l|}{ GL-NG ($ g(x)=x(1-x)$)}  &   \multicolumn{1}{c|}{{\bf0.6186}}             &{\bf 23.0427}    &{\bf7.0260}  &{\bf4.8187}   \\ 
\hline 
 \multicolumn{1}{l|}{Linear Approx} &\multicolumn{1}{c|}{{\bf0.6186}}  &21.4615&6.7831  &4.7030   \\ 
 \multicolumn{1}{l|}{Smooth-GL \cite{kalofolias2016learn}} &     \multicolumn{1}{c|}{0.5840}          &18.6902     &6.3571  &4.6357    \\ 
 \multicolumn{1}{l|}{QuadGame-GL \cite{leng2020learning}} &    \multicolumn{1}{c|}{0.6048 }           &18.3942     &5.9083  &4.1413   \\ 
 \multicolumn{1}{l|}{SpecTemp \cite{segarra2017network}} & \multicolumn{1}{c|}{0.4828}           &20.6126         &6.6480     &4.6357   \\
 \multicolumn{1}{l|}{Glasso \cite{friedman2008sparse}} & \multicolumn{1}{c|}{0.5786}                  &18.1852     &5.3394  &2.7934   \\
  \hline
\end{tabular}}
\caption{Comparing the AUC and welfare of different graph learning models on {\tt Groningen} and {\tt Dutch}. } \label{tab:perf-gron_dutch}
\end{table}       

\begin{figure}
    \centering
    \includegraphics[width=0.95\linewidth]{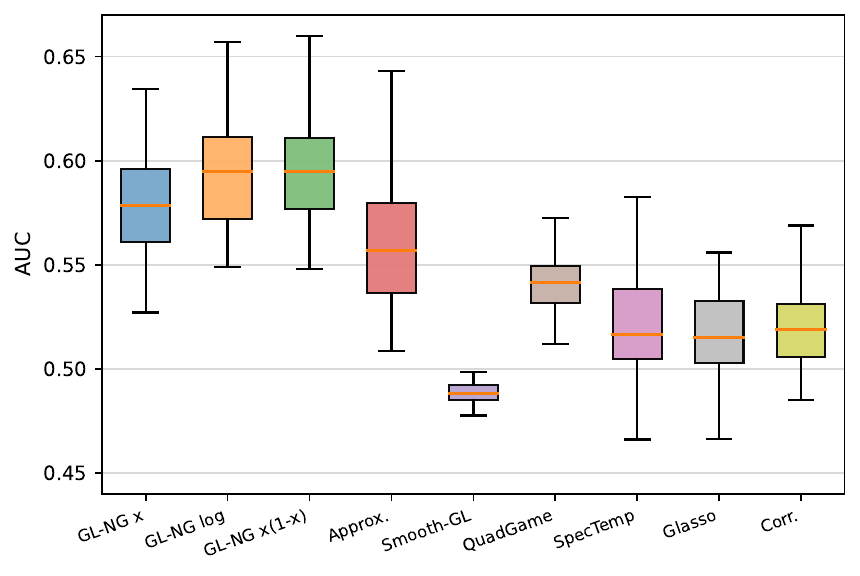}
    \caption{Comparing the AUC performance of the graph learnt from {\tt IndianVillage}\cite{banerjee2013diffusion}. }
    \label{fig:vil_auroc}
\end{figure}


\begin{table}[t]
\centering
\begin{tabular}{lccc}
\toprule
${\sf Wel}( \hat{\GSO} ) - {\sf Wel}( {\GSO}^{\sf true} )$ & \textbf{Maximum} & \textbf{Average} & \textbf{Minimum} \\ 
\midrule
GL-NG ($f(x) = x$)         &   \bfseries 4.2754              &   {\bf 3.2077 }              &  \bfseries 2.1693              \\ 
Linear Approx.  &    -0.4790          &    -1.4241             &     -2.3927            \\ 
Smooth-GL \cite{kalofolias2016learn} &     -2.3288            &     -5.8141            &     -8.6115            \\ 
QuadGame-GL \cite{leng2020learning} &     -1.9386            &     -3.5631            &     -5.5949            \\ 
SpectTemp \cite{segarra2017network}  &     0.2653            &     -2.8122            &     -6.6532            \\ 
Glasso\cite{friedman2008sparse} & -0.7666 & -4.1383  & -10.3377 \\
Correlation & -2.6829  & -3.3693  & -4.4143\\
\midrule
GL-NG ($f(x) = \log(1+x)$)         &   \bfseries 1.4464              &  {\bf 1.1892 }             &  \bfseries 0.8461              \\ 
Linear Approx. &   -0.0676              &    -0.6144              &  -1.1429               \\ 
Smooth-GL \cite{kalofolias2016learn}  &   -0.5440              &   -2.8292              &  -5.0419               \\ 
QuadGame-GL \cite{leng2020learning}  &   -0.8206              &   -1.6814              &  -3.3074                \\ 
SpectTemp \cite{segarra2017network} &     0.1453            &     -1.2226            &     -3.6324            \\
Glasso\cite{friedman2008sparse} & -0.3700 & -2.0347  &  -4.9003 \\
Correlation &  -0.7413  & -1.4614  & -1.9076
\\
\midrule
GL-NG ($g(x) = x(1-x)$)         &   \bfseries 1.0312              &  {\bf 0.7769 }             &  \bfseries 0.5476              \\ 
Linear Approx. &    -0.0292              &    -0.4454              &  -0.9317               \\ 
Smooth-GL \cite{kalofolias2016learn}  &   -0.4165              &    -1.0737              &  -1.7126               \\ 
QuadGame-GL \cite{leng2020learning}  &   -0.5334              &   -1.2174              &  -2.5642                \\ 
SpectTemp \cite{segarra2017network} &     0.1147            &     -0.8539            &      -2.6901            \\
Glasso\cite{friedman2008sparse} & -0.2443 & -1.4959   &  -3.7367 \\
Correlation &   -0.4475  &  -1.0201   & -1.3848\\
  \bottomrule
\end{tabular}
\caption{Gain in social welfare of the graphs learned from the {\tt IndianVillage} data \cite{banerjee2013diffusion}. The three row blocks report welfare gains evaluated under the interaction functions $f(x)=x$, $f(x)=\log(1+x)$, and $g(x)=x(1-x)$, respectively.}
\label{tab:wel}
\end{table}

\vspace{.1cm}
\noindent $\bullet$ {\bf Real Data: Case Study II.}
Our second set of real data experiments considers learning the graph topologies from a set of real data taken from Indian villages \cite{banerjee2013diffusion}. The dataset consists of survey data from $40$ villages, each taken as an individual graph to be learned by \eqref{eq:glfp_ng}. Here, the networks have sizes ranging from $N=77$ to $N=330$ agents, and $M=16$ samples of graph signals are observed for each network. 

We set the parameters in \eqref{eq:glfp_ng} with $\lambda = 100$, \(\beta = 50\), \({\bm b} = {\tt h} + 1\), where \( {\tt h}_i \in \{0,1\} \) indicates if the agent is a potential microfinance client (cf.~`{\tt hhSurveyed}' in the dataset) and ${\bm b}$ is normalized such that ${\bm b}^\top {\bm 1}=1$. 
Figure~\ref{fig:vil_auroc} shows the boxplots of AUC values for graph topologies learned under different settings and algorithms, compared against the ground-truth network reported in \cite{banerjee2013diffusion}. We observe that the solutions found by \eqref{eq:glfp_ng} consistently provides accurate graph topology estimation against the benchmarks.
\update{Similarly, Table III reports the maximum/average/minimum gain in social welfare relative to the ground truth. We observe that the graphs learned by \eqref{eq:glfp_ng} attain higher welfare values under the prescribed network-game model, which is consistent with the role of the functional prior in steering the estimator toward graph structures that reach high values on the target task. A natural network may settle for a compromise among multiple competing objectives, whereas our regularized estimator emphasizes a prescribed functional aspect. It is therefore possible for the learned graph to attain a higher task-related score; see also Section \ref{sec:structural_interpretation}. The structural recovery quality is evaluated separately through AUC.}

\subsection{Gene Regulatory Dynamics} \label{sec:num-gene}
In this subsection, we focus on the case with priors induced by gene regulatory dynamics in \eqref{eq:glfp_gene}.

\vspace{.1cm}
\noindent $\bullet$ {\bf Synthetic Data (DREAM4).}
The first experiment considers the {\it In-Silico} (i.e., synthetic) data from the  \href{https://gnw.sourceforge.net/dreamchallenge.html#dream4challenge}{DREAM4 Challenge}. The DREAM4 networks are sub-networks of curated transcriptional regulatory networks from E. coli (RegulonDB 6.0 \cite{gama2008regulondb}) and S. cerevisiae \cite{balaji2006comprehensive}, with all self-loops removed.
The dataset was used in the DREAM4 network inference challenge and contains \(N = 100\) genes and $M=100$ perturbation experiments. Note that the effect of knockouts in GRNs is complex and the data-fitting loss \eqref{eq:data_ridge} is a simplification. This experiment pertains to the case with imprecise data models.  In the following, we consider \eqref{eq:glfp_gene} with the parameters \( a = 240\), \(\beta = 10, \sigma = 1\).

We also report the level of resilience for each learned graph topology subject to random edge removal. We simulate \(n_{\sf per} = 200\) perturbations by randomly deleting edges from the learned graph with {probability \(p \in \{ 0, 0.1 \} \)}, producing a perturbed topology set \(\hat{\mathcal{S}}\). We then compute the expected resilience metric as $\EE[ \Res( \GSO ) ] = n_{\sf per}^{-1} \sum_{ \GSO \in \hat{\mathcal{S}}} \Res( \GSO )$.

Fig.~\ref{fig:gene_dream4} reports the AUC together with the corresponding resilience values of the graph topology found by \eqref{eq:glfp_gene} as \(\lambda\) varies within $[0,500]$, which illustrates the effect of random edge removal on resilience.
\update{Recall that under the fixed-total-weight constraint ${\bm 1}^\top \GSO {\bm 1}=a$, the topology-level resilience proxy in \cite{gao2016universal} reduces to the same degree-based score underlying \eqref{eq:glfp_gene_approx}. Hence, the empirical performance of \eqref{eq:glfp_gene_approx} serves as an indirect evaluation of the topology-level mechanism highlighted in \cite{gao2016universal}.} As a benchmark, we compare the performance of \eqref{eq:glfp_gene_approx}, the ridge regression \cite{tjarnberg2013optimal, hillerton2022fast}, and the resilience of the ground truth topology. \update{We additionally tested Correlation and SmoothGL in the same perturbation-response setting; their DREAM4 AUC values are only 0.2363 and 0.2359, respectively. As expected, network-game-based baselines perform substantially worse in this GRN setting, which confirms that the issue is not simply the choice of baseline, but the mismatch between their modeling assumptions and the perturbation-response nature of the GRN data.} For easier comparison, we use a scaled regularization parameter in~\eqref{eq:glfp_gene_approx} that is set as \( \hat{\lambda} =  \lambda /1600 \).
We observe from the figure that both the proposed GL-GENE method and its approximation lead to improved AUC and robustness of the learned graphs. In particular, the bilevel-based method achieves higher AUC and robustness, demonstrating the effectiveness of incorporating bilevel optimization into the inference process.

\begin{figure}[t]
\centering
\resizebox{0.97\linewidth}{!}{\definecolor{mycolor1}{rgb}{0.00000,0.44700,0.74100}%
\definecolor{mycolor2}{rgb}{0.85000,0.32500,0.09800}%
\definecolor{mycolor3}{rgb}{0.92900,0.69400,0.12500}%
\definecolor{mycolor4}{rgb}{0.49400,0.18400,0.55600}%
\definecolor{mycolor5}{rgb}{0.46600,0.67400,0.18800}%
\definecolor{mycolor6}{rgb}{0.35000,0.35000,0.35000}%
\begin{tikzpicture}
\begin{groupplot}[
    group style={group name=myplot, group size=3 by 1, horizontal sep=2cm}]
\nextgroupplot[%
width=6.5cm,
height=6cm,
xmin=0,
xmax=530,
xlabel style={font=\color{white!15!black}},
xlabel={\Large Reg.~param.~$\lambda$},
ymin=0.49,
ymax=0.68,
ylabel style={font=\color{white!15!black}},
ylabel={\Large AUC},
tick label style={font=\Large},
axis background/.style={fill=white},
xmajorgrids,
ymajorgrids,
legend style={legend cell align=left, align=left, draw=white!15!black}
]
\addplot [color=mycolor1, line width=1.5pt, mark=o, mark options={solid, mycolor1}]
  table[row sep=crcr]{%
0	0.650530009994077\\
50	0.650530009994077\\
100	0.650530009994077\\
150	0.650530009994077\\
200	0.650530009994077\\
300	0.650530009994077\\
450	0.650530009994077\\
490	0.650530009994077\\
530	0.650530009994077\\
};

\addplot [color=mycolor2, line width=1.5pt, mark=asterisk, mark options={solid, mycolor2}]
  table[row sep=crcr]{%
0	0.650530009994077\\
50	0.653698271394728\\
100	0.664262034535091\\
150	0.66586178098164\\
200	0.670787681374001\\
300	0.670954827694699\\
450	0.668507493244744\\
490	0.658554190109565\\
530	0.644682202213503\\
};

\addplot [color=mycolor3, line width=1.5pt, mark=square, mark options={solid, mycolor3}]
  table[row sep=crcr]{%
0	0.650530009994077\\
50	0.650804153094462\\
100	0.650753546509475\\
150	0.653173119632809\\
200	0.655820277798341\\
300	0.661905213577139\\
450	0.653879298378738\\
490	0.636110314258217\\
530	0.618029015213207\\
};

\addplot [color=black, dashed, line width=2.0pt]
  table[row sep=crcr]{%
0	0.5259\\
530	0.5259\\
};\label{plots:specttemp}

\addplot [color=mycolor6, line width=1.5pt, mark=pentagon*, mark options={solid, mycolor6}]
  table[row sep=crcr]{%
0	0.5000\\
530	0.5000\\
};\label{plots:glasso}
\coordinate (top) at (rel axis cs:1,0);
\nextgroupplot[%
width=6.5cm,
height=6cm,
xmin=0,
xmax=530,
xlabel style={font=\color{white!15!black}},
xlabel={\Large Reg.~param.~$\lambda$},
ymin=0,
ymax=40,
tick label style={font=\Large},
ylabel style={font=\color{white!15!black}},
ylabel={\Large $\mathbb{E}[\Res(\hat{\GSO})]/2$},
axis background/.style={fill=white},
xmajorgrids,
ymajorgrids,
legend style={  
fill opacity=0.8,
draw opacity=0.8,
text opacity=0.8,
legend cell align=left, 
align=left, 
at={(0.7,0.7)}, anchor=north,
draw=white!15!black}
]

\addplot [color=mycolor2, line width=1.5pt, dashed,mark=asterisk, mark options={solid, mycolor2}]
  table[row sep=crcr]{%
0	0\\
50	10.9510166900749\\
100	33.3358510538971\\
150	35.213624938771\\
200	35.5913202914663\\
300	35.989354448435\\
450	36.3068703053459\\
490	36.4566257452604\\
530	36.4478508453985\\
};\label{plots:glfp_genep01}

\addplot [color=mycolor3, line width=1.5pt, mark=square, dashed,mark options={solid, mycolor3}]
  table[row sep=crcr]{%
0	0\\
50	0\\
100	0\\
150	5.82372320597432\\
200	30.3514421990222\\
300	35.5713553792475\\
450	36.96860663648\\
490	37.1526273471112\\
530	36.9123917765545\\
};\label{plots:glfp_gene_appp01}

\addplot [color=mycolor4, line width=1.5pt, mark=diamond, mark options={solid, mycolor4}]
  table[row sep=crcr]{%
0	12.5811562185299\\
50	12.5811562185299\\
100	12.5811562185299\\
150	12.5811562185299\\
200	12.5811562185299\\
300	12.5811562185299\\
450	12.5811562185299\\
490	12.5811562185299\\
530	12.5811562185299\\
};\label{plots:GroundTruth}

\addplot [color=mycolor4, line width=1.5pt,dashed,mark=diamond, mark options={solid, mycolor4}]
  table[row sep=crcr]{%
0	9.10842862002164\\
50	9.10842862002164\\
100	9.10842862002164\\
150	9.10842862002164\\
200	9.10842862002164\\
300	9.10842862002164\\
450	9.10842862002164\\
490	9.10842862002164\\
530	9.10842862002164\\
};\label{plots:GroundTruthp01}

\addplot [color=mycolor1, line width=1.5pt, mark=o, mark options={solid, mycolor1}]
  table[row sep=crcr]{%
0	4.56983333689941\\
50	4.56983333689941\\
100	4.56983333689941\\
150	4.56983333689941\\
200	4.56983333689941\\
300	4.56983333689941\\
450	4.56983333689941\\
490	4.56983333689941\\
530	4.56983333689941\\
};\label{plots:ridge}

\addplot [color=mycolor1, line width=1.5pt,dashed, mark=o, mark options={solid, mycolor1}]
  table[row sep=crcr]{%
0	0\\
50	0\\
100	0\\
150	0\\
200	0\\
300	0\\
450	0\\
490	0\\
530	0\\
};\label{plots:ridgep01}

\addplot [color=mycolor2, line width=1.5pt, mark=asterisk, mark options={solid, mycolor2}]
  table[row sep=crcr]{%
0	4.56982448087999\\
50	37.3803431776512\\
100	38.1226730151684\\
150	38.5266531514552\\
200	38.8159288937235\\
300	39.1640178497921\\
450	39.3620846626965\\
490	39.4389728449405\\
530	39.4739706017887\\
};\label{plots:glfp_gene}

\addplot [color=mycolor3, line width=1.5pt, mark=square, mark options={solid, mycolor3}]
  table[row sep=crcr]{%
0	4.56983333689941\\
50	34.4032867059916\\
100	36.4402899254636\\
150	37.2956879078939\\
200	37.8695668604226\\
300	38.7579170990473\\
450	39.7506005331859\\
490	39.94512656142\\
530	40.0172335442021\\
};\label{plots:glfp_gene_approx}

\coordinate (bot) at (rel axis cs:0,0);

\end{groupplot}

\path (top|-current bounding box.north)--
      coordinate(legendpos)
      (bot|-current bounding box.north);

\node[
    anchor=south,
    draw,
    inner sep=0.3em
  ]at([yshift=1.75ex,xshift=2ex]legendpos){%
    \Large
    \begin{tabular}{@{}ll@{\hspace{0.8em}}ll@{\hspace{0.8em}}ll@{}}
    \tikz[baseline=-0.6ex]{\draw[color=mycolor4, line width=1.5pt] (0,0) -- (0.45,0); \draw[color=mycolor4, line width=1.5pt, mark=diamond, mark options={solid, mycolor4}] plot coordinates {(0.225,0)};} & GroundTruth &
    \tikz[baseline=-0.6ex]{\draw[color=mycolor1, line width=1.5pt] (0,0) -- (0.45,0); \draw[color=mycolor1, line width=1.5pt, mark=o, mark options={solid, mycolor1}] plot coordinates {(0.225,0)};} & Ridge &
    \tikz[baseline=-0.6ex]{\draw[color=mycolor2, line width=1.5pt] (0,0) -- (0.45,0); \draw[color=mycolor2, line width=1.5pt, mark=asterisk, mark options={solid, mycolor2}] plot coordinates {(0.225,0)};} & GL-GENE \\
    \tikz[baseline=-0.6ex]{\draw[color=mycolor3, line width=1.5pt] (0,0) -- (0.45,0); \draw[color=mycolor3, line width=1.5pt, mark=square, mark options={solid, mycolor3}] plot coordinates {(0.225,0)};} & GL-GENE-App &
    \tikz[baseline=-0.6ex]{\draw[color=black,dashed,  line width=2.0pt] (0,0) -- (0.45,0); \draw[color=black, dashed, line width=2.0pt] plot coordinates {(0.225,0)};} & SpectTemp~\cite{segarra2017network} &
    \tikz[baseline=-0.6ex]{\draw[color=mycolor6, line width=1.5pt] (0,0) -- (0.45,0); \draw[color=mycolor6, line width=1.5pt, mark=pentagon*, mark options={solid, mycolor6}] plot coordinates {(0.225,0)};} & Glasso~\cite{friedman2008sparse}
    \end{tabular}
  };

\end{tikzpicture}%
}
 \caption{Performance of algorithms on DREAM4. Left: AUC. Right: Expected resilience metric $\mathbb{E}[\Res(\hat{\mathbf{S}})]$ across $\lambda$. Dashed line represents the perturbation set with $p=0.1$.}\vspace{-.2cm} \label{fig:gene_dream4}
\end{figure}
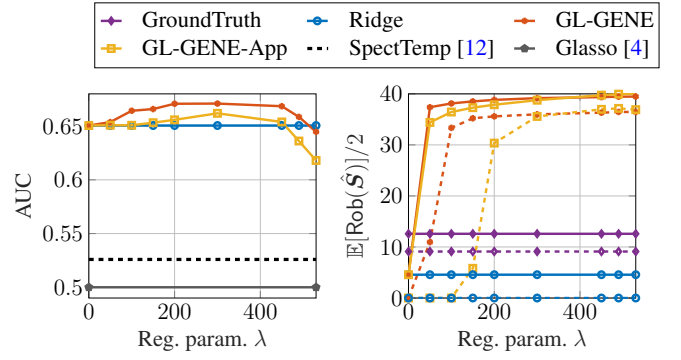



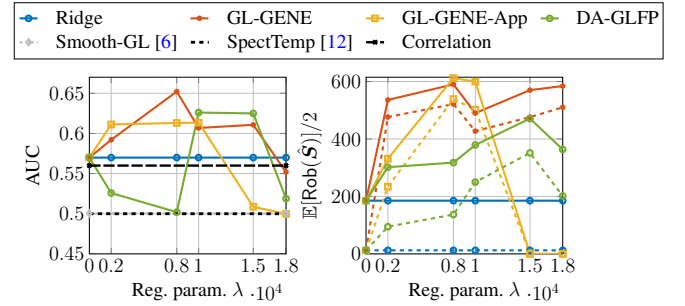
\begin{figure}[t]
\centering
\resizebox{0.97\linewidth}{!}{\definecolor{mycolor1}{rgb}{0.00000,0.44700,0.74100}%
\definecolor{mycolor2}{rgb}{0.85000,0.32500,0.09800}%
\definecolor{mycolor3}{rgb}{0.92900,0.69400,0.12500}%
\definecolor{mycolor4}{rgb}{0.75000,0.75000,0.75000}%
\definecolor{mycolor5}{rgb}{0.46600,0.67400,0.18800}%
\definecolor{mycolor6}{rgb}{0.35000,0.35000,0.35000}%
\definecolor{mycolor7}{rgb}{0.00000,0.00000,0.00000}%
\begin{tikzpicture}
\begin{groupplot}[
    group style={group name=myplot, group size=2 by 1, horizontal sep=2cm}]

\nextgroupplot[
width=6.5cm,
height=6cm,
    xlabel={\Large Reg.~param.~$\lambda$},
    ylabel={\Large AUC},
    xmin=0, xmax=18000,
    xtick={0,2000,8000,10000,15000,18000},
    ymin=0.45, ymax=0.67,
    grid=major,
    tick label style={font=\Large},
    legend style={at={(0.05,0.05)}, anchor=south west}
]

\addplot [color=mycolor1, line width=1.5pt, mark=o, mark options={solid, mycolor1}]
  table[row sep=crcr]{%
0	0.5699\\
2000	0.5699\\
8000	0.5699\\
10000	0.5699\\
15000	0.5699\\
18000	0.5699\\
};

\addplot [color=mycolor2, line width=1.5pt, mark=asterisk, mark options={solid, mycolor2}]
  table[row sep=crcr]{%
0	0.5699\\
2000	0.5921\\
8000	0.6521\\
10000	0.6068\\
15000	0.6107\\
18000	0.5518\\
};

\addplot [color=mycolor3, line width=1.5pt, mark=square, mark options={solid, mycolor3}]
  table[row sep=crcr]{%
0	0.5699\\
2000	0.6112\\
8000	0.6130\\
10000	0.6135\\
15000	0.5087\\
18000	0.4998\\
};

\addplot [color=mycolor5, line width=1.5pt, mark=o, mark options={solid, mycolor5}]
  table[row sep=crcr]{%
0	0.5699\\
2000	0.5258\\
8000	0.5020\\
10000	0.6260\\
15000	0.6249\\
18000	0.5189\\
};

\addplot [color=mycolor4, line width=1.7pt, dash pattern=on 2pt off 2pt, mark=diamond, mark options={solid, mycolor4}]
  table[row sep=crcr]{%
0     0.500411\\
18000 0.500411\\
};

\addplot [color=black, dashed, line width=2.0pt]
  table[row sep=crcr]{%
0     0.500000\\
18000 0.500000\\
};

\addplot [color=mycolor7, line width=1.7pt, dash pattern=on 8pt off 2pt, mark=x, mark options={solid, mycolor7}]
  table[row sep=crcr]{%
0     0.559969\\
18000 0.559969\\
};


\coordinate (top) at (rel axis cs:1,0);

\nextgroupplot[
width=6.5cm,
height=6cm,
    xlabel={\Large Reg.~param.~$\lambda$},
ylabel={\Large $\mathbb{E}[\Res(\hat{\GSO})]/2$},
tick label style={font=\Large},
xmin=0, xmax=18000,
    xtick={0,2000,8000,10000,15000,18000},
    ymin=0, ymax=615,
    grid=major,
    legend style={at={(0.05,0.78)}, anchor=north west}
]

\addplot [color=mycolor2, line width=1.5pt, dashed,mark=asterisk, mark options={solid, mycolor2}]
  table[row sep=crcr]{%
0	12.4423\\
2000	476.3185\\
8000	521.4932\\
10000	427.1965\\
15000	476.0562\\
18000	510.1151\\
};

\addplot [color=mycolor1, line width=1.5pt,dashed, mark=o, mark options={solid, mycolor1}]
  table[row sep=crcr]{%
0	12.4423\\
2000	12.4423\\
8000	12.4423\\
10000	12.4423\\
15000	12.4423\\
18000	12.4423\\
};\label{plots:ridgep01dream5}

\addplot [color=mycolor1, line width=1.5pt, mark=o, mark options={solid, mycolor1}]
  table[row sep=crcr]{%
0	185.2445\\
2000	185.2445\\
8000	185.2445\\
10000	185.2445\\
15000	185.2445\\
18000	185.2445\\
};

\addplot [color=mycolor3, dashed, line width=1.5pt, mark=square, mark options={solid, mycolor3}]
  table[row sep=crcr]{%
0	12.4423\\
2000	233.9634\\
8000	538.9781\\
10000	 501.6081\\
15000	0.2492 \\
18000	0.1421\\
};

\addplot [color=mycolor2, line width=1.5pt, mark=asterisk, mark options={solid, mycolor2}]
  table[row sep=crcr]{%
0	185.2445\\
2000	535.9090\\
8000	590.0341\\
10000	490.1043\\
15000	570.0314\\
18000	584.1791\\
};

\addplot [color=mycolor3, line width=1.5pt, mark=square, mark options={solid, mycolor3}]
  table[row sep=crcr]{%
0	185.2445\\
2000	331.2309\\
8000	611.9716\\
10000	599.5863\\
15000	0.3818\\
18000	0.1553\\
};

\addplot [color=mycolor5, line width=1.5pt, mark=o, mark options={solid, mycolor5}]
  table[row sep=crcr]{%
0	185.2445\\
2000	301.8251\\
8000	317.7742\\
10000	379.1488\\
15000	471.1130\\
18000	363.5787\\
};\label{plots:bi_diag0}

\addplot [color=mycolor5, line width=1.5pt,dashed, mark=o, mark options={solid, mycolor5}]
  table[row sep=crcr]{%
0	12.4423\\
2000	94.7898\\
8000	136.9334\\
10000	249.5878\\
15000	351.5053\\
18000	201.7476\\
};\label{plots:bi_diag1}

\coordinate (bot) at (rel axis cs:0,0);
\end{groupplot}

\path (top|-current bounding box.north)--
      coordinate(legendpos)
      (bot|-current bounding box.north);

\node[
    anchor=south,
    draw,
    inner sep=0.3em
  ]at([yshift=1.75ex,xshift=2ex]legendpos){%
    \Large
    \begin{tabular}{@{}ll@{\hspace{0.6em}}ll@{\hspace{0.6em}}ll@{\hspace{0.6em}}ll@{}}
    \tikz[baseline=-0.6ex]{\draw[color=mycolor1, line width=1.5pt] (0,0) -- (0.45,0); \draw[color=mycolor1, line width=1.5pt, mark=o, mark options={solid, mycolor1}] plot coordinates {(0.225,0)};} & Ridge &
    \tikz[baseline=-0.6ex]{\draw[color=mycolor2, line width=1.5pt] (0,0) -- (0.45,0); \draw[color=mycolor2, line width=1.5pt, mark=asterisk, mark options={solid, mycolor2}] plot coordinates {(0.225,0)};} & GL-GENE &
    \tikz[baseline=-0.6ex]{\draw[color=mycolor3, line width=1.5pt] (0,0) -- (0.45,0); \draw[color=mycolor3, line width=1.5pt, mark=square, mark options={solid, mycolor3}] plot coordinates {(0.225,0)};} & GL-GENE-App &
    \tikz[baseline=-0.6ex]{\draw[color=mycolor5, line width=1.5pt] (0,0) -- (0.45,0); \draw[color=mycolor5, line width=1.5pt, mark=o, mark options={solid, mycolor5}] plot coordinates {(0.225,0)};} & DA-GLFP \\
    \tikz[baseline=-0.6ex]{\draw[color=mycolor4, line width=1.7pt, dash pattern=on 2pt off 2pt] (0,0) -- (0.45,0); \draw[color=mycolor4, line width=1.7pt, mark=diamond, mark options={solid, mycolor4}] plot coordinates {(0.225,0)};} & Smooth-GL~\cite{kalofolias2016learn} &
      \tikz[baseline=-0.6ex]{\draw[color=black,dashed,  line width=2.0pt] (0,0) -- (0.45,0); \draw[color=black, dashed, line width=2.0pt] plot coordinates {(0.225,0)};} & SpectTemp~\cite{segarra2017network} &
    \tikz[baseline=-0.6ex]{\draw[color=mycolor7, line width=1.7pt, dash pattern=on 8pt off 2pt] (0,0) -- (0.45,0); \draw[color=mycolor7, line width=1.7pt, mark=x, mark options={solid, mycolor7}] plot coordinates {(0.225,0)};} & Correlation &
    & 
    \end{tabular}
  };

\end{tikzpicture}%
}
 \caption{Performance of algorithms on DREAM5. Left: AUC. Right: Expected resilience metric $\mathbb{E}[\Res(\hat{\mathbf{S}})]$ across $\lambda$. Dashed line represents the perturbation set with $p=0.1$.}\vspace{-.3cm}\label{fig:gene_dream5}
\end{figure}

\vspace{.1cm}
\noindent $\bullet$ {\bf Real Data (DREAM5 \textit{E. coli}).}
The second experiment considers the {\it In-Vivo} {(i.e., real)} data from the DREAM5 challenge \cite{marbach2012wisdom}. The dataset contains gene expression data from \textit{E. coli} under various perturbation conditions, with \(N = 4511\) genes and \(M = 56\) perturbation conditions. The steady-state expression levels are treated as the observed data, denoted by \(\boldsymbol{X} \in \mathbb{R}^{4511 \times 56}\).
We consider \eqref{eq:glfp_gene} with the parameters \( a = 3500, \beta = 1, \sigma = 100 \).
Note that with $N=4511$, computing the matrix inverse in TTGD can be computationally prohibitive. 
\update{Thus, for this experiment, we approximate the square matrix ${\sf J}_{\bm y} {\sf Y}$ in \eqref{eq:grad_ideal} by its diagonal and refer to this approach as DA-GLFP (cf.~Remark \ref{rem:ttgd_complexity}).}
The initialization of $\bm{y}$, $\GSO$ for \eqref{eq:glfp_gene_approx}, \eqref{eq:glfp_gene}, and DA-GLFP are all taken as the output of ridge regression. Each algorithm is run for 5000 iterations. On a laptop computer equipped with Apple M1 Pro,
it takes 93 minutes for \eqref{eq:glfp_gene}, while it only takes 44 minutes for ridge regression, 47 minutes for \eqref{eq:glfp_gene_approx}, and 48 minutes for DA-GLFP.
\update{We ignore GLASSO in this setup as the program did not finish within 4 hours.}


We evaluate the inference performance by comparing with the DREAM5 gold standard, which represents a subset of the gene regulatory network. \update{The gold standard is treated as an unweighted directed support and the inferred continuous edge weights are used directly as ranking scores.} All predicted edges are ranked by their inferred weights, and \update{following the DREAM5 protocol,} the top $10^5$ links are used to compute the AUC.
To assess the robustness of the graph topology learned, we report the $\Res(\GSO)$ values for each learned graph topology. In addition, we simulate $n_{\sf per} = 50$ perturbations by randomly deleting edges from the learned graph with probability $p = 0.1$, forming a perturbed set $\hat{\mathcal{S}}$.

Figure~\ref{fig:gene_dream5} presents the performance of the algorithms as $\lambda$ increases. The regularization parameter $\hat{\lambda}$ in~\eqref{eq:glfp_gene_approx} is scaled relative to that in~\eqref{eq:glfp_gene}, with $\hat{\lambda} = \lambda /240000$. 
Observe that the proposed GL-GENE method and its approximation outperform the vanilla ridge regression in terms of AUC for a wide range of $\lambda$. Among all methods, the bilevel-based approach \eqref{eq:glfp_gene} yields the best tradeoff between AUC and resilience. However, the resilience of the topology learned by the approximate method \eqref{eq:glfp_gene_approx} varies across runs. In comparison, although DA-GLFP performs slightly worse than \eqref{eq:glfp_gene}, its resilience is more consistent than that of \eqref{eq:glfp_gene_approx}. \update{This is consistent with the model mismatch explanation above: Off-the-shelf baselines designed for smooth graph signals or network-game-induced functional priors 
are not well aligned with perturbation-response GRN data, so their performance is substantially worse.}\vspace{-.2cm}


\section{Conclusions}
We introduce a new functional prior framework for graph topology learning. The proposed framework incorporates implicit regularizers, leading to a bilevel optimization formulation. We develop a TTGD algorithm with provable convergence guarantees under the regularity conditions stated in the paper. Our theoretical analysis and empirical results demonstrate that functional priors can induce desirable properties and outperform traditional structural priors.
\update{Our work also opens up two interesting directions for future study, namely, (i) to develop a convergence theory that covers settings where the lower-level dynamics admit multiple equilibria, and (ii) to establish consistency guarantees for the functional prior estimator.
}\vspace{-.2cm}

\allowdisplaybreaks
\appendices

\section{Proofs of Structural Interpretations}
\subsection{Proof of Proposition~\ref{prop:approx}} \label{pf:approx_glng}
As $\bm{y}^{\sf NE}(\GSO)$ is in the interior of ${\cal Y}$, it is also the solution to
$\bm{y} = \widetilde{\sf T}(\bm{y};\GSO)$.
Here, $\widetilde{\sf T}_i(\bm{y};\GSO) := \argmax_{y_i}U_i(y_i,\bm{y}_{-i};\GSO)$ is the unconstrained best response map. Consequently, the Lipschitzness and positivity of $f$ and $g$ imply that 
$\bm{0} \leq \bm{y}^{\sf NE}\leq \ell_1\GSO\bm{y}^{\sf NE} + \bm{b}$		
for any $\GSO\in {\cal S}_{\sf ng}$.
Since $\GSO \geq \bm{0}$ with row sums $\GSO\bm{1} = c\bm{1}$, we have $\|\GSO\|_\infty = c$. Taking the $\infty$-norm on both sides of the preceding componentwise inequality gives
\[
\|\bm{y}^{\sf NE}\|_\infty \leq c\ell_1 \|\bm{y}^{\sf NE}\|_\infty + \|\bm{b}\|_\infty,
\]
and hence $\|\bm{y}^{\sf NE}(\GSO)\|_\infty \leq \|\bm{b}\|_\infty / (1-c\ell_1)$. Using $\|\cdot\|_2 \leq \sqrt{N}\|\cdot\|_\infty$, we obtain
\begin{equation*}
    \|\bm{y}^{\sf NE}(\GSO)\|_2 \leq \frac{\sqrt{N}\,\|\bm{b}\|_\infty}{1-c\ell_1}.
\end{equation*}
Since $c \ell_1 < 1$, iterating the componentwise bound $\bm{y}^{\sf NE} \leq \ell_1\GSO\bm{y}^{\sf NE} + \bm{b}$ and using $(\GSO^k \bm{b})_i \leq c^k \|\bm{b}\|_\infty$ for $\GSO\in{\cal S}_{\sf ng}$, we have
\begin{align*}
	\bm{y}^\star&\leq (\bm{I}+ \ell_1\GSO)\bm{b} + \|\bm{b}\|_{\infty}\left((c \ell_1)^2 + (c \ell_1)^3 +\cdots \right)\bm{1}\\
	&= (\bm{I}+ \ell_1\GSO)\bm{b} + (1-c \ell_1)^{-1}(c \ell_1)^2\|\bm{b}\|_{\infty}\bm{1}.
\end{align*}
Hence, for any $\GSO\in{\cal S}_{\sf ng}$, 
\begin{align*}
   & \tr( \GSO^\top {\bm D} ) + \beta \| \GSO \|_F^2 - \lambda\bm{1}^\top\bm{y}^\star(\GSO) \\
   & \geq \tr( \GSO^\top {\bm D} ) + \beta \| \GSO \|_F^2 - \lambda \ell_1\bm{1}^\top\GSO\bm{b} - \lambda\bm{1}^\top\bar{\bm{c}}, 
\end{align*}
where $\bar{\bm c} = \frac{ (cl_1)^2}{1-cl_1} \| {\bm b} \|_\infty {\bf 1} + {\bm b}$. This proves the first part of \eqref{eq:bounds_glng}.

On the other hand, there exists a $\mu_1\leq \ell_1$ such that for any $\GSO\in{\cal S}_{\sf ng}$, we have
$\bm{y}^\star \geq \mu_1\GSO\bm{y}^\star + \bm{b}$.
Applying recursion gives
\begin{align*}
	\bm{y}^\star &\geq (\bm{I}+\mu_1\GSO)\bm{b} + b^\star\left((c\mu_1)^2 + (c\mu_1)^3 +\cdots\right)\bm{1} \\
	&= (\bm{I}+\mu_1\GSO)\bm{b} + (1-c\mu_1)^{-1}(c\mu_1)^2b^\star\bm{1}.
\end{align*}
Consequently, for any $\GSO\in{\cal S}_{\sf ng}$, 
\begin{align*}
 &\tr( \GSO^\top {\bm D} ) + \beta \| \GSO \|_F^2 - \lambda\bm{1}^\top\bm{y}^\star (\GSO) \\
 & \leq \tr( \GSO^\top {\bm D} ) + \beta \| \GSO \|_F^2 - \lambda \mu_1\bm{1}^\top\GSO\bm{b} - \lambda\bm{1}^\top\hat{\bm{c}},
\end{align*}
with $\hat{\bm c} = \frac{(c\mu_1)^2}{1-c\mu_1}b^\star {\bf 1} + {\bm b}$. This proves the second part of \eqref{eq:bounds_glng}.
\vspace{-.2cm}
\subsection{Proof of Proposition~\ref{prop:kkt}} \label{pf:kkt}
Introducing the dual variables $\lambda \in \RR, {\bm \eta} \in \mathbb{R}^N$, ${\bm h} \in \mathbb{R}^N$, $\boldsymbol{\mu} \in \mathbb{R}_+^{N \times N}$, the  Lagrangian of \eqref{eq:glfp_ng-app} is
\beq \notag
\begin{aligned}
& {\cal L}( \GSO, \lambda, \bm{\eta}, \bm{\mu}, {\bm h} )
= \tr ( \GSO^\top {\bm D} ) + \beta \| \GSO \|_F^2 - \tilde{\lambda} {\bm 1}^\top \GSO{\bm b} \\
  & \textstyle \quad - {\bm \eta}^\top(\GSO{\bm 1} - c{\bm 1}) - \langle \boldsymbol{\mu}, \GSO \rangle-\sum_{i=1}^N h_i{\bm e}_i^\top \GSO {\bm e}_i.
\end{aligned}
\eeq
The first order necessary condition of \eqref{eq:glfp_ng-app} yields
\vspace{-.1cm}
\[
\vspace{-.2cm}
{\bm D} + 2\beta \GSO - \tilde{\lambda} {\bm 1}{\bm b}^\top - \boldsymbol{\mu} - {\bm \eta}{\bm 1}^\top- {\rm Diag}( {\bm h} )
= {\bm 0},
\]
while the complementary slackness condition yields $S_{ij} \mu_{ij} = 0$. Together with $\GSO \geq {\bm 0}$, for $i \neq j$, we have two cases:
\vspace{-.1cm}

\noindent{\bf Case 1.} If $S_{ij} > 0$, then we must have $\mu_{ij} = 0$ and thus\vspace{-.2cm}
\[
    S_{ij} = (2\beta)^{-1} (- D_{ij} + \tilde{\lambda} b_j + \eta_i).
\] 
\noindent{\bf Case 2.} If $S_{ij} = 0$, then with $\mu_{ij} \geq 0$, we get \vspace{-.2cm}
\[
- D_{ij} + \tilde{\lambda} b_j + \eta_i \leq 0.
\]  \vspace{-.2cm}
Examining the two cases yields the stated expression for $S_{ij}^\star$.\vspace{-.2cm}
\subsection{Proof of Proposition \ref{prop:approx_GLgene}}
\vspace{-.1cm}
\label{pf:approx_glgene}
Observe that $y_i = \sum_{j=1, j \neq i}^N S_{ij} \frac{y_j^2}{y_j^2 + 1}$ for any $\GSO \in \mathcal{S}_{\sf nd}$ and $i \in V$. We first notice the elementary inequality
$( {1+e^{-\sigma y}} )^{-1} \leq (1/4) \sigma y + (1/2)$, which holds for any $y \geq 0$.
As ${\bm y}$ satisfies the constraints in \eqref{eq:glfp_gene}, we have
\vspace{-.2cm}
    \begin{align*} 
        &\Phi( \GSO; {\bm y} )  = \|\GSO \boldsymbol{X} + \boldsymbol{P}\|_F^2 + \beta \|\GSO\|_F^2 - \lambda{\bm 1}^\top\frac{{\bm 1}}{{\bm 1}+e^{-\sigma\y}} \notag \\
        & \geq \|\GSO \boldsymbol{X} + \boldsymbol{P}\|_F^2 + \beta \|\GSO\|_F^2 - {\sigma\lambda\bm{1}^\top\bm{y}}/{4} - {N\lambda}/{2}. \vspace{-.5cm}
    \end{align*}
For $y\geq 0$, as $\frac{y^2}{1+y^2} \leq \frac{1}{2}y$, we can infer from \eqref{eq:gene_LL_stat} that\vspace{-.2cm}
    \begin{align*}
        &y_i = \sum\nolimits_{j \neq i} S_{ij} \frac{y_j^2}{y_j^2 + 1} \leq \left( \frac{ \GSO\bm{y} }{2} \right)_i \\
        &= \sum\nolimits_{j \neq i} (\GSO^2)_{ij} \frac{y_j^2 / 2}{y_j^2 + 1} 
        \leq \left( \frac{ \GSO^2\bm{1} }{2} \right)_i.
    \end{align*}
Combining the above with the displayed lower bound proves the statement.\vspace{-.2cm}

\subsection{Proof of Proposition \ref{prop:gene-interpret}}\vspace{-.1cm}
\label{pf:1s21}
\update{
For any $\GSO \in {\cal S}_{\sf nd}$, define
\[
R_i := \sum_{j=1}^N S_{ij},
\qquad
C_i := \sum_{j=1}^N S_{ji},
\qquad i\in[N].
\]
Then, $\sum_{i=1}^N R_i = \sum_{i=1}^N C_i = a$, and we have
\begin{align*}
{\bm 1}^\top \GSO^2 {\bm 1}
&= \sum_{i,j,\ell} S_{i\ell} S_{\ell j}
 = \sum_{\ell=1}^N \Big(\sum_{i=1}^N S_{i\ell}\Big)\Big(\sum_{j=1}^N S_{\ell j}\Big)
 = \sum_{\ell=1}^N C_\ell R_\ell.
\end{align*}
Moreover, since ${\rm diag}(\GSO)={\bm 0}$ and $\GSO\ge {\bm 0}$, we get
\[
R_\ell + C_\ell
= \sum_{j\neq \ell}S_{\ell j} + \sum_{i\neq \ell}S_{i\ell}
\le \sum_{i,j}S_{ij}
= a,
\qquad \forall\, \ell\in[N].
\]
This yields
\begin{equation}\label{eq:per-node-bound}
C_\ell R_\ell
\le \frac{(C_\ell+R_\ell)^2}{4}
\le \frac{a(C_\ell+R_\ell)}{4},
\qquad \forall\, \ell\in[N].
\end{equation}
Summing over $\ell$ gives
\begin{equation}\label{eq:global-upper-bound}
{\bm 1}^\top \GSO^2 {\bm 1}
= \sum_{\ell=1}^N C_\ell R_\ell
\le \frac{a}{4}\sum_{\ell=1}^N (C_\ell+R_\ell)
= \frac{a^2}{2}.
\end{equation}
It follows that every feasible $\GSO$ satisfies ${\bm 1}^\top \GSO^2 {\bm 1}\le a^2/2$. Equality holds when
\[
S_{i^\star j^\star}=S_{j^\star i^\star}=\frac{a}{2},
\qquad
S_{ij}=0 \ \text{otherwise},
\qquad i^\star\neq j^\star,
\]
so the maximizer $\GSO^\star$ has value $a^2/2$. Furthermore, equality in \eqref{eq:global-upper-bound} implies equalities in \eqref{eq:per-node-bound} for each $\ell\in[N]$, i.e.,
\[
C_\ell R_\ell
= \frac{(C_\ell+R_\ell)^2}{4}
= \frac{a(C_\ell+R_\ell)}{4},
\qquad \forall\, \ell\in[N].
\]
Therefore, for each $\ell\in[N]$, we have
\[
C_\ell=R_\ell=0
\qquad \text{or} \qquad
C_\ell=R_\ell=\frac{a}{2}.
\]
Let $T:=\{\ell\in[N]:R_\ell=a/2\}$.
Since $\sum_{\ell=1}^N R_\ell=a$ and $R_\ell\in\{0,a/2\}$ for $\ell\in[N]$, we have $|T|=2$. Hence, we may write $T=\{i^\star,j^\star\}$ with $i^\star\neq j^\star$. For $\ell\notin T$, both the $\ell$-th row and $\ell$-th column of $\GSO^\star$ vanish. Hence, the only possible nonzero entries are $S^\star_{i^\star j^\star}$ and $S^\star_{j^\star i^\star}$, and the row-sum constraints give
\[
S^\star_{i^\star j^\star}
=
S^\star_{j^\star i^\star}
=
\frac{a}{2}.
\]
All other entries are zero. This proves the claim. 
}

{
\footnotesize
\bibliographystyle{IEEEtran}
\bibliography{ref}
}


\section{Proof of Theorem \ref{thm:ttgd}}\label{app:prof_conv_thm}
\subsection{Implications of H\ref{assu:cal_S} to H\ref{assu:Y_grd}}
We preface the proof of Theorem \ref{thm:ttgd} by showing that H\ref{assu:cal_S}-H\ref{assu:Y_grd} imply the following statements:
\begin{Prop}\label{prop:Lipschitz-ness_all_functions}
    Under H\ref{assu:cal_S} to H\ref{assu:Y_grd}, the following hold:
    \begin{enumerate}[start=1,label={(\bfseries S\arabic*)}]
        \item For any $\GSO\in{\cal S}$, there exists a unique solution $\bm{y}^\star(\GSO)$ to ${\sf Y}(\bm{y};\GSO) = \bm{0}$. \label{statement:ystar0}
        \item There exists a $\widetilde{B}>0$ such that $\|\bm{y}^\star(\GSO)\|_2 \leq \widetilde{B}$ for any $\GSO \in \mathcal{S}$. Furthermore, ${\bm y}^\star ( \GSO )$ is $L_y$-Lipschitz w.r.t.~$\GSO\in{\cal S}$ for some constant $L_y$.\label{statement:ystar}
        \item For any $B > 0$, there exists an $L_{\widehat{\Phi}} > 0$ such that the gradient map $\widehat{\grd}\Phi(\cdot)$ is $L_{\widehat{\Phi}}$-Lipschitz w.r.t. $(\GSO;\bm{y})$ over $\{(\GSO; \bm{y}) : \GSO\in{\cal S}, \bm{y}\in{\cal Y}, \|\bm{y}\|_2 \leq B\}$. \label{statement:hat-phi}
        \item The gradient $\grd\ell(\GSO)$ is $L_\ell$-Lipschitz w.r.t. $\GSO\in{\cal S}$.\label{statement:grd-ell}
    \end{enumerate}
\end{Prop} 
\begin{proof}
We divide the proof into several parts as follows. 

\vspace{.1cm}
\noindent \underline{\bf H\ref{assu:self-F} \& H\ref{assu:Y} $\Longrightarrow$ \ref{statement:ystar0}.} We show that ${\sf Y}(\cdot;\GSO)=\bm{0}$ admits a unique solution in ${\cal Y}$ for any $\GSO\in{\cal S}$. Observe that H\ref{assu:Y}, applied along the segment between any two points $\bm{y}_1,\bm{y}_2\in{\cal Y}$ and combined with a Taylor expansion, implies
    \begin{align}\label{eq:bound-JyY}
        & (1-\mody)\|\bm{v}\|_2^2 \leq \bm{v}^\top {\sf J}_{\bm{y}}{\sf Y}(\bm{y};\GSO)\bm{v}\ \ \forall\,\bm{v}\in\mathbb{R}^N, \notag \\
        & \|{\sf J}_{\bm{y}}{\sf Y}(\bm{y};\GSO)\|_2 \leq 1+\mody.
    \end{align}
for any $\bm{y}\in{\cal Y}$ and $\GSO\in{\cal S}$. The first inequality shows that the symmetric part of ${\sf J}_{\bm{y}}{\sf Y}(\bm{y};\GSO)$ is strictly positive definite, so ${\sf J}_{\bm{y}}{\sf Y}(\bm{y};\GSO)$ is invertible everywhere on ${\cal Y}\times{\cal S}$ and its smallest singular value is at least $1-\mody$. Combined with H\ref{assu:self-F} and the strong monotonicity of ${\sf Y}(\cdot;\GSO)$, this implies that ${\sf Y}(\cdot;\GSO)=\bm{0}$ has a unique solution in ${\cal Y}$ for any $\GSO\in{\cal S}$. This concludes \ref{statement:ystar0}.

\vspace{.1cm}
\noindent \underline{\bf H\ref{assu:cal_S} \& H\ref{assu:self-F} \& H\ref{assu:Y} \& H\ref{assu:Y_grd} $\Longrightarrow$ \ref{statement:ystar}.}  We know from \ref{statement:ystar0} that $\bm{y}^\star(\GSO)$ is well defined. Furthermore, ${\sf J}_{\bm y}{\sf Y} \neq \bm{0}$ as indicated by \eqref{eq:bound-JyY}. Thus, by the implicit function theorem, $\bm{y}^\star(\GSO)$ is a continuous function w.r.t. $\GSO\in{\cal S}$. Since ${\cal S}$ is a compact set, there exists a $\widetilde{B}$ such that $\|\bm{y}^\star(\GSO)\|_2 \leq\widetilde{B}$.

Next, we show the Lipschitzness of $\bm{y}^\star(\cdot)$ over ${\cal S}$. Taking $\bar{\bm{y}} = \bm{y}^\star(\GSO)$, the implicit function theorem gives
\[
{\sf J}_{\GSO}\bm{y}^\star(\GSO) = -{\sf J}_{\bm{y}}{\sf Y}(\bar{\bm{y}};\GSO)^{-1}{\sf J}_{\GSO}{\sf Y}(\bar{\bm{y}};\GSO).
\]
Our approach is to bound the terms in ${\sf J}_{\GSO}\bm{y}^\star(\cdot)$ respectively. As $\|\bm{y}^\star(\GSO)\|_2\leq\widetilde{B}$, H\ref{assu:Y_grd} indicates that for any $\GSO_0, \GSO\in{\cal S}$, 
    \begin{align}
        &\|{\sf J}_{\GSO}{\sf Y}(\bar{\bm{y}};\GSO)\|_F \notag \\
        & \leq \|{\sf J}_{\GSO}{\sf Y}(\bar{\bm{y}}^0;\GSO_0)\|_F + L_{\sf Y}\|\GSO - \GSO_0\|_F + L_{\sf Y}\|\bar{\bm{y}} - \bar{\bm{y}}^0\|_2 \notag \\
        & \leq \|{\sf J}_{\GSO}{\sf Y}(\bar{\bm{y}}^0;\GSO_0)\|_F + L_{\sf Y}D + 2L_{\sf Y}\widetilde{B}. \label{eq:bound-JSY}
    \end{align}
Here, $\bar{\bm{y}}^0:=\bm{y}^\star(\GSO_0)$ and $D$ is the diameter of the compact set ${\cal S}$. Since $\|{\sf J}_{\bm{y}}{\sf Y}(\bm{y};\GSO)^{-1}\|_2 \leq 1/(1-\mody)$ from \eqref{eq:bound-JyY}, together with \eqref{eq:bound-JSY}, we see that for any $\GSO\in{\cal S}$, 
    \begin{align*}
    \|{\sf J}_{\GSO}\bm{y}^\star(\GSO)\|_F 
    \leq &L_y:=\frac{\|{\sf J}_{\GSO}{\sf Y}(\bar{\bm{y}}^0;\GSO_0)\|_F + L_{\sf Y}D + 2L_{\sf Y}\widetilde{B}}{1-\mody}.
    \end{align*}
This shows that $\bm{y}^\star(\GSO)$ is $L_y$-Lipschitz w.r.t. $\GSO\in{\cal S}$.

\vspace{.1cm}
\noindent \underline{\bf H\ref{assu:cal_S} to H\ref{assu:Y_grd} $\Longrightarrow$ \ref{statement:hat-phi}.} For any $B > 0$, H\ref{assu:phi_grd} implies that given $\GSO_0\in{\cal S}$ and $\bm{y}_0\in{\cal Y}\cap\{\bm{y}:\|\bm{y}\|_2 \leq B\}$, for any $\GSO\in{\cal S}$ and $\bm{y}\in{\cal Y}\cap\{\bm{y}:\|\bm{y}\|_2 \leq B\}$, we have
\begin{align*}
\|\grd_{\bm y}\Phi(\GSO;\bm{y})\|_F \leq M_{\Phi} :=\|\grd\Phi(\GSO_0;\bm{y}_0)\|_F + (2B + D)L_{\Phi}.
\end{align*}
Similarly, H\ref{assu:Y_grd} implies that for any $(\GSO;\bm{y})$ and given $(\GSO_0;\bm{y}_0) \in \{(\GSO; \bm{y}) : \GSO\in{\cal S}, \bm{y}\in{\cal Y}, \|\bm{y}\|_2 \leq B\}$, we have
    \begin{align*}
        \|{\sf J}_{\GSO}{\sf Y}(\bm{y};\GSO)\|_F \leq M_{Y} := \|{\sf J}{\sf Y}(\bm{y}_0;\GSO_0)\|_F + (2B + D)L_{\sf Y}.
    \end{align*}

Together with \eqref{eq:bound-JyY} and the triangle inequality, we conclude that for any $(\GSO_1;\bm{y}_1),(\GSO_2;\bm{y}_2) \in \{(\GSO; \bm{y}) : \GSO\in{\cal S}, \bm{y}\in{\cal Y}, \|\bm{y}\|_2 \leq B\}$, 
\begin{align*}
    \|\widehat{\grd}\Phi(\GSO_1;\bm{y}_1) - \widehat{\grd}\Phi(\GSO_2;\bm{y}_2)\|_F
    \leq  L_{\widehat{\Phi}}\|(\GSO_1;\bm{y}_1) - (\GSO_2;\bm{y}_2)\|_F.
\end{align*}
Here, $L_{\widehat{\Phi}}:=L_{\Phi} + \frac{M_{\Phi}L_{\sf Y}(1+\mody) + M_{Y}L_{\Phi}(1+\mody) + M_{\Phi}M_YL_{\sf Y}}{(1-\mody)^2}$. This establishes \ref{statement:hat-phi}.

\vspace{.1cm}
\noindent \underline{\bf \ref{statement:ystar} \& \ref{statement:hat-phi} $\Longrightarrow $ \ref{statement:grd-ell}.}
Note that $\grd\ell(\GSO) = \widehat{\grd}\Phi(\GSO;\bm{y}^\star(\GSO))$. For any $\GSO_1,\GSO_2\in{\cal S}$, take $\bar{\bm{y}}^i = \bm{y}^\star(\GSO_i)$ for $i = 1,2$. Then,
\begin{align*}
        \|\grd\ell(\GSO_1) - \grd\ell(\GSO_2)\|_F
        & \leq L_{\widehat{\Phi}}\left(\|\bar{\bm{y}}^1-\bar{\bm{y}}^2\|_2 + \|\GSO_1 - \GSO_2\|_F\right) \\
        & \leq L_{\widehat{\Phi}}(L_y + 1)\|\GSO_1 - \GSO_2\|_F,
\end{align*}
where $L_{\widehat{\Phi}}$ is the Lipschitz constant in \ref{statement:hat-phi} when $B$ is taken as $\widetilde{B}$ from \ref{statement:ystar}. It follows that $\grd\ell(\GSO)$ is $L_\ell$-Lipschitz w.r.t. $\GSO\in{\cal S}$, where $L_\ell:=L_{\widehat{\Phi}}(L_y + 1)$.
\end{proof}

\subsection{Proof of Theorem \ref{thm:ttgd}}
Throughout, for the sake of brevity, we denote $\bar{\bm{y}}^k:=\bm{y}^\star(\GSO^k)$, $\mu_g:=1-\mody$, and $L_g:=1+\mody$.
We first characterize the progress of the lower-level and upper-level updates in Propositions \ref{prop:inequality-recursion1} and \ref{prop:inequality-recursion2}, respectively. 
\begin{Prop}\label{prop:inequality-recursion1}
    If the lower-level step size $\alpha$ satisfies $\alpha\leq\frac{ \mu_g }{(L_g)^2}$ and H\ref{assu:cal_S} to H\ref{assu:Y_grd} hold, we have 
    \begin{align}
            \|\bm{y}^{k+1}-\bar{\bm{y}}^k\|_2^2\leq &\left(1 - \frac{\alpha \mu_g }{2}\right)\|\bm{y}^k - \bar{\bm{y}}^{k-1}\|_2^2 \label{recursion-y}\\
            &+ L_y^2\left(\frac{2}{\alpha \mu_g} - 1\right)\|\GSO^{k-1} - \GSO^k\|_F^2,\notag
    \end{align}
    where $\mody$ and $L_y$ are given in H\ref{assu:Y} and \ref{statement:ystar}, respectively.
\end{Prop}
\begin{proof}

Since $\bm{y}^{k+1} = \bm{y}^{k} - \alpha {\sf Y}(\bm{y}^k;\GSO^k)$, we have
\begin{equation*}
    \begin{aligned}
        &\|\bm{y}^{k+1} - \bar{\bm{y}}^k\|_2^2 = \|\bm{y}^k - \bar{\bm{y}}^k - \alpha {\sf Y}(\bm{y}^k;\GSO^k)\|_2^2 \\
        & \leq (1 + (\alpha L_{g})^2)\|\bm{y}^k - \bar{\bm{y}}^k\|_2^2 \\
         &\quad~- 2\alpha \, \langle \bm{y}^k - \bar{\bm{y}}^k, {\sf Y}(\bm{y}^k;\GSO^k) - {\sf Y}(\bar{\bm{y}}^k;\GSO^k)\rangle \\
        & \leq (1 + (\alpha L_{g})^2 - 2\alpha\mu_g)\|\bm{y}^k - \bar{\bm{y}}^k\|_2^2.
    \end{aligned}
\end{equation*}
The last two inequalities are due to H\ref{assu:Y}.
With $\alpha \leq \frac{\mu_g}{(L_{g})^2}$, we can simplify the first term as
\begin{equation*}
    \begin{aligned}
       & (1 + (\alpha L_{g})^2 - 2\alpha\mu_g)\|\bm{y}^k - \bar{\bm{y}}^k\|_2^2 \leq (1 - \alpha\mu_g)\|\bm{y}^k - \bar{\bm{y}}^k\|_2^2 \\
       & \leq (1 - \alpha\mu_g)\left[ (1+z)\|\bm{y}^k - \bar{\bm{y}}^{k-1}\|_2^2 + \left(1+\frac{1}{z}\right)\|\bar{\bm{y}}^{k-1} - \bar{\bm{y}}^k\|_2^2\right],
    \end{aligned}
\end{equation*}
where $z>0$ is arbitrary. Using \ref{statement:ystar}, we bound the right-hand side as
\begin{equation*}
    \begin{aligned}
        &(1 - \alpha\mu_g) (1+z)\|\bm{y}^k - \bar{\bm{y}}^{k-1}\|_2^2 + L_y^2\left(1+\frac{1}{z}\right)\|\GSO^{k-1} - \GSO^k\|_F^2 \\
        & \leq \left(1 - \frac{ \alpha\mu_g }{2} \right)\|\bm{y}^k - \bar{\bm{y}}^{k-1}\|_2^2 + L_y^2\left(\frac{2}{\mu_g\alpha} - 1\right)\|\GSO^{k-1} - \GSO^k\|_F^2,
    \end{aligned}
\end{equation*}
where we set $z = \frac{\mu_g\alpha}{2(1-\mu_g\alpha)}$ in the last equality.
\end{proof} 

A direct corollary of the above proposition is that $\bm{y}^k$ is bounded, which can be proven by solving the recurrence \eqref{recursion-y} and using $\|\GSO^k - \GSO^{k-1}\|_F \leq D$. In particular, let
\beq 
B :=\widetilde{B} + \sqrt{\|\bm{y}^{1} - \bar{\bm{y}}^{0}\|_2^2  + \frac{2D^2L_y^2}{\alpha\mu_g}\left(\frac{2}{\alpha\mu_g}-1\right)},
\eeq 
where $D$ is the diameter of ${\cal S}$ and $\widetilde{B}$, $L_y$ are given in \ref{statement:ystar}. 
\begin{Corollary}\label{coro:boundness-yk}
    Under the setting of Proposition \ref{prop:inequality-recursion1}, we have
    \begin{equation*}
    \begin{aligned}
        \|\bm{y}^k\|_2 &\leq B,~\forall~k \geq 0.
    \end{aligned}
\end{equation*} 
\end{Corollary}



The above corollary enables us to control the progress of the upper-level objective function $\ell(\cdot)$:
\begin{Prop}\label{prop:inequality-recursion2}
Under the setting of Proposition \ref{prop:inequality-recursion1}, there exist $L_{\widehat{\Phi}}$ and $L_{\ell}$ such that for any $k$,
    \begin{equation}\label{eq:decrease_ell}
        \begin{aligned}
            \ell (\GSO^{k+1})- \ell (\GSO^{k}) \leq& -\left(\frac{1}{2\gamma}-\frac{L_{\ell}}{2}\right)\|\GSO^{k+1}-\GSO^{k}\|_F^2 \\
            &+ \frac{\gamma}{2}L_{\widehat{\Phi}}^{2}\|\bm{y}^{k+1}-\bar{\bm{y}}^k\|_2^2.
        \end{aligned}
    \end{equation}
\end{Prop}
\begin{proof}
\ref{statement:grd-ell} implies that
\begin{equation}\label{inequ_ingredient_lip_gradient_l}
\begin{aligned}
    \ell (\GSO^{k+1})- \ell (\GSO^{k})  \leq &\langle\nabla \ell(\GSO^{k}),\GSO^{k+1}-\GSO^{k}\rangle \\
    & +(L_{\ell}/2)\|\GSO^{k+1}-\GSO^{k}\|_F^2.
\end{aligned}
\end{equation}
Since $\|\bm{y}^k\|_2\leq B$ and $\|\bm{y}^\star(\GSO)\|_2\leq\widetilde{B}$, \ref{statement:hat-phi} implies that there exists an $L_{\widehat{\Phi}}$ such that for any $k$,
\begin{align*}
    & \|\grd\ell(\GSO^k) - \widehat{\grd}\Phi(\GSO^k;\bm{y}^{k+1})\|_F \\
    & = \|\widehat{\grd}\Phi(\GSO^k;\bar{\bm{y}}^{k}) - \widehat{\grd}\Phi(\GSO^k;\bm{y}^{k+1})\|_F
    \leq L_{\widehat{\Phi}}\|\bar{\bm{y}}^{k} -\bm{y}^{k+1}\|_2.
\end{align*}
Subsequently, we can bound
\begin{equation*}
\begin{aligned}
    & \langle\nabla \ell(\GSO^{k}),\GSO^{k+1}-\GSO^{k}\rangle \\
    & = \langle\nabla \ell (\GSO^{k}) - \widehat{\grd}\Phi (\GSO^{k};\bm{y}^{k+1} ) + \widehat{\grd}\Phi (\GSO^{k};\bm{y}^{k+1} ) , \GSO^{k+1}-\GSO^{k}\rangle \\
    & \leq \langle \nabla \ell(\GSO^{k})- \widehat{\grd}\Phi (\GSO^{k};\bm{y}^{k+1} ), \GSO^{k+1}-\GSO^{k}\rangle \\
       &\quad~- \frac{1}{\gamma}\|\GSO^{k+1} - \GSO^k\|_F^2 \\
    & \leq \frac{\gamma}{2}\|\nabla \ell (\GSO^{k})-\widehat{\grd}\Phi (\GSO^{k};\bm{y}^{k+1} )\|_F^2 - \frac{1}{2\gamma}\|\GSO^{k+1} - \GSO^k\|_F^2 \\
    & \leq \frac{\gamma}{2}L_{\widehat{\Phi}}^2\|\bar{\bm{y}}^k - \bm{y}^{k+1}\|_2^2 - \frac{1}{2\gamma}\|\GSO^{k+1} - \GSO^k\|_F^2 ,
\end{aligned}
\end{equation*}
where the first inequality is from the update rule of $\GSO$ and the last inequality is from \ref{statement:hat-phi}. Finally, combining the above result with \eqref{inequ_ingredient_lip_gradient_l} gives the desired inequality \eqref{eq:decrease_ell}.
\end{proof}

To complete the proof, we borrow the following lemma from \cite[Lemma 3.6]{hong2023two}.
 \begin{Lemma}\label{lemma:inequality-general}
         Consider nonnegative sequences $\{\Omega_k\}, \{\Gamma_k\}$, $\{\Theta_k\}$. Let $c_0, c_1, d_0, d_1 > 0$ be such that
    \beq \label{eq:lemma-inequality-general}
    \begin{aligned}
        &\Omega_{k+1} \leq \Omega_{k} - c_0\Theta_{k+1} + c_1\Gamma_{k+1}, \\
        &\Gamma_{k+1} \leq (1-d_0)\Gamma_k + d_1\Theta_k.
    \end{aligned}
    \eeq
    If $\frac{c_0}{c_1} >\frac{d_1}{d_0}$, then
    \begin{align*} \textstyle
    \frac{1}{K}\sum_{k = 1}^K\Theta_k = {\cal O}\left( {1} / {K}\right), ~~\frac{1}{K}\sum_{k = 1}^K\Gamma_k = {\cal O}\left( {1} / {K}\right).
    \end{align*}
 \end{Lemma}
 \noindent Consider the following substitutions 
    $\Theta^k = \|\GSO^k - \GSO^{k-1}\|_F^2$, $\Gamma^k = \|\bm{y}^k - y^\star(\GSO^{k-1})\|_2^2$, $\Omega^k = \ell(\GSO^k)$. From Propositions~\ref{prop:inequality-recursion1} and \ref{prop:inequality-recursion2}, we observe that by setting $c_0 = \frac{1}{2\gamma} - \frac{L_{\ell}}{2}$, the inequalities in \eqref{eq:lemma-inequality-general} hold with $c_1 = \frac{\gamma}{2}L_{\widehat{\Phi}}^2$, $d_0 = \frac{\alpha(1-\mody)}{2}$, $d_1 = L_y^2(\frac{2}{\alpha(1-\mody)} - 1)$.

By setting $\alpha = \frac{1-\mody}{(1+\mody)^2}$ and 
    \[
    \gamma \leq \min\{ {3} / ({4L_{\ell}}), \alpha ({1-\mody}) / ({4L_{\widehat{\Phi}}L_y}) \},
    \]
we observe that 
\begin{equation*}
\begin{aligned}
    \frac{c_0}{c_1} &= \frac{1 - \gamma L_{\ell}}{\gamma^2 L_{\widehat{\Phi}}^2} \geq \frac{1}{4\gamma^2L_{\widehat{\Phi}}^2} \geq \frac{4L_y^2}{\mu_g^2\alpha^2} > L_y^2\frac{4 - 2\mu_g\alpha}{\alpha^2\mu_g^2} = \frac{d_1}{d_0}.
\end{aligned}
\end{equation*}
Applying Lemma \ref{lemma:inequality-general} gives
\begin{align}
& \textstyle \frac{1}{K}\sum_{k = 1}^K\|\GSO^k - \GSO^{k-1}\|_F^2 = \mathcal{O}(\frac{1}{K}) ,\label{upperbond_deltax}  \\
& \textstyle \frac{1}{K}\sum_{k = 1}^K\|\bm{y}^k - \bar{\bm{y}}^{k-1}\|_2^2 = \mathcal{O}(\frac{1}{K}) . \label{upperbond_deltay}
\end{align}    
We conclude the proof of Theorem \ref{thm:ttgd} by observing that
    \begin{equation*}
        \begin{aligned}
    & \gamma \, {\sf G}_{\gamma} (\GSO^k) = \| \GSO^k - {\sf Proj}_{\cal S} ( \GSO^k - \gamma \grd \ell( \GSO^k ) ) \|_F \\
    & \leq \|  \GSO^k -  {\sf Proj}_{ \cal S } ( \GSO^k - \gamma \widehat{\grd} \Phi( \GSO^k ; {\bm y}^{k+1} )) \|_F\\
    & \quad~+ \|{\sf Proj}_{ \cal S } ( \GSO^k - \gamma \widehat{\grd} \Phi( \GSO^k ; {\bm y}^{k+1} )) - {\sf Proj}_{\cal S} ( \GSO^k - \gamma \grd \ell( \GSO^k ) ) \|_F \\
    & \leq \|( \GSO^k -  \GSO^{k+1})\|_F + \gamma\| \widehat{\grd} \Phi( \GSO^k ; {\bm y}^{k+1} ) - \grd \ell( \GSO^k )\|_F \\
    & \leq \|( \GSO^k -  \GSO^{k+1})\|_F + \gamma L_{\widehat{\Phi}}\|{\bm y}^{k+1} - \bar{\bm y}^{k}\|_2.
        \end{aligned}
    \end{equation*}
    It follows from \eqref{upperbond_deltax} and \eqref{upperbond_deltay} that
    \begin{equation*}
        \begin{aligned}
        & \textstyle \frac{1}{K}\sum_{k = 1}^K {\sf G}_\gamma ( \GSO^k ) = \mathcal{O}( {1} / {K}). \\
        \end{aligned}
    \end{equation*}

\section{Proof of Proposition~\ref{prop:conv_gllq}} \label{app:proof_conv_prop}
Throughout this subsection, we assume that H\ref{assu:lipsf_lq} holds for the LQ game setting. Note that for LQ games, we have
\begin{equation}\label{eq:NG_Y}
    \begin{aligned}
            {\sf Y}^{lq}(\bm{y};\GSO) = \bm{y} - \max\{\bm{0}, \GSO f(\bm{y}) + \bm{b}\}.
    \end{aligned}
\end{equation}
We define
\begin{equation}\label{eq:NG_tildeY}
    \begin{aligned}
    \tilde{{\sf Y}}^{lq}(\bm{y};\GSO) = \bm{y} - (\GSO f(\bm{y}) + \bm{b})
    \end{aligned}
\end{equation}
and show that $\tilde{{\sf Y}}^{lq}(\bm{y};\GSO) = {\sf Y}^{lq}(\bm{y};\GSO)$ for any $\bm{y}\in{\cal Y}$ and $\GSO\in{\cal S}_{\sf ng}$. Formally, we have
\begin{Lemma}\label{lemma:trunc_Y}
    For LQ games, 
    \begin{align*}
        \tilde{{\sf Y}}^{lq}(\bm{y};\GSO) 
        = {\sf Y}^{lq}(\bm{y};\GSO)
    \end{align*}
    for any $\GSO\in{\cal S}_{\sf ng}$ and $\bm{y}\in[0,\infty)^N$.
\end{Lemma}
\begin{proof}
    In the LQ game setting, note that $f(x)\geq 0$ for any $x\in[0,\infty)$ and $\bm{b}\geq \bm{0}$ as indicated by H\ref{assu:lipsf_lq}. Hence, we have $\GSO f(\bm{y}) + \bm{b}\geq \bm{0}$ for any $\bm{y}\in{\cal Y}$ and $\GSO\in{\cal S}_{\sf ng}$. This proves that $\tilde{{\sf Y}}^{lq}(\bm{y};\GSO) = {\sf Y}^{lq}(\bm{y};\GSO)$ over the concerned region.
\end{proof}

\begin{Remark}
    The proof above shows that ${\sf F}(\bm{y};\GSO):=\bm{y} - {\sf Y}(\bm{y};\GSO)$ is a self-map of $\bm{y}$ over ${\cal Y}$ for any $\bm{S}\in{\cal S}_{\sf ng}$ under H\ref{assu:lipsf_lq}. Hence, H\ref{assu:self-F} holds with this assumption for LQ games.
\end{Remark}
Next, we show that $\tilde{{\sf Y}}^{lq}$, together with $\Phi(\cdot)$ defined in \eqref{eq:glfp_ng}, satisfy H\ref{assu:Y} to H\ref{assu:Y_grd}.

\begin{Lemma}
    If $\Phi( \GSO;{\bm y} ) = {\rm Tr}( \GSO^\top {\bm D} ) + \beta \| \GSO \|_F^2 - \lambda {\bf 1}^\top {\bm y}$ as defined in \eqref{eq:glfp_ng}, then H\ref{assu:phi_grd} holds with ${\cal S} = {\cal S}_{\sf ng}$ and ${\cal Y}\subseteq[0,\infty)^N$.
\end{Lemma}
\begin{proof}
    Note that $\grd\Phi(\GSO;\bm{y}) = (\bm{D} + 2\beta\GSO;-\lambda\bm{1})$. Hence, for any $\GSO_1, \GSO_2\in{\cal S}_{\sf ng}$ and $\bm{y}_1, \bm{y}_2\in{\cal Y}$,
    \begin{align*}
        &\|\grd\Phi(\GSO_1;\bm{y}_1) - \grd\Phi(\GSO_2;\bm{y}_2)\|_F = \|(2\beta(\GSO_1-\GSO_2);\bm{0})\|_F \\
        & = 2\beta\|\GSO_1-\GSO_2\|_F \leq 2\beta\|(\GSO_1;\bm{y}_1) - (\GSO_2;\bm{y}_2)\|_F.
    \end{align*}
    This proves the statement.
\end{proof}

\begin{Lemma}\label{lemma:assumYhold}
    $\tilde{{\sf Y}}^{lq}$ satisfies H\ref{assu:Y}.
\end{Lemma}
\begin{proof}
    For any $\GSO\in{\cal S}_{\sf ng}$ and $\bm{y}_1,\bm{y}_2\in{\cal Y}$, we note that
    \begin{align*}
        &\langle\tilde{{\sf Y}}^{lq}(\bm{y}_1;\GSO) - \tilde{{\sf Y}}^{lq}(\bm{y}_2;\GSO),\bm{y}_1 - \bm{y}_2\rangle \\
        & = \|\bm{y}_1 - \bm{y}_2\|_2^2 - \langle\bm{S}(f(\bm{y}_1) - f(\bm{y}_2)),\bm{y}_1 - \bm{y}_2\rangle \\
        & \geq  \|\bm{y}_1 - \bm{y}_2\|_2^2 - \|\GSO\|_2L_{f,1}\|\bm{y}_1 - \bm{y}_2\|_2^2\\
        & \geq (1-cL_{f,1})\|\bm{y}_1 - \bm{y}_2\|_2^2.
    \end{align*}
    The last inequality is due to $\|\GSO\|_2 \leq \|\GSO\|_{\infty} = c$ for any $\GSO\in{\cal S}_{\sf ng}$.
    For the second part of H\ref{assu:Y}, note that 
        \begin{align*}
        &\|\tilde{{\sf Y}}^{lq}(\bm{y}_1;\GSO) - \tilde{{\sf Y}}^{lq}(\bm{y}_2;\GSO)\|_2 \\
        & \leq \|\bm{y}_1 - \bm{y}_2\|_2 + \|\bm{S}(f(\bm{y}_1) - f(\bm{y}_2))\|_2 \\
        & \leq (1+cL_{f,1})\|\bm{y}_1 - \bm{y}_2\|_2.
    \end{align*}
    This concludes the proof.
\end{proof}

\begin{Lemma}
    $\tilde{{\sf Y}}^{lq}$ satisfies H\ref{assu:Y_grd}.
\end{Lemma}
\begin{proof}
    Note that
    \begin{align}
        {\sf J}\tilde{{\sf Y}}^{lq} &= \left[\bm{I} - \GSO{\sf Diag}(f'(\bm{y})); f(\bm{y})\otimes\bm{I}_N\right].
    \end{align}
    Hence, for any $\GSO_1,\GSO_2\in{\cal S}_{\sf ng}$ and $\bm{y}_1,\bm{y}_2\in{\cal Y}$, we have
    \begin{align*}
        &\|{\sf J}\tilde{{\sf Y}}^{lq}(\bm{y}_1;\GSO_1) - {\sf J}\tilde{{\sf Y}}^{lq}(\bm{y}_2;\GSO_2)\|_F \\
        & \leq \|\GSO_1{\sf Diag}(f'(\bm{y}_1)) - \GSO_2{\sf Diag}(f'(\bm{y}_2))\|_F \\
        &\quad~+ \|(f(\bm{y}_1) - f(\bm{y}_2))\otimes\bm{I}_N\|_F \\
        & \leq \|\GSO_1\|_2\|{\sf Diag}(f'(\bm{y}_1))-{\sf Diag}(f'(\bm{y}_2))\|_F \\
        &\quad~+ \|{\sf Diag}(f'(\bm{y}_1))\|_2\|\GSO_1 - \GSO_2\|_F + \|(f(\bm{y}_1) - f(\bm{y}_2))\otimes\bm{I}_N\|_F\\
        & \leq cL_{f,2}\|\bm{y}_1 - \bm{y}_2\|_2 + L_{f,1}\|\GSO_1 - \GSO_2\|_F + \sqrt{N}L_{f,1}\|\bm{y}_1 - \bm{y}_2\|_2 \\
        & \leq \sqrt{2}(cL_{f,2} + \sqrt{N}L_{f,1})\|(\GSO_1;\bm{y}_1) - (\GSO_2;\bm{y}_2)\|_F.
    \end{align*}
\end{proof}

For the proof of Proposition \ref{prop:conv_glrt}, recall that
\[{\sf Y}^{rt}(\bm{y};\GSO) = \bm{y} - \min\{\bm{a}, g(\GSO\bm{y}) + \bm{b}\}.\]
Similarly, we define 
\[\tilde{{\sf Y}}^{rt}(\bm{y};\GSO) = \bm{y} - ( g(\GSO\bm{y}) + \bm{b}).\]
Following the proof for Proposition \ref{prop:conv_gllq}, we can show that ${\sf Y}^{rt} = \tilde{{\sf Y}}^{rt}$ over the concerned region and $\tilde{{\sf Y}}^{rt}$ satisfies H\ref{assu:Y} to H\ref{assu:Y_grd} when H\ref{assu:lipsg_rt} holds for RT games.

\end{document}